\documentclass[final,3p,times,authoryear]{elsarticle}
\biboptions{square}

\usepackage[T1]{fontenc}
\usepackage{caption}
\usepackage{graphicx}
\usepackage{amsmath,amsfonts,amssymb}
\usepackage{array}
\usepackage{hyperref}
\usepackage{fancyhdr}
\usepackage{lastpage}
\usepackage{physics}
\usepackage{caption}
\usepackage{subcaption}
\usepackage[dvipsnames]{xcolor}
\usepackage{amsthm}
\usepackage{float}
\newtheorem{theorem}{Theorem}[section]

\theoremstyle{definition}

\newtheorem{prop}{Proposition}

\journal{Mathematical Biosciences}

\begin{document}

\begin{frontmatter}



\title{Turing Patterns in a Morphogenetic Model with Single Regulatory Function} %% Article title


%% \author[label1Sb,label2sb]{Saad Benjelloun}
%% \affiliation[label1sb]{organization={Leonardo DeVinci},
%%             city={La Defense},
%%             country={France}}
%%
%% \affiliation[label2sb]{organization={Makhbar Mathematical Research Institute},
%%             city={Casablanca},
%%             country={Morocco}}

\author[um6p]{Mohamed Amine Ouchdiri} %% Author name

\affiliation[um6p]{organization={UM6P},%Department and Organization
            addressline={College of Computing}, 
            city={Benguerir},
            %postcode={}, 
            %state={},
            country={Morocco}}

\author[label1sb,label2sb]{Saad Benjelloun}
\affiliation[label1sb]{organization={De Vinci Higher Education},
             addressline={De Vinci Research Center},
             %city={},
             state={Paris},
             country={France}}
\affiliation[label2sb]{organization={Makhbar Mathematical Research Institute},
             %addressline={},
             city={Casablanca},
             %postcode={},
             %state={},
             country={Morocco}}

\author[um6p]{Adnane Saoud} %% Author name

\author[esp]{Irene Otero-Muras} %% Author name
\affiliation[esp]{organization={Computational Synthetic Biology Group, Institute for Integrative Systems Biology (CSIC-UV)},%Department and Organization
            %addressline={}, 
            city={Valencia},
            postcode={46980}, 
            %state={},
            country={Spain}}
            
%% Abstract
\begin{abstract}
Confirming Turing's theory of morphogens in developmental processes is challenging, and synthetic biology has opened new avenues for testing Turing's predictions. Synthetic mammalian pattern formation has been recently achieved through a reaction-diffusion system based on the short-range activator (Nodal) and the long-range inhibitor  (Lefty) topology, where a single function regulates both morphogens.  In this paper, we investigate the emergence of Turing patterns in the synthetic Nodal-Lefty system. First, we prove the existence of a global solution and derive conditions for Turing instability through linear stability analysis. Subsequently, we examine the behavior of the system near the bifurcation threshold, employing weakly nonlinear analysis, and using multiple time scales, we derive the amplitude equations for supercritical and subcritical cases. The results demonstrate the ability of the system to support various patterns, with the subcritical Turing instability playing a crucial role in the formation of dissipative structures observed experimentally.
\end{abstract}


%% Keywords
\begin{keyword}


Reaction–diffusion system, Turing diffusion-driven instability, weakly non-linear analysis, morphogenesis.

\end{keyword}

\end{frontmatter}


%\tableofcontents
\section{Introduction}
\label{sec1}
%% Labels are used to cross-reference an item using \ref command.

%small introduction to turing patterns

%Mathematical modeling is a powerful tool for understanding the dynamics of complex biological systems. 

%Zebrafish embryogenesis, Xenopus embryos}.

Introduced by Alan Turing, reaction-diffusion theory of morphogenesis explains how two interacting chemical species, or morphogens, can transition from a stable homogeneous state to spatially organized patterns through diffusion-driven instability \cite{turing1990chemical}. Morphogens are essential signaling molecules for embryonic patterning, influencing tissue development through concentration-dependent cellular responses \cite{wolpert1969positional, Ulloa2007}. Despite the high relevance of self-organized reaction-diffusion (RD) patterns in biology, the conditions for pattern-enabling systems \emph{in vivo} are still not fully clarified and remain subject to active research \cite{Paul2024}.

To achieve pattern formation, Gierer and Meinhardt \cite{gierer1972theory} proposed the activator-inhibitor model for Turing pattern formation, which relies on local self-enhancement coupled with long-range inhibition. A prime example of this mechanism is found in the Nodal–Lefty signaling pathway, which is involved in mesodermal induction, axis formation, and left-right patterning during vertebrate development. 
Specifically, Nodal and Lefty are members of the transforming growth factor-$\beta$ (TGF-$\beta$) superfamily that play crucial roles in early vertebrate development \cite{schier2003nodal}. Nodal functions as an auto-activating morphogen that promotes its own expression while driving mesodermal specification and embryonic axis establishment \cite{robertson2004embryonic}. In contrast, Lefty restricts Nodal signaling spatially, creating the sharp concentration gradients essential for proper tissue patterning and left-right asymmetry \cite{nakamura2006generation}. Remarkably, Lefty diffuses approximately 29 times faster than Nodal, naturally establishing the spatial scale separation required for Turing instability and making this system an ideal experimental model for testing morphogenetic theory \cite{muller2012differential}.

Building on these natural properties, Sekine \emph{et al.} \cite{sekine2018synthetic} successfully engineered the first mammalian synthetic circuit capable of exhibiting Turing patterns. Their systematic investigation revealed distinct RD patterns driven by two different feedback mechanisms involving Lefty \cite{chen2004two}: competitive inhibition and a combination of competitive inhibition with direct inhibition. Interestingly, competitive inhibition alone produced maze-like Turing patterns (labyrinthine), while the combined feedback mechanism generated solitary spot patterns. Through linear analysis, the authors successfully identified the regions in which the Turing instability condition is met. 

Nevertheless, these results suggest that further investigation is needed beyond linear analysis, as linear stability theory does not always capture the full dynamics of pattern-forming systems, particularly those involving instabilities resulting from subcritical bifurcations \cite{Krause2024}. This limitation motivates the present study, where we employ weakly nonlinear analysis to provide a more complete understanding of pattern formation in the synthetic Nodal-Lefty system.

In this paper, we apply weakly nonlinear analysis (WNL) \cite{wollkind1994weakly} to identify the parameter regions for subcritical and supercritical bifurcation types in the synthetic Nodal-Lefty system. This approach builds on early work by Stuart \cite{stuart1960nonlinear} and multiple scales methods developed by Nayfeh \cite{nayfeh1973perturbation}, with key contributions from Newell, Whitehead, and Segel \cite{newell1969finite} who derived the amplitude equations for spatially modulated patterns near onset. Cross and Hohenberg \cite{cross1993pattern} provided the comprehensive review that established this framework for pattern formation studies. Since then, this method has been widely applied to pattern formation \cite{wollkind1994weakly, stephenson1995weakly, cruywagen1997biological, peng2016turing, han2017cross, gambino2019pattern, liu2011sequential}.

The WNL method distinguishes between two fundamental types of Turing instabilities. In supercritical bifurcations, patterns emerge gradually with small initial amplitudes that grow continuously in space. In subcritical bifurcations, the spatial structure appears abruptly with well-defined amplitudes, creating localized structures through discontinuous transitions. To capture these dynamics, we derive the Stuart-Landau equation in two dimensions, computing coefficients to third order for supercritical cases and fifth order for subcritical cases.

Our parameter space analysis reveals distinct bifurcation regimes. In the competitive inhibition model, we identify a large subcritical region that precedes the supercritical region. Conversely, the combined competitive and direct inhibition model shows a more restricted Turing region that predominantly favors subcritical behavior. Importantly, our analysis reveals two key findings: the solitary spot patterns observed experimentally by Sekine \emph{et al.} \cite{sekine2018synthetic} result directly from subcritical Turing instability, and our theoretical predictions demonstrate good agreement with computational results in both bifurcation regimes.

Other research has explored alternative Nodal signaling mechanisms. For instance, wave pinning studies in Xenopus embryos show how Lefty limits Nodal signaling to control cell polarization and embryonic symmetry breaking—processes distinct from Turing pattern formation \cite{middleton2013wave, nakamura2006generation}. While Turing patterns create spatial periodicity, these alternative symmetry-breaking mechanisms may serve different developmental functions, requiring further quantitative analysis to determine their relative importance.

The remainder of this paper is organized as follows: Section \ref{sec2} introduces the Nodal-Lefty reaction-diffusion model. Section \ref{sec3} demonstrates global solution existence. Section \ref{sec4} derives conditions for Turing diffusion instability. Section \ref{sec5} employs two-dimensional weakly nonlinear analysis to distinguish supercritical and subcritical regimes. Section \ref{sec6} presents numerical simulations with validation of theoretical predictions through supercritical and subcritical examples (Section \ref{sec6_1}), pattern expansion in the supercritical regime (Section \ref{sec6_2}), and localized spots in the subcritical regime (Section \ref{sec6.3}).
\label{sec2}

\begin{table*}[]
\small
\[
\begin{array}{lll}
\text { Parameter } & \text { Description } & \text { Reference value/range and Unit } \\
\hline \alpha_{\mathrm{N}} & \text { Maximum production rate of Nodal } & \text{Randomly selected and varied} \; \mathrm{nM} \mathrm{min}^{-1} \\
\alpha_{\mathrm{L}} & \text { Maximum production rate of Lefty } & \text{Randomly selected and varied}\;\mathrm{nM} \mathrm{min}^{-1} \\
K_{+} & \text {Association rate of Nodal and Lefty } & \text{Randomly selected and varied}\mathrm{~min}^{-1} \mathrm{nM}^{-1} \\
n_{\mathrm{N}} & \text { Hill coefficient of activation by Nodal } & 2.63 \text{ (dimensionless)}\\
n_{\mathrm{L}} & \text { Hill coefficient of inhibition by Lefty } & 1.09 \text{ (dimensionless)}\\
K_{\mathrm{N}} & \text { Dissociation coefficient of Nodal } & 9.28 \text{ nM} \\
K_{\mathrm{L}} & \text { Dissociation coefficient of Lefty } & 14.96 \text{ nM} \\
\gamma_{\mathrm{N}} & \text { Degradation rate of Nodal } & 2.37 \times 10^{-3} \text{ min}^{-1} \\
\gamma_{\mathrm{L}} & \text { Degradation rate of Lefty } & 5.65 \times 10^{-3} \text{ min}^{-1} \\
D_{\mathrm{N}} & \text { Diffusion coefficient of Nodal } & 1.96 \, \mu \text{m}^2 \mathrm{~min}^{-1} \\
D_{\mathrm{L}} & \text { Diffusion coefficient of Lefty } & 56.39 \text{ mm}^2 \mathrm{~min}^{-1} \\
\hline
\end{array}
\]
\caption{List of models parameters and their reference value. }
\label{tab:parameters}
\end{table*}
 %The mathematical model
\section{Reaction Diffusion system of Nodal and Lefty}
\label{sec2}
In order to understand the patterning mechanism of the activator-inhibitor circuit, Sekine et al.\cite{sekine2018synthetic} proposed the following reaction-diffusion system of Nodal and Lefty  that gives rise to various spatial patterns:
\begin{equation}
\begin{cases}
\begin{aligned}
\dfrac{\partial N(X, T)}{\partial T} &= \alpha_{\mathrm{N}} H(N,L)-\gamma_{\mathrm{N}} N-K_{+} N L\\
&\quad +D_{\mathrm{N}} \Delta N, 
\end{aligned}
& X \in \Omega, T>0 ,\\
\\
\begin{aligned}
\dfrac{\partial L(X, t)}{\partial T} &= \alpha_{\mathrm{L}} H(N,L)-\gamma_{\mathrm{L}} L-K_{+} N L\\
&\quad +D_{\mathrm{L}} \Delta L, 
\end{aligned}
& X \in \Omega, T>0 ,
\end{cases}
\label{N_L_dim}
\end{equation}
where $T$ is the time variable, $X \in \Omega$ is the spatial variable, and $\Omega \subset R^n(n=1,2)$ is the closed and bounded spatial domain that generates one-dimensional and two-dimensional spatial Turing patterns. The terms $N(X,T)$ and $L(X,T)$ describe the temporal evolution of the spatial concentrations of Nodal and Lefty protein ligands, respectively, where $D_N$ and $D_L$ are the diffusion coefficients for Nodal and Lefty, respectively, and $\Delta=\frac{\partial^2 }{\partial \boldsymbol{X}^2}+\frac{\partial^2 }{\partial \boldsymbol{Y}^2}$ is the Laplacian operator. The parameters $\alpha_{\mathrm{N}}$ and $\alpha_{\mathrm{L}}$ are the maximum production rates of Nodal and Lefty, respectively, and $H$ is the nonlinear regulatory function given as follows:
\begin{equation}
H(N,L)=\dfrac{N^{n_{\mathrm{N}}}}{N^{n_{\mathrm{N}}}+\left[K_{\mathrm{N}}\left\{1+\left(\frac{L}{K_{\mathrm{L}}}\right)^{n_{\mathrm{L}}}\right\}\right]^{n_{\mathrm{N}}}},
\end{equation}
the regulatory function $H$ determines how the production rates of Nodal and Lefty depend on their concentrations. Since the same promoters control both morphogens, a single regulatory function is justified for both, which is sigmoidal form and follows the Hill rule without the inhibitor Lefty.

Protein ligands require receptors to initiate signaling. Two different mechanisms have been postulated that can describe the negative feedback by Lefty. The first mechanism involves competition at receptor sites: Lefty inhibits Nodal by competing for the same receptors. In this interaction, the Hill coefficient of Nodal $n_N$, which describes the steepness of the response of its activation, similarly affects the interaction of Lefty, and imposing the Hill coefficient of Lefty $n_L \approx 1$ suggests that both morphogens are activated with comparable steepness of response. Thus, they are competing on receptors at these binding sites. The parameters $K_{\mathrm{N}}$ and $K_{\mathrm{L}}$ are the dissociation coefficients of Nodal and Lefty, respectively. 

The second mechanism is when Lefty restricts the signal of Nodal directly. Given the uncertainty of the decisive mechanism limiting Nodal signaling from \emph{in vivo} data. Sekine et al. modeled the two scenarios: one where the association rate of Nodal and Lefty $K_{+}$ is nonzero, the two modes of inhibition are active, and if $K_{+}$ is equal to zero, the inhibition of Lefty works only via competition on receptors. 

In the absence of the measurement of the association rate of Nodal and Lefty $K_{+}$, this parameter will be varied alongside the maximum production rate of Nodal and Lefty $(\alpha_{\mathrm{N}},\alpha_{\mathrm{L}})$. The degradation process is chosen to be linear, as described by the two degradation rates $\gamma_\mathrm{N}$ and $\gamma_\mathrm{L}$ for Nodal and Lefty, respectively.

The parameter values are derived from a combination of experimental measurements and theoretical considerations based on the synthetic mammalian system developed by Sekine et al. \cite{sekine2018synthetic}. Specifically: (1) \textbf{Diffusion coefficients} are calculated from experimentally measured characteristic diffusion lengths using the HiBiT visualization system, where Lefty diffuses approximately 29 times faster than Nodal; (2) \textbf{Hill coefficients} are directly obtained from experimental dose-response curves measuring receptor activation; (3) \textbf{Dissociation constants} represent experimentally determined binding affinities from biochemical assays; (4) \textbf{Degradation rates} are measured using cycloheximide treatment to block protein synthesis; (5) \textbf{Production rates} and \textbf{association rate} are systematically varied as control parameters to explore different dynamical regimes, with ranges chosen to encompass biologically relevant expression levels observed in mammalian cell culture systems. Table \ref{tab:parameters} presents the parameters and their values.

To facilitate the mathematical analysis, we have non-dimensionalized the model, thereby reducing the number of parameters. We introduce new units of space, time, and densities as follows:
$$
n=\frac{N}{K_{\mathrm{N}}}, \quad l=\frac{L}{K_{\mathrm{L}}}, \quad t= \frac{\alpha_N}{K_N} T, \quad x=\sqrt{\frac{\alpha_N D_N}{K_N} } X.
$$

Then, from the system \eqref{N_L_dim}, we obtain
\begin{equation}
\begin{cases} \frac{\partial n(\boldsymbol{x}, t)} {\partial t}=h(n,l)-\gamma_{\mathrm{n}} n-k_{+} n l+\Delta n & x \in \Omega, t>0, \\ 
\\
\frac{\partial l(\boldsymbol{x}, t)}{\partial t}=\tau h(n,l) -\gamma_{\mathrm{l}} l-\nu  k_{+}n l+d \Delta l, & x \in \Omega, t>0.\end{cases}
\label{N_L_Nondim}
\end{equation}
where $h$ is the non-dimensional form of the regulatory function given as follows:
\begin{equation}
h(n,l)=\dfrac{n^{n_n}}{n^{n_n}+\left(1+l^{n_l}\right)^{n_n}},
\end{equation}
and the parameters 
$$
\begin{aligned}
&\tau=\frac{K_N\alpha_L}{K_L\alpha_N}, \quad \gamma_n=\frac{K_N }{\alpha_N}\gamma_N, \quad \gamma_l=\frac{K_N }{\alpha_N}\gamma_L, \quad k_+ = \dfrac{K_N}{\alpha_N}K_+ K_L,\\ &\nu=\frac{K_N}{K_L}, \quad d=\dfrac{D_L}{D_N}.
\end{aligned}
$$

We also add to \eqref{N_L_Nondim} the following initial condition:
$$
n(x,0)=n_0(x), \quad  l(x, 0)=l_0(x), x \in \Omega
$$
and we use the Neumann boundary condition.
$$
\partial_\mu n=0, \quad \partial_\mu l=0, \quad \text { on }  \; \partial \Omega,
$$
where $\mu$ is outward unit normal to the boundary $\partial \Omega$.

\section{Existence of global solution}
\label{sec3}
%Section text. See Subsection %\ref{subsec1}.

This section examines the existence of a global solution for the system \eqref{N_L_Nondim}. First, we introduce the model, which includes the Neumann boundary and initial conditions. For the sake of simplicity, we only consider the model in one dimension.

Let $\Omega=[0, L]$ be a bounded domain with
boundary $\partial \Omega=\{0, L\}$. The system is written as follows:

\begin{equation}
\begin{cases}\frac{\partial n}{\partial t}=\dfrac{n^{n_n}}{n^{n_n}+\left(1+l^{n_l}\right) ^{n_n}}-\gamma_{\mathrm{n}} n-k_{+} n l+\frac{\partial^2 n}{\partial \boldsymbol{x}^2}, & x \in \Omega, t>0, \\ \frac{\partial l}{\partial t}=\tau \dfrac{n^{n_n}}{n^{n_n}+\left(1+l^{n_l}\right) ^{n_n}} -\gamma_{\mathrm{l}} l-\nu  k_{+}n l+d \frac{\partial^2 l}{\partial x^2}, & x \in \Omega, t>0, \\ \frac{\partial n}{\partial x}=\frac{\partial l}{\partial x}=0, & x \in \partial \Omega, t>0, \\ n(0, x)=n_0(x) \geq 0, \quad l(0, x)=l_0(x) \geq 0, & x \in \Omega.\end{cases}
\label{N_L_nondim_gen}
\end{equation}

For the unique local solution of the system \eqref{N_L_Nondim}, one can rely on the results \cite{amann1985global}. In this paper, we prove the global existence by demonstrating the boundedness of $\|n(t,.)\|_{L_{\infty}}$ and $\|l(t,.)\|_{L_{\infty}}$ for all time t, using the strong comparison theorem \cite{smoller2012shock}.
%First, we define  $D_T=(0, T] \times \Omega$ and $S_T=(0, T] \times \partial \Omega$. 

\begin{theorem}
The solution of the system of Nodal and Lefty \eqref{N_L_nondim_gen} with $n_0(x), l_0(x) \in C^2(\Omega)$ and $0\leq n_0(x)$, $0\leq l_0(x)$ is global and bounded for all $t \geq 0$.
\label{exs_of_global}
\end{theorem}

%\begin{theorem}

%Let $(n(x, t), l(x, t)) \in\left[C(D_T) \cap C^{2,1}(D_T)\right]^2$ be the solution of \eqref{N_L_nondim_gen} with
%$$
%0 \leq n_0(x) \leq \dfrac{1}{\gamma_n} \text { and } 0 \leq l_0(x) \leq \dfrac{\tau}{\gamma_l}.
%$$

%Then, $T=\infty$ and the solution is bounded for all $t \leq 0$.

%\end{theorem}
   \begin{proof}
Following Lemma 14.20 in \cite{smoller2012shock}, we first establish that solutions maintain non-negativity when starting from non-negative initial data $n_0(x), l_0(x)$. Next, we consider the initial value problem:
\begin{equation}
\left\{\begin{array}{l}
\frac{d w}{d t}=1 - \gamma_n w \quad \text { for } t>0, \\
w(0)=w_0=\max _{\Omega} n_0(x) ,
\end{array}\right.
\label{comp_1}
\end{equation}
then the solution of the ordinary differential equation \eqref{comp_1}is given as follows
\begin{equation}
w(t)=\frac{e^{-\gamma_n t}\left(w_0 \gamma_n+e^{\gamma_n t}-1\right)}{\gamma_n}.
\end{equation}

Since we have 
$$\frac{\partial n}{\partial t}-\frac{\partial^2 n}{\partial \boldsymbol{x}^2} \leq 1 - \gamma_n n, (x, t) \in \Omega \times[0, \infty),$$
then, using the Theorem 10.1 in \cite{smoller2012shock} we have  
$$
n(x, t) \leq w(t) \leq \max \left\{ \max _{\Omega} n_0(x) , \dfrac{1}{\gamma_n}\right\} \triangleq M_1, \quad t>0.
$$

Similarly, we show the boundedness of $l$ by using the following initial value problem
\begin{equation}
\left\{\begin{array}{l}
\frac{d z}{d t}=\tau - \gamma_l z \quad \text { for } t>0, \\
z(0)=z_0=\max _{\Omega} l_0(x),
\end{array}\right.
\label{comp_2}
\end{equation}
the solution of the ODE \eqref{comp_2} is given by
\begin{equation}
z(t)=\frac{e^{-\gamma_l t}\left(z_0 \gamma_l+\tau(e^{\gamma_l t}-1)\right)}{\gamma_l}.
\end{equation}

We conclude from Theorem 10.1 in \cite{smoller2012shock} that 
$$
l(x, t) \leq z(t) \leq \max \left\{ \max _{\Omega} l_0(x) , \dfrac{\tau}{\gamma_l}\right\} \triangleq M_2, \quad t>0.
$$

Thus, $\|n(t, .)\|_{L_{\infty}}$, $\|l(t, .)\|_{L_{\infty}}$ are bounded for all $t$ and the solution exists globally.
\end{proof}

%\begin{remark}
%The global existence of solutions in reaction diffusion system can %also be understood by through the mass control property. This %property means that the sum of the two components remains %controlled in $L_1$ norm at every time $t$ throughout the interval %of existence. Alongside this, additional hypotheses on the %nonlinearities of the system can ensure that the solution is %uniformly bounded in the $L_\infty$ suggests global existence. For %further details, one can consult the survey by Michel Pierre %\cite{pierre2010global}. Notably, the nodal and Lefty %morphogenetic system can align with these properties. 
%\end{remark}

\begin{figure*}[]
      \centering
	   \begin{subfigure}{0.45\linewidth}
		\includegraphics[width=\linewidth]{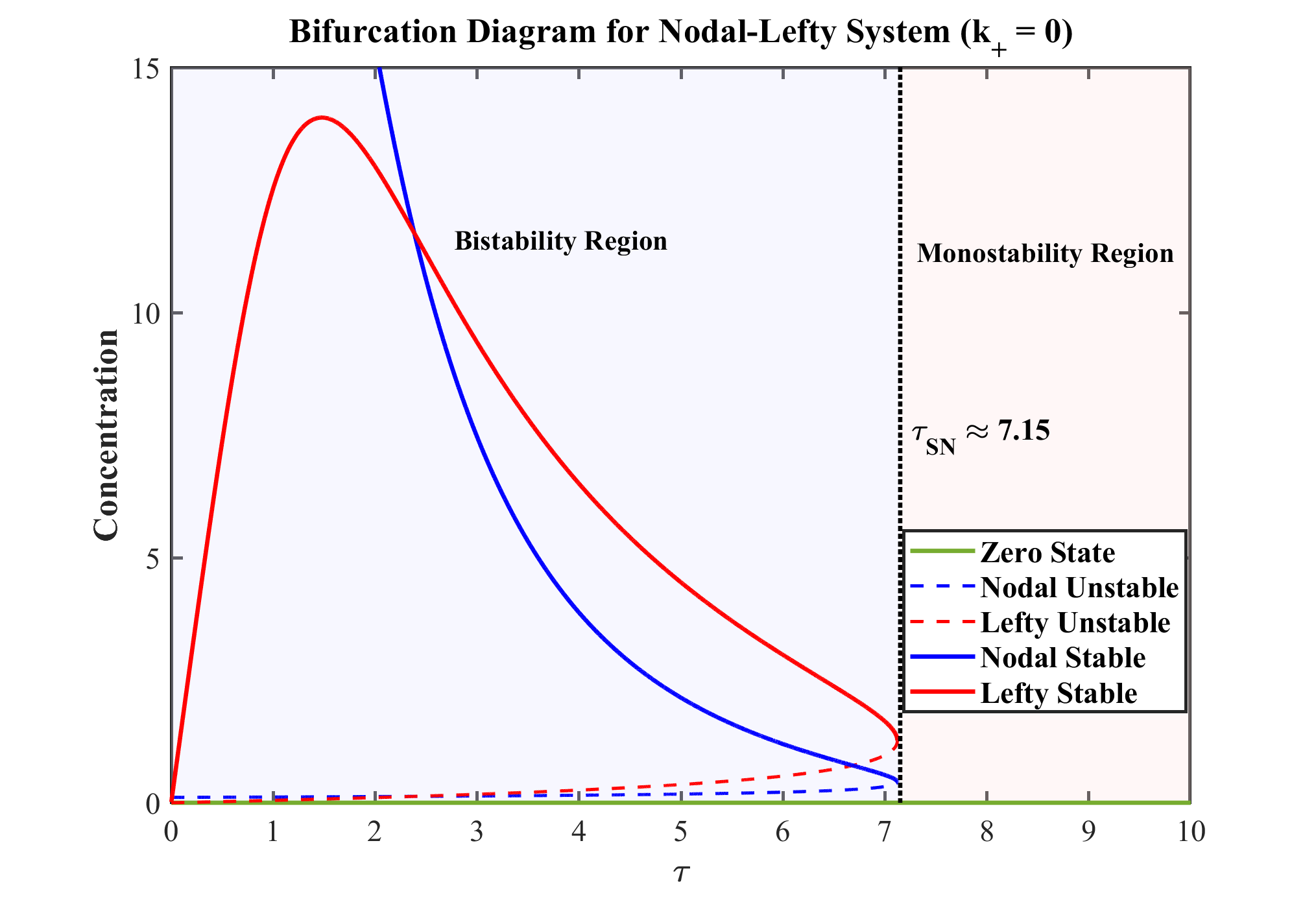}
		\caption{}
		\label{fig:subfig01}
	   \end{subfigure}
	   \begin{subfigure}{0.45\linewidth}
		\includegraphics[width=\linewidth]{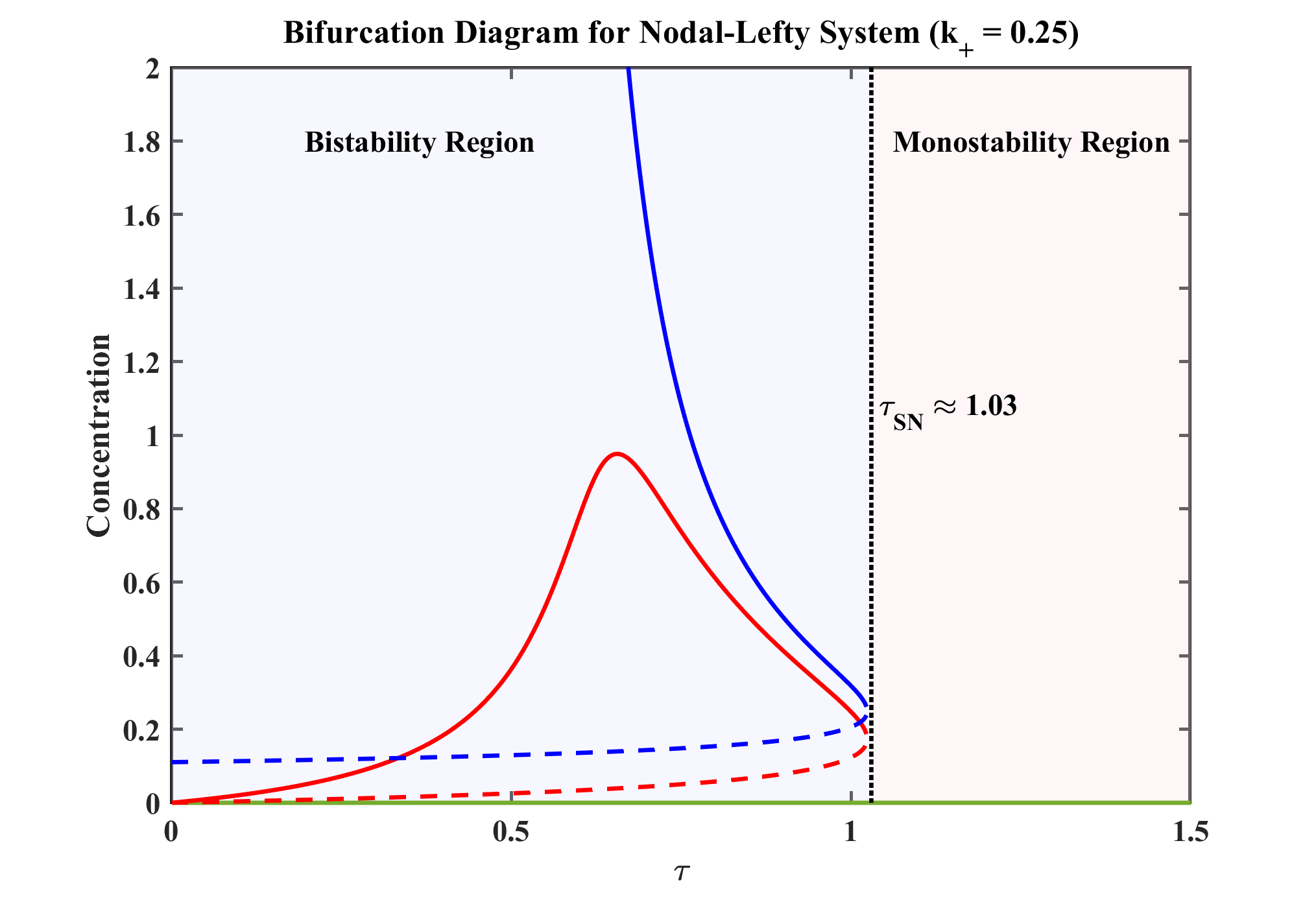}
		\caption{}
		\label{fig:subfig02}
	    \end{subfigure}
      \vfill
	   \begin{subfigure}{0.45\linewidth}
		\includegraphics[width=\linewidth]{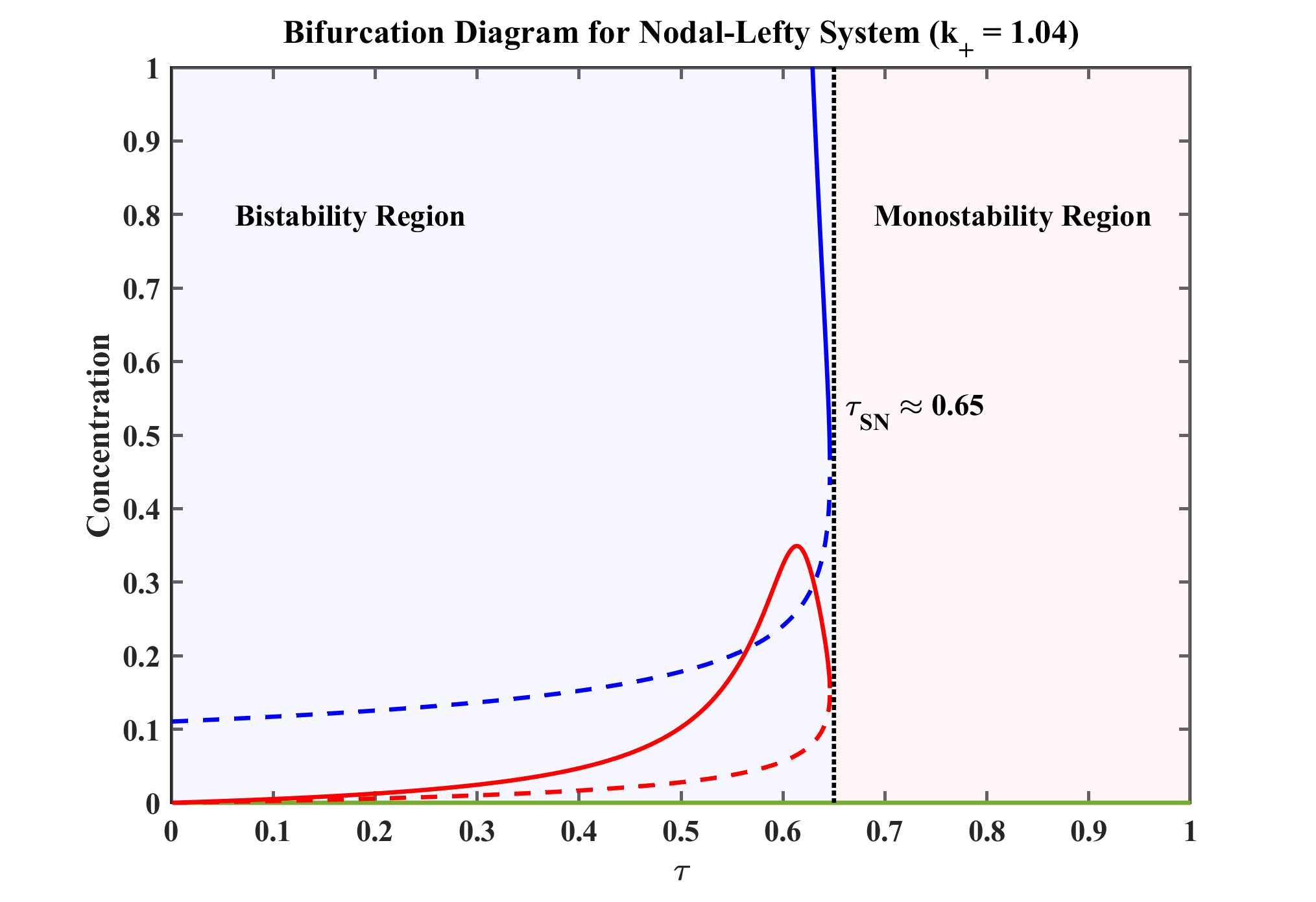}
		\caption{}
		\label{fig:subfig1}
	   \end{subfigure}
	   \begin{subfigure}{0.44\linewidth}
		\includegraphics[width=\linewidth]{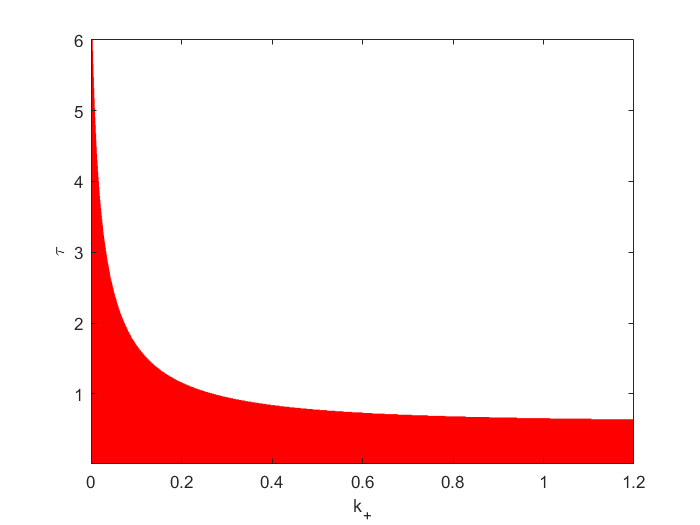}
		\caption{}
		\label{fig:subfig2}
	    \end{subfigure}
	\vfill
	     \begin{subfigure}{0.45\linewidth}
		 \includegraphics[width=\linewidth]{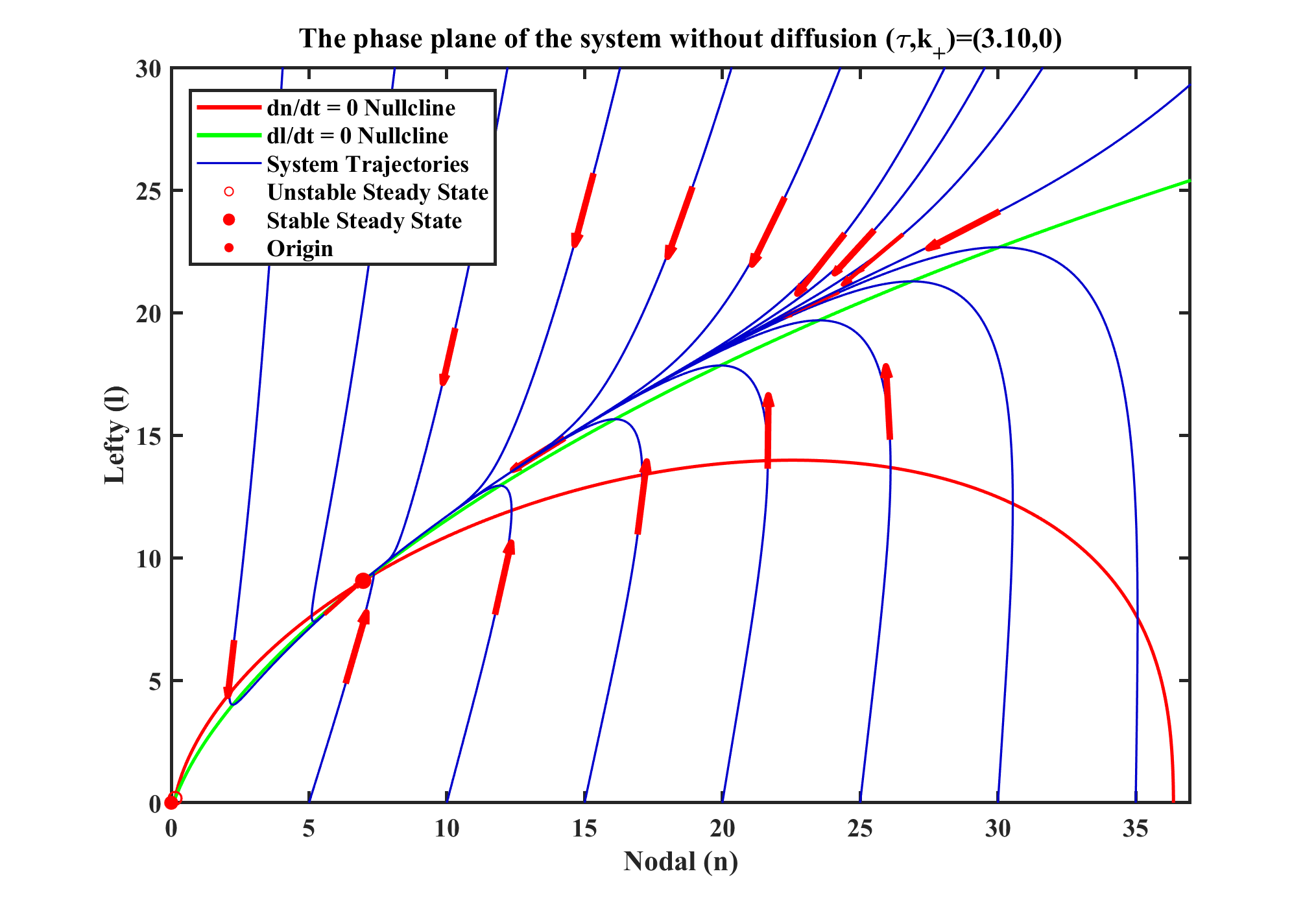}
		 \caption{}
		 \label{fig:subfig3}
	      \end{subfigure}
	       \begin{subfigure}{0.45\linewidth}
		  \includegraphics[width=\linewidth]{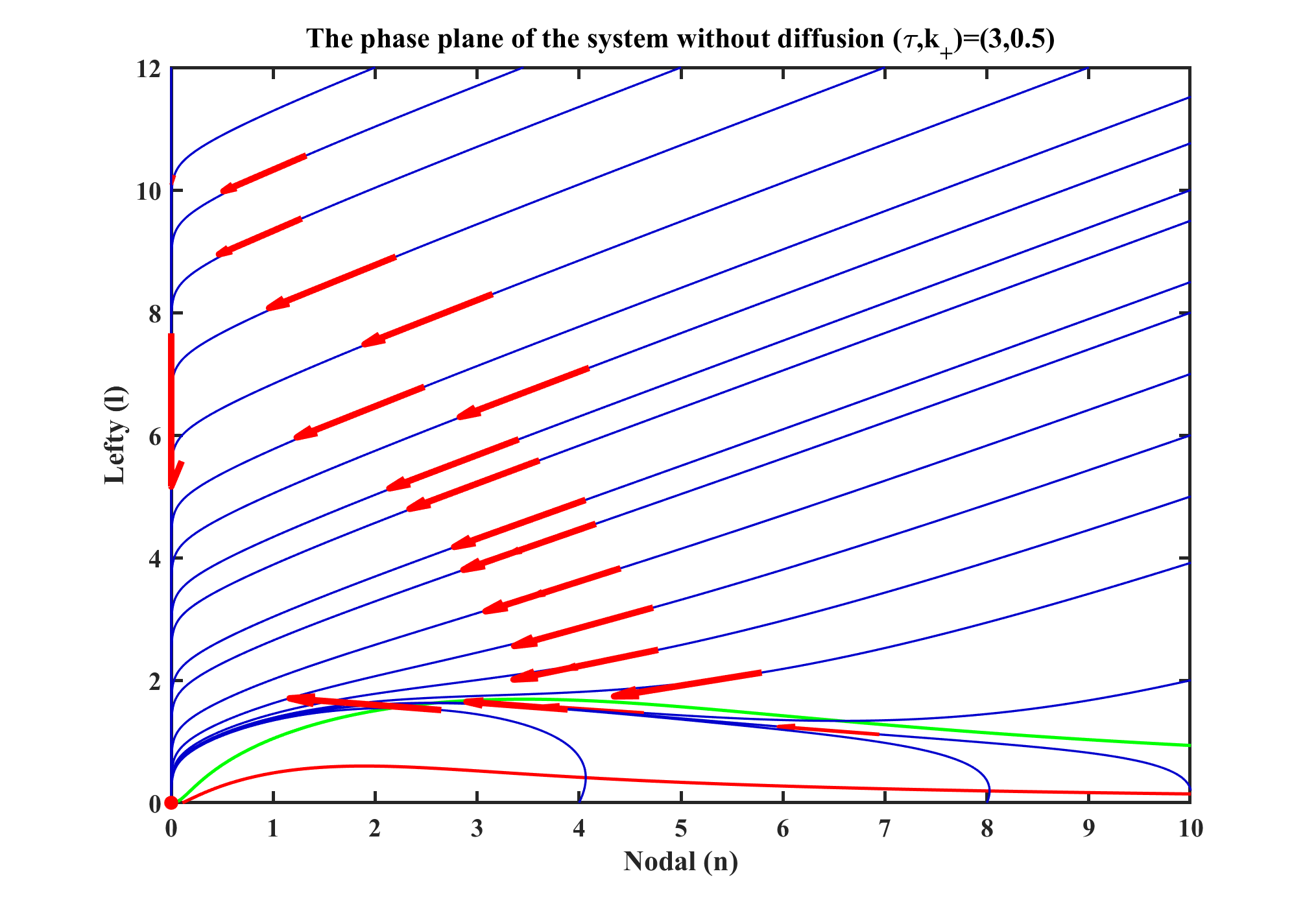}
		  \caption{}
		  \label{fig:subfig4}
	       \end{subfigure}
	\caption{Three bifurcation diagrams along the $\tau$ direction: the stable branches are depicted with solid lines, and the unstable branch with a dashed line. (a) For $k_+=0$ at a critical value  $\tau_{SN}\approx 7.12$, a saddle-node bifurcation occurs, allowing for the coexistence of two branches along with the stable branch of the trivial steady state for $\tau \leq\tau_{SN}$. (b) For $k_+=0.25$, the critical value of $\tau$ is $\tau_{SN}\approx1.05$. (c) For $k_+=1.04$, the critical value of $\tau$ is $\tau_{SN}\approx0.65$. (d) The stability region in the plane $(k_+, \tau)$, the region of bistability is denoted by the red area, while the remaining parameters represent the region of monostability. This monostability is distinguished by a single stable steady state $E_0 = (0, 0)$. (e) The phase plane of the system without diffusion, for specific values of $(\tau,k_{+})=(3.10,0)$, illustrates bistability with two stable nodes, $E_0=(0,0)$ and $E_2=(6.95,9.05)$ (indicated by red points), and a saddle point $E_1=(0.13,0.18)$ (marked by a red circle). (f) For a different set of $(\tau,k_{+})=(3,0.5)$ values, the phase plane depicts monostability, where the nullclines of Nodal and Lefty are illustrated by red and green curves, respectively. While the other parameters remain fixed as follows: $n_n=2.63$, $n_l=1.09$, $\gamma_n=0.02$, and $\gamma_l=0.06$.}
	\label{fig:subfigures4}
\end{figure*}

%% Use \subsection commands to start a subsection.

\section{Diffusion driven instability}
\label{sec4}

In this section, we investigate the diffusion-driven instability of system \eqref{N_L_Nondim}, starting with identifying the system's stability region without diffusion, a foundational step for the potential emergence of Turing patterns. For this reason, we introduce the following ODE system
\begin{equation}
\begin{aligned}
& \frac{\partial n}{\partial t}= h(n,l)-\gamma_{\mathrm{n}} n-k_{+} n l,\\
& \frac{\partial l}{\partial t}= \tau h(n,l)-\gamma_{\mathrm{l}} l-\nu k_{+} n l.
\end{aligned}
\label{ode_nodal_lefty}
\end{equation}

Given that we are working within a biological context, our interest is identifying the real positive steady states. In this system, finding the algebraic expression of the steady states poses a significant challenge, particularly because we have set the Hill coefficient $n_n$ to $2.63$ as proposed in \cite{sekine2018synthetic}.
First, the system admits a trivial steady state $E_0=(0,0)$ that always exists; the existence of the nontrivial positive steady states $(n^*,l^*)$ is the solution of the following two equations.
\begin{equation}
\begin{aligned}
h\left(n^*, l^*\right)  & =\gamma_n n^*+k_{+} n^* l^*, \\
\tau h\left(n^*, l^*\right) & =\gamma_l l^*+\nu k_{+} n^* l^*,
\end{aligned}
\label{steady_state_cond1}
\end{equation}
substituting the expression $h(n^*,l^*)$ from the first equation in the second equation \eqref{steady_state_cond1}, we can express the $l^*$ in function of $n^*$, $\tau$, and $k_+$. This yields to 
\begin{equation}
l^*=\frac{\tau \gamma_n n^*}{\gamma_l+(\nu-\tau) k_{+} n^*}.
\end{equation}

Now substitute this expression of $l^*$ in the first equation of the steady state \eqref{steady_state_cond1}. After simplification, we can deduce an algebraic condition that governs the existence of the steady states. We write this condition as follows:
\begin{equation}
\begin{split}
\frac{n^{* n_n}} {n^{* n_n}+\left(1+\left(\frac{\tau \gamma_n n^*}{\gamma_l+(\nu-\tau) k_{+} n^*}\right)^{n_l}\right) ^{n_n}} &= \gamma_n n^* + k_{+} n^*\left(\frac{\tau \gamma_n n^*}{\gamma_l+(\nu-\tau) k_{+} n^*}\right).
\end{split}
\label{steady_state_cond2}
\end{equation}

A simplified case can be illustrated by setting the parameters $n_n=2$, $n_l=1$, and $k_+=0$. Then we have $l^*=\frac{\tau \gamma_n n^*}{\gamma_l}$, and we can drive from the condition \eqref{steady_state_cond2} a quadratic equation for $n^*$ as 
\begin{equation}
\tau^2 \gamma_n^3 n^{* 2}+n^*\left (\gamma_l^2\left (1-\gamma_n\right) -2 \tau \gamma_n^2 \gamma_l\right) -\gamma_n \gamma_l^2=0.
\label{quadratic}
\end{equation}

The condition for the discriminant $\Delta$ to be positive for \eqref{quadratic}, and therefore the existence of two real solutions is
\begin{equation}
\tau<\tau_{SN}=\frac{\gamma_l\left(1-\gamma_n\right)}{4 \gamma_n^2},
\end{equation}
This inequality depends on $\tau$, and below a critical value of $\tau_{SN}$, a saddle-node bifurcation can occur. Assuming this condition is verified, we can find the two real solutions as follows:
\begin{align}
\begin{cases}
n_1^* &= \dfrac{-\left(\gamma_l^2\left(1-\gamma_n\right) -2 \tau \gamma_n^2 \gamma_l\right)}{2 \tau^2 \gamma_n^3}-\dfrac{\sqrt{\gamma_l^4\left(1-\gamma_n\right) ^2-4 \tau \gamma_n^2 \gamma_l^3\left(1-\gamma_n\right)}}{2 \tau^2 \gamma_n^3}, \\
n_2^* &= \dfrac{-\left(\gamma_l^2\left(1-\gamma_n\right) -2 \tau \gamma_n^2 \gamma_l\right)}{2 \tau^2 \gamma_n^3} +\dfrac{\sqrt{\gamma_l^4\left(1-\gamma_n\right) ^2-4 \tau \gamma_n^2 \gamma_l^3\left(1-\gamma_n\right)}}{2 \tau^2 \gamma_n^3}.
\end{cases}
\label{real_sol}
\end{align}
%Increasing the cooperativity of the regulatory function $h$ by setting a value of $n_n>1$ plays an important role in expanding the Turing space, as demonstrated in similar morphogenetic model in \cite{diambra2015cooperativity}. 
%

%
In the general case, the Hill exponent is not an integer number \cite{gutierrez2012cooperative}, we keep $n_n=2.63$ and $n_l=1.09$, and the other parameters of the system as introduced in Table \ref{tab:parameters}. We observe numerically that below a critical value of $\tau=\tau_c$, two real and positive fixed points $E_{1}$ and $E_{2}$ emerge as a result of the nullclines of Nodal and Lefty intersecting. Notably, the fixed point $E_2$ is stable equilibrium, alongside the trivial steady state $E_0=(0,0)$, which is always stable; this leads to bistability across a wide range of parameters. Meanwhile, $E_1$ remains unstable or acts as a saddle point. In the rest of the paper, we denote $E_2$ by $E_2=(n^*,l^*)$. Therefore, we generate three bifurcation diagrams plotted in $\tau$ direction for different values of $k_+$ in figures \ref{fig:subfig01},\ref{fig:subfig02} and \ref{fig:subfig1}, revealing a saddle-node bifurcation at a critical value $\tau_c$. The region of bistability is illustrated numerically in the plane $(k_{+},\tau)$ figure \ref{fig:subfig2}. An example of bistability, showing two stable nodes separated by a saddle point, is presented in figure \ref{fig:subfig3}. The existence of $E_2$ is not only biologically significant but also crucial in our case for providing the possibility of developing Turin patterns. Once outside the bistability region, we are only left with the trivial steady state $E_0$ as a stable node shown in figure \ref{fig:subfig4}, which maintains its stability even with the addition of diffusion. Consequently, the formation of Turing patterns cannot be expected in such a scenario. 

Throughout our analysis, we assume the existence of the stable equilibrium state $E_2$. To study spatial pattern formation, we follow Turing's approach of examining how this stable state can become unstable through diffusion effects. We start by linearizing system \eqref{N_L_Nondim} around the positive steady state $E_2=(n^*, l^*)$:
\begin{align}
\dot{\mathbf{w}} &= \left.J\right|_{\left(n^*, l^*\right)}\mathbf{w}+D^{d} \nabla^2 \mathbf{w}, \; \text{where} \;\mathbf{w} \equiv \left(\begin{array}{c}
n-n^* \\
l-l^*
\end{array}\right)\; \text{and} \;
D^d = \left(\begin{array}{cc}
1 & 0 \\
0 & d
\end{array}\right),
\label{L}
\end{align}
and $\left. J\right|_{\left(n^*, l^*\right)}$ the Jacobian matrix evaluated at the point $E_2$ is given as follows:
$$
\left.J\right|_{\left(n^*, l^*\right)}=J=\left(\begin{array}{cc}
h_n-\gamma_n -k_+ l^*  & h_l-k_+ n^* \\
\\
\\
\tau h_n -\nu k_+ l^* &  \tau h_l-\gamma_l -\nu k_+ n^*
\end{array}\right) ,
$$
$h_l$ and $h_n$ are the partial derivatives of the regulatory function $h$ depending on Nodal and Lefty evaluated in $(n^{*},l^{*})$.

%The aim of Turing approach is actually inducing instability in an initially stable system through the interaction of reaction and diffusion processes. Thus, we are 
We examine stability by considering perturbations $\delta e^{i k x+\lambda t}$, where $k$ gives the wavenumber, $\lambda$ describes temporal evolution through its real part, and $\delta$ indicates the perturbation magnitude. The perturbation exists when $\operatorname{det}(\lambda I_d-J+k^2D^d)=0$, yielding the dispersion relation $\lambda=\lambda(k)$. This relation appears as:
\begin{equation}
\lambda^2+t\left(k^2\right) \lambda+g\left(k^2\right)=0,
\label{lambda_char}
\end{equation}
where
\begin{equation}
t\left(k^2\right)=k^2(1+d)-\operatorname{tr} J,
\end{equation}
and
 \begin{equation}
 g\left(k^2\right)=k^4 d-k^2\left( J_{22}+ d J_{11}\right)+\operatorname{det} J .
\end{equation}

For steady state $E_2$ to be stable, we require $\operatorname{tr}(J) < 0$ and $\operatorname{det}(J) > 0$. For the diffusion-driven instability to occur, we need $\operatorname{Re}(\lambda)>0$ for some $k\neq 0$, which is equivalent to $g(k^2)<0$ for some $k\neq 0$; this requires $dJ_{11}+ J_{22}>0$. The critical condition ensuring diffusivity-driven instability occurs when $g(k^2)=0$ has two real solutions $k_{1,2}$, which is equivalent to imposing $\left( dJ_{11}+J_{22}\right)^2-4  d \operatorname{det} J>0$. Figure \ref{turing_space} illustrates the region in the $(k_+,\tau)$ parameter space where diffusion-driven instability and the formation of spatial patterns occur for a fixed set of the other system parameters. The condition for marginal stability at some $k = k_c$ is:
\begin{equation}
\min \left(g\left(k_c^2\right)\right)=0,
\end{equation}
and the minimum of $g$ is attained when:
\begin{equation}
k_c^2=\frac{d_cJ_{11}+J_{22}}{2 d_c},
\label{wavenumber}\end{equation}
Setting $\operatorname{det(J)}=\frac{\left( d_cJ_{11}+J_{22}\right)^2}{4 d_c}$ and using the boundary condition $\left( dJ_{11}+ J_{22}\right)^2-4 d \operatorname{det} J=0$, we obtain the critical value for the diffusion parameter
\begin{equation}
d_c=\frac{J_{11} J_{22}-2 J_{12} J_{21} \pm \sqrt{\left(J_{11} J_{22}-2 J_{12} J_{21}\right)^2-J_{11}^2 J_{22}^2}}{J_{11}^2}.
\end{equation}

When $d>d_c$, our system exhibits a finite $k$ stationary instability that generates patterns, characterized by $g(k^2)<0$ and $Re(\lambda(k^2))>0$. The critical wavenumber lies between the roots $k_1^2$ and $k_2^2$ of $g(k^2)$. Consequently, the equilibrium state $E_2$ becomes unstable, leading to the following theorem.
\begin{theorem}
Assuming the existence of the stable positive equilibrium $E_2 = (n^*, l^*)$ of system \eqref{N_L_Nondim}, then $E_2$ is spatially unstable when $d > d_c$.

\label{th3}
\end{theorem}
%The sign of the real part of $\lambda$ indicates whether the solution is growing or decreasing over time. However, if the real part of $\lambda$ is positive for some values of k, some of the components of the superposition will grow exponentially over time, in this case the diffusion driven instability occurs. 

Our previous analysis of diffusion-driven instability did not account for domain geometry and boundary conditions. We examine these effects to determine which patterns can emerge within the geometric constraints. In a finite domain, having $d>d_c$ alone does not guarantee pattern formation since the critical wavenumber $k_c$ might not match any admissible mode for the given domain and boundary conditions. However, when $d>d_c$, we find a range $\left(k_1^2, k_2^2\right)$ of wavenumbers where $g\left(k^2\right)<0$ and $\operatorname{Re}(\lambda)>0$. The boundaries of this interval occur at $k_1^2$ and $k_2^2$ where $g\left(k^2\right)=0$, as shown in figure \ref{lambda}. This interval must be large enough to contain at least one mode admissible by the Neumann boundary conditions. 
We focus on pattern formation in two dimensions, particularly in a rectangular domain $\Omega=\left[0, L_x\right] \times\left[0, L_y\right]$. The solutions to the linear system \eqref{L} with Neumann boundary conditions are: 
\begin{equation}
\mathbf{w}=\sum_{s_1, z_1 \in \mathbb{N}} \mathbf{f}_{s_1 z_1} e^{\lambda (k_{s_1 z_1}^2) t} \cos (\frac{s_1 \pi}{L_x} x) \cos(\dfrac{z_1 \pi}{L_y} y),
\end{equation}
\begin{equation}
 k_{s_1 z_1}^2=\left(\dfrac{s_1 \pi}{L_x}\right)^2+\left(\dfrac{z_1 \pi}{L_y}\right)^2 ,
 \end{equation}
where $\mathbf{f}_{s_1, z_1}$ denotes the Fourier coefficients corresponding to the initial conditions, while the values $\lambda(k_{s_1 z_1}^2)$ are obtained from the dispersion relation \eqref{lambda_char}. Pattern formation occurs when mode pairs satisfy

\begin{equation}
\begin{aligned}
 & k_1^2<k^2 \equiv \Phi_1^2+\Psi_1^2<k_2^2, \quad \text { where } \Phi_1 \equiv \frac{s_1 \pi}{L_x}, \Psi_1 \equiv \frac{z_1 \pi}{L_y}, \\ 
 & \text { and} \; \lambda(k^2)>0 ,
 \end{aligned}
 \label{wave}
\end{equation}
this condition occurs when $d > d_c$. We will focus on the scenario where only a single unstable eigenvalue satisfies the Neumann boundary conditions and lies between $k_1$ and $k_2$, as specified in Equation \eqref{wave}. We denote this eigenvalue as $\bar{k}_c$, differentiating it from the critical value $k_c$.

For a specific $\bar{k}_c$ such that $k_1 \leq \bar{k}_c \leq k_2$, the rectangular domain may support one, two, or several integer pairs $(m, n)$ that satisfy:
\begin{equation}
\bar{k}_c^2=\Phi_1^2+\Psi_1^2=\left(\frac{s_1 \pi}{L_x}\right)^2+\left(\frac{z_1 \pi}{L_y}\right)^2,
\label{wave2}
\end{equation}
Each case corresponds to an eigenvalue $\lambda$ of single, double, or higher multiplicity. The domain dimensions $L_x$ and $L_y$ determine which case occurs and what patterns emerge. Throughout this work, we limit our study to the single multiplicity case.

\begin{figure*}
  \centering
  \subfloat[]{\includegraphics[width=0.46\textwidth]{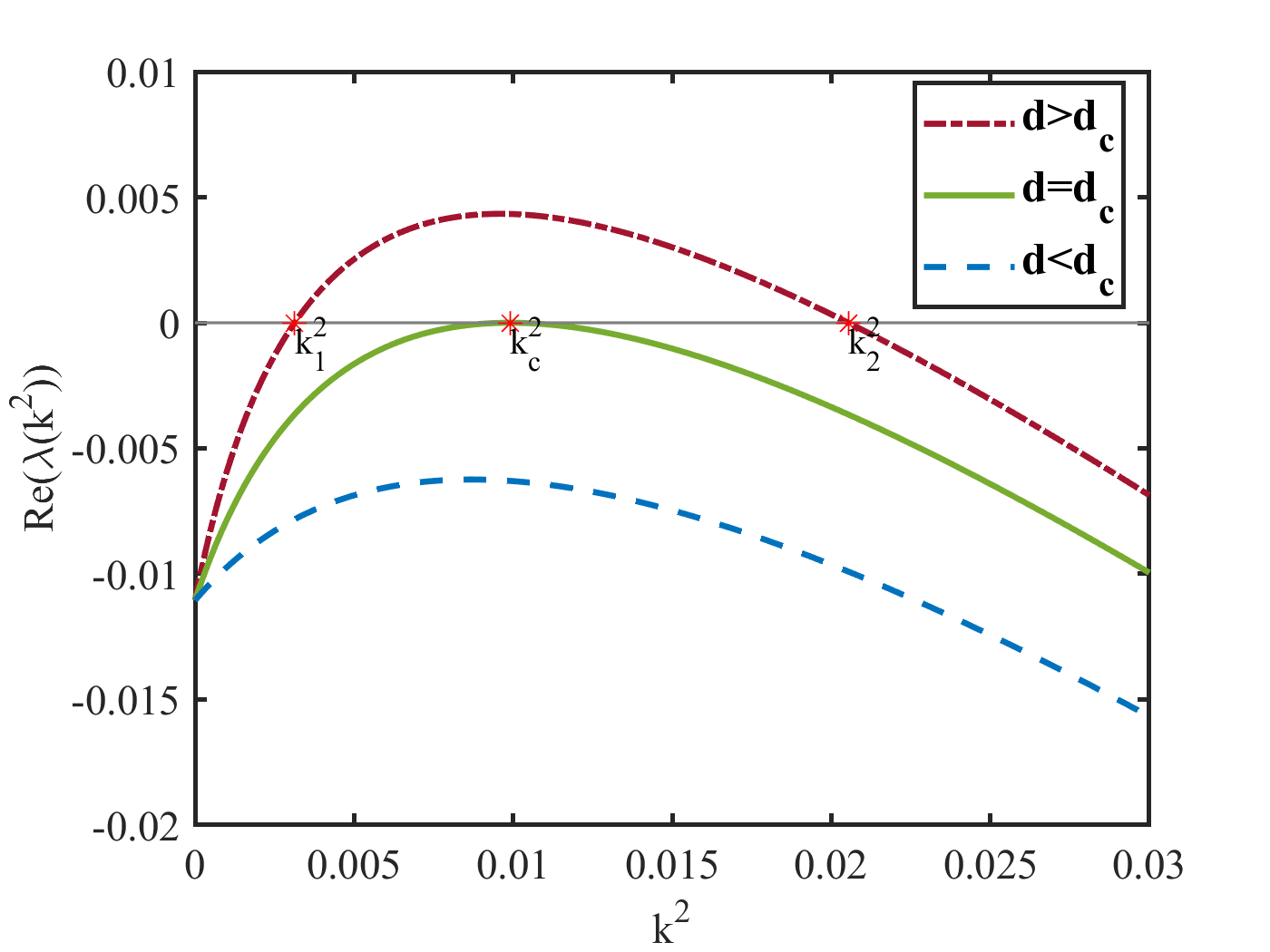}\label{lambda}}
  \hspace{15mm}
  \subfloat[]{\includegraphics[width=0.435\textwidth]{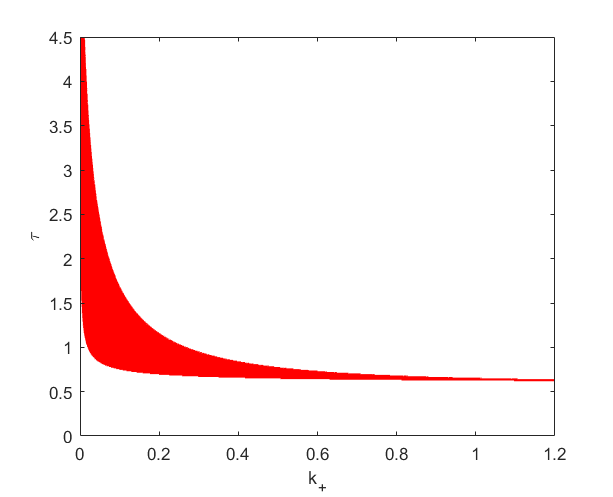}\label{turing_space}}
  \caption{(a) The plot of the real part of the growth rate of the k-th mode for different values of $d$. A band of growing mode is present for $d=28>d_c=18.34$. (b) Turing space in the $(\tau,k_{+})$ plane (the red region). The parameters are chosen as $\tau=3.10$, $k_{+}=0$, $n_n=2.63$, $n_l=1.09$, $\gamma_n=0.02$, $\gamma_l=0.06$ and $\nu=0.62$. }
\end{figure*}
%Section text. See Subsection \ref{subsec1}.
\begin{figure*}[t]
     \centering
     \begin{subfigure}[b]{0.45\textwidth}
         \centering
         \includegraphics[width=\textwidth]{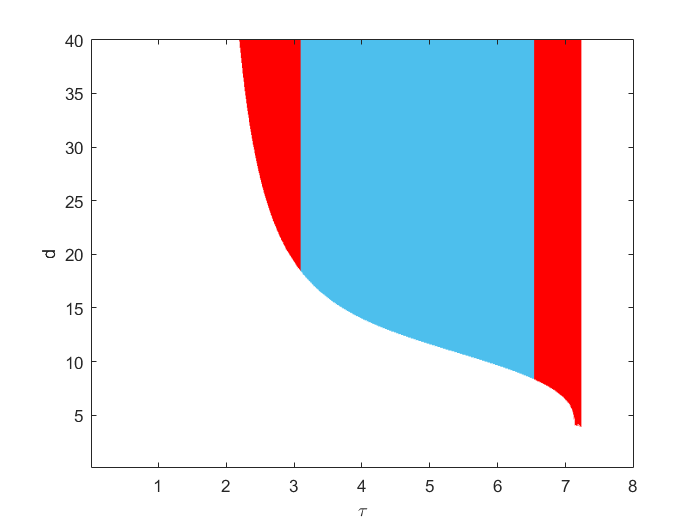}
         \caption{}
         \label{fig0:subfig1}
     \end{subfigure}
     \hfill
     \begin{subfigure}[b]{0.45\textwidth}
         \centering
         \includegraphics[width=\textwidth]{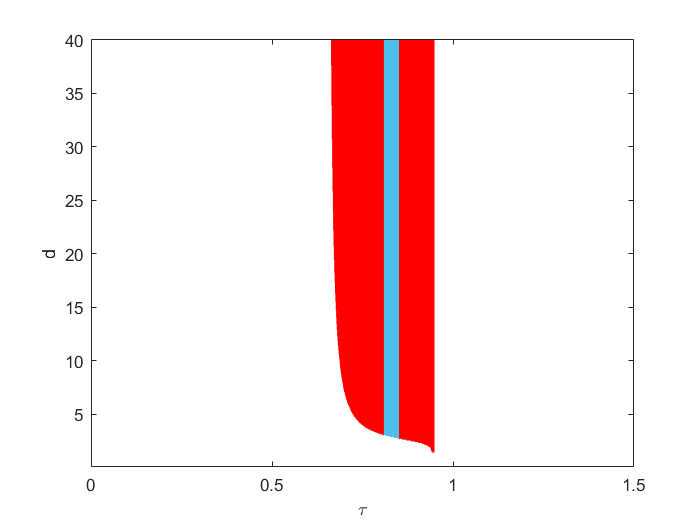}
         \caption{}
         \label{fig0:subfig2}
     \end{subfigure}
     \hfill
     \begin{subfigure}[b]{0.45\textwidth}
         \centering
         \includegraphics[width=\textwidth]{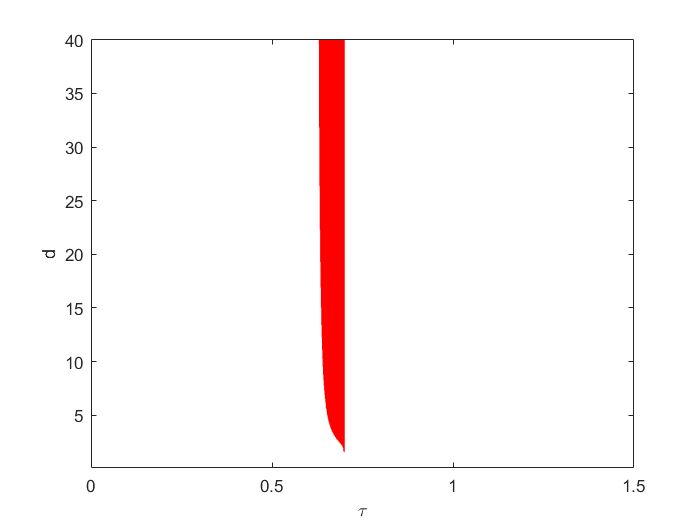}
         \caption{}
         \label{fig0:subfig3}
     \end{subfigure}
          \hfill
     \begin{subfigure}[b]{0.45\textwidth}
         \centering
         \includegraphics[width=\textwidth]{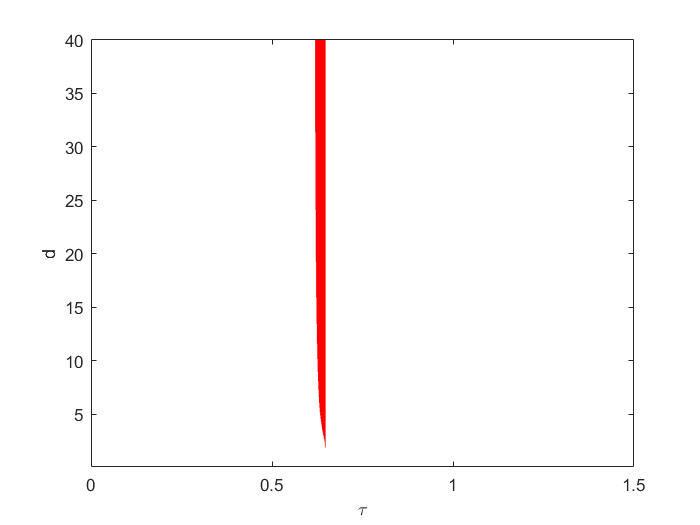}
         \caption{}
         \label{fig0:subfig4}
     \end{subfigure}
        \caption{ The instability Turing region in the plane $\left(\tau, d\right)$ is shown. In blue, the supercritical region, in red is the subcritical region. (a) $k_+=0$. (b) $k_+=0.3$. (c) $k_+=0.7$. (d) $k_+=1.04$. The parameters are chosen as $n_n=2.63$, $n_l=1.09$, $\gamma_n=0.02$, $\gamma_l=0.06$ and $\nu=0.62$ }
        \label{fig0:subfigures4}
\end{figure*}
\section{Weakly nonlinear analysis
}
\label{sec5}

In this section, we employ weakly nonlinear analysis near the bifurcation point to determine the amplitude equation. Our approach combines the multiple scale method \cite{newell1969finite} with asymptotic analysis around the stable steady state.

The multiple-scale method aims to separate the evolution of the solution into two distinct time scales: a fast time, where linear theory remains valid, and a slow time, where nonlinear effects become significant, rendering linear analysis insufficient to describe the behavior of the system \cite{matkowsky1970nonlinear}. On this slow time scale, the amplitude of the pattern evolves. To capture this, we introduce two scaled coordinates that divide time into two components: the slow time $T=\varepsilon t$, and the fast time $t$, where $\varepsilon $ serves as the control parameter measuring the system's proximity to the bifurcation, as indicated in \eqref{bif1}. We expand our solution in terms of $\varepsilon$, where the main term combines the basic pattern with its slowly varying amplitude \cite{hoyle2006pattern}. Our analysis focuses on temporal pattern modulation, without considering slow spatial variations. 

Near the homogeneous steady state $E_2$, we rewrite our original model as:
\begin{equation}
\begin{aligned}
\frac{\partial \mathbf{w}}{\partial t} &= \mathcal{L}^d\mathbf{w} + \frac{1}{2} \mathcal{Q}(\mathbf{w}, \mathbf{w})+ \dfrac{1}{6}
\begin{pmatrix}
h_{nnn}\left(n-n^*\right)^3 + 3h_{nnl}\left(n-n^*\right)^2\left(l-l^*\right) \\
+ 3h_{nll}\left(n-n^*\right)\left(l-l^*\right)^2 + h_{lll}\left(l-l^*\right)^3 \\
\\
\tau\big(h_{nnn}\left(n-n^*\right)^3 + 3h_{nnl}\left(n-n^*\right)^2\left(l-l^*\right) \\
+ 3h_{nll}\left(n-n^*\right)\left(l-l^*\right)^2 + h_{lll}\left(l-l^*\right)^3\big)
\end{pmatrix},
\end{aligned}
\label{recast}
\end{equation}
where $w$ is defined by equation \eqref{L}, with $\mathcal{L}^d$ representing the linear operator:
\begin{equation}
\mathcal{L}^d=J+D^d\nabla^2, 
\end{equation}
$J$ and $D$ are given by \eqref{L}.         
And the bi-linear operator $\mathcal{Q}$ acting on $x=(x^n,x^l)$ and $y=(y^n,y^l)$, respectively, is defined by 
\begin{equation}
\begin{aligned}
\mathcal{Q}(x,y)=&\left(\begin{array}{c} h_{nn} x^n y^n+\left( h_{nl}-k_{+}\right)\left(x^n y^l+x^ly^n\right)+ h_{ll}
 x^l y^l\\
 \\
 \tau h_{nn} x^n y^n+\left( \tau h_{nl}-\nu k_{+}\right) \left(x^n y^l+x^ly^n\right)+ 
 \tau h_{ll}x^l y^l
\end{array}\right),
\end{aligned}
\end{equation}
the terms $h_{nn}$, $h_{ll}$, and $h_{nl}$ represent the second-order partial derivatives of $h$, relating to Nodal and Lefty, and are calculated at the point $(n^*,l^*)$. Similarly, $h_{nnn}$, $h_{lll}$, $h_{nnl}$, and $h_{nll}$ denote the third-order partial derivatives of $h$, also evaluated at $(n^*,l^*)$.

Our goal is to derive the Stuart-Landau equation that captures the pattern dynamics near bifurcation, which we achieve by expanding both the temporal operator and bifurcation parameter as follows:
\begin{equation}
\frac{\partial}{\partial t}= \varepsilon \frac{\partial}{\partial T_1}+ \varepsilon^2 \frac{\partial}{\partial T_2},
 \label{time1}
 \end{equation}
\begin{equation}
d=d_c+\textcolor{blue}{\varepsilon} d_1+\varepsilon^2 d_2+O\left(\varepsilon^3\right),
\label{bif1}
\end{equation}
 where $\varepsilon$ is the control parameter, we define it as $\varepsilon^2=(d-d_c)/d_c$, and we expend the solution as
 
 \begin{equation}
\mathbf{w}=\varepsilon \mathbf{w}_1+\varepsilon^2 \mathbf{w}_2+\varepsilon^3 \mathbf{w}_3+O\left(\varepsilon^4\right).
\label{Wn1}
\end{equation}

Substituting expansions \eqref{time1}-\eqref{Wn1} into \eqref{recast}, with $\mathbf{w}_i=(n_i,l_i)^T$, yields the following linear equations at each order of $\varepsilon$:
\begin{equation}
O(\varepsilon): \quad \mathcal{L}^{d_c} \mathbf{w}_1=\mathbf{0},
\label{1}
\end{equation}

\begin{equation}
O\left(\varepsilon^2\right): \quad \mathcal{L}^{d_c} \mathbf{w}_2=\mathbf{F},
\label{2}
\end{equation}

\begin{equation}
O\left(\varepsilon^3\right): \quad \mathcal{L}^{d_c} \mathbf{w}_3= \mathbf{G}, 
\label{3}
\end{equation}
\vspace{2mm}
where
$$
\mathbf{F}=\frac{\partial \mathbf{w}_1}{\partial T_1}-\frac{1}{2}\mathcal{Q}(\mathbf{w}_1,\mathbf{w}_1)-\left(\begin{array}{cc}
0 & 0 \\
0 &d_1
\end{array}\right) \nabla^2 \mathbf{w}_1,
$$
and
\begin{equation}
\begin{split}
\mathbf{G} &= \frac{\partial \mathbf{w}_1}{\partial T_2}+\frac{\partial \mathbf{w}_2}{\partial T_1}-\mathcal{Q}(\mathbf{w}_1,\mathbf{w}_2)-\left(\begin{array}{cc}
0 & 0\\
0 & d_2
\end{array}\right) \nabla^2 \mathbf{w}_1-\left(\begin{array}{cc}
0 & 0\\
0 & d_1
\end{array}\right) \nabla^2 \mathbf{w}_2\\
&-\dfrac{1}{6}\left(\begin{array}{l}
h_{nnn}n_1^3+3h_{nnl}n_1^2l_1+ 3h_{lln}n_1l_1^2+h_{ll}l_1^3 \\
\tau \left( h_{nnn}n_1^3+3h_{nnl}n_1^2l_1+ 3h_{lln}n_1l_1^2+h_{lll}l_1^3\right)
\end{array}\right) .
\end{split}
\end{equation}

Under Neumann boundary conditions, equation \eqref{1} admits the solution: 
\begin{equation}
\mathbf{w}_1=A_1\left(T_1, T_2\right) \phi \cos \left(\Phi_1 x\right) \cos \left(\Psi_1 y\right),
\label{sol1}
\end{equation}
$A_1$ represents the amplitude of the pattern depending on the low scale time, and the vector  
$$
\mathbf{\phi}=\left(\begin{array}{c}
1 \\
M
\end{array}\right) \in \operatorname{Ker}\left(J-\bar{k}_c^2 D^{d_c}\right) \text { with } M=-\dfrac{J_{11}-\bar{k}_c^2}{J_{12}}. 
$$

\vspace{0.2cm}

At the second order we write the vector $F$ as 
\begin{equation}
\begin{split}
F = &\left(\frac{\partial A_1}{\partial T_1} \phi+d_1 \bar{k_c}^2 A_1\left(\begin{array}{c}
0 \\
M
\end{array}\right)\right) \cos \left(\Phi_1 x\right)\cos \left(\Psi_1 y\right)-\frac{1}{8} A_1^2\sum_{i, j=0.2} \mathcal{Q}(\phi, \phi) \cos \left(i \Phi_1 x\right) \cos \left(j \Psi_1 y\right),
\end{split}
\end{equation}
imposing the Fredholm alternative  condition $\langle F, \psi\rangle=0$, where $<,>$ denotes the scalar product in $L^2(0, \dfrac{2 \pi}{\bar{k}_c})$ and
\begin{equation}
\mathbf{\psi}=\left(\begin{array}{c}
1 \\
M^*
\end{array}\right) \cos \left(\Phi_1 x\right)\left(\Psi_1 y\right) 
\in \operatorname{Ker}\left\{\left(J-\bar{k}_c^2 D^{d_c}\right)^{\dagger}\right\},
\label{psi}
\end{equation}
with $M^*=-\dfrac{J_{11}-\bar{k}_c^2}{J_{21}}$, and $\left(J-k_c^2 D^{d_c}\right)^{\dagger}$ denotes the adjoint matrix of $\left(J-\bar{k}_c^2 D^{d_c}\right)$. At this order we find the Stuart-Landau equation as follow
\begin{equation}
\frac{\partial A_1}{\partial T_1}=\alpha A_1, \quad \alpha=-\dfrac{d_1 MM^* \bar{k}_c^2}{\left( 1+MM^*\right)},
\end{equation}
the amplitude behavior cannot be determined from this order of analysis. We therefore extend to higher order, imposing $T_1=0$ and $d_1=0$ to eliminate secular terms. With the Fredholm alternative satisfied, equation \eqref{2} admits the solution:
\begin{equation}
\mathbf{w}_2=A_1^2 \sum_{i, j=0,2} \mathbf{w}_{2 i j} \cos \left(i \Phi_1 x\right) \cos \left(j \Psi_1 y\right),
\label{sol2}
\end{equation}
the vectors $\mathbf{w}_{2ij}$ are determined by solving the following linear systems:
\begin{equation}
\mathfrak{L}_{i j}^1 \mathbf{w}_{2 i j}=-\frac{1}{8} \mathcal{Q}(\phi, \phi), \quad i, j=0,2,
\label{w_2_ij}
\end{equation}
with $\mathfrak{L}_{i j}^1=J-\left(i^2 \Phi_1^2+j^2 \Psi_1^2\right) D^{d_c}$.\\

The third-order term is obtained by evaluating vector $G$ with solutions \eqref{sol1} and \eqref{sol2}, resulting in:
\begin{equation}
\mathbf{G}=\left(\frac{d A_1}{d T_2} \phi+A_1 \mathbf{G}_{11}^{(1)}+A_1^3 \mathbf{G}_{11}^{(3)}\right) \cos \left(\Phi_1 x\right)\cos \left(\Psi_1x\right)+A_1^3 \mathbf{G}^*,
\end{equation}
where
$$\begin{aligned}
& \mathbf{G}_{11}^{(1)}=d_2 \bar{k_c}^2\left(\begin{array}{c}0 \\ M\end{array}\right),\\
& \mathbf{G}_{11}^{(3)}=-\left(\mathcal{Q}\left(\phi, w_{200}+\dfrac{1}{2}w_{202}+\dfrac{1}{2}w_{220}+\dfrac{1}{4}w_{222}\right)+\frac{9}{16} \boldsymbol{\Gamma}\right),
\end{aligned}$$
$$\begin{aligned}
& \mathbf{G}^*= \mathbf{G}_1^* \cos \left(3 \Phi_1 x\right) \cos \left(\Psi_1 y\right)+ \mathbf{G}_2^* \cos \left(\Phi_1 x\right) \cos \left(3 \Psi_1 y\right)+ \mathbf{G}_3^* \cos \left(3 \Phi_1 x\right) \cos \left(3 \Psi_1 y\right),
\end{aligned}$$
with
$$\begin{aligned} & \mathbf{G}_1^*=-\frac{1}{2} \mathcal{Q}\left(\phi, \mathbf{w}_{220}+\frac{1}{2} \mathbf{w}_{222}\right)-\frac{3}{16} \boldsymbol{\Gamma}, \\ & \mathbf{G}_2^*=-\frac{1}{2} \mathcal{Q}\left(\phi, \mathbf{w}_{202}+\frac{1}{2} \mathbf{w}_{222}\right)-\frac{3}{16} \boldsymbol{\Gamma}, \\ & \mathbf{G}_3^*=-\frac{1}{4} \mathcal{Q}\left(\phi, \mathbf{w}_{222}\right)-\frac{1}{16} \boldsymbol{\Gamma} ,\end{aligned}$$
and 
$$
\mathbf{\Gamma}=\dfrac{1}{6}\left(\begin{array}{l}
h_{nnn}+3h_{nnl}M+ 3h_{lln}M^2+h_{lll}M^3 \\
\tau \left(h_{nnn}+3h_{nnl}M+ 3h_{lln}M^2+h_{lll}M^3\right)
\end{array}\right).
$$

\vspace{0.2cm}
By enforcing the solvability condition $\langle \mathbf{G}, \psi \rangle = 0$, we derive the Stuart–Landau equation that governs the amplitude $A_1$:

\begin{equation}
\frac{\partial A_1}{\partial T_2} = \sigma A_1 - L A_1^3,
\label{landeau_2_positive}
\end{equation}

The coefficients $\sigma$ and $L$ are computed as follows:
\begin{equation}
\sigma = -\frac{\left\langle \mathbf{G}_{11}^{(1)} \cos(\Phi_1 x) \cos(\Psi_1 y),\, \psi \right\rangle}
{\left\langle \phi\, \cos(\Phi_1 x) \cos(\Psi_1 y),\, \psi \right\rangle},\;\text{and}
\;
L = \frac{\left\langle \mathbf{G}_{11}^{(3)} \cos(\Phi_1 x) \cos(\Psi_1 y),\, \psi \right\rangle}
{\left\langle \phi\, \cos(\Phi_1 x) \cos(\Psi_1 y),\, \psi \right\rangle}.
\end{equation}

To analyze the qualitative dynamics of the Stuart–Landau equation~\eqref{landeau_2_positive}, we distinguish two cases based on the sign of the Landau constant $L$:
\begin{enumerate}
\item $L > 0$: In this situation, a supercritical bifurcation occurs in the system.\\
\item $L < 0$: In this situation, a subcritical bifurcation occurs in the system.
\end{enumerate}

\vspace{0.2cm}

\subsection{The supercritical bifurcation case}
Considering the scenario where both $\sigma$ and $L$ are positive parameters, we find that the solution of the Stuart–Landau equation converges to a stable stationary state given by $A_{1\infty} = \sqrt{\sigma / L}$. This indicates that, over time, the system evolves toward a state in which the amplitude of spatial oscillations is primarily determined by $A_{1\infty}$. This leads us to the following proposition.

\begin{prop}
Suppose the following two conditions are satisfied:
\begin{enumerate}
\item The parameter $\varepsilon^2 = \dfrac{d - d_c}{d_c}$ is sufficiently small such that the stable steady state $E_2=(n_*,\, l_*)$ of system~\eqref{N_L_Nondim} in two dimensions becomes unstable only to modes corresponding to the eigenvalue $\bar{k}_c^2$ defined in~\eqref{wave2}, and there exists a unique pair of integers $(z_1,\, s_1)$ satisfying equation~\eqref{wave2}.

\item In equation~\eqref{landeau_2_positive}, the Landau coefficient  $L$ has a positive value.

\end{enumerate}

Therefore, by applying the weakly nonlinear analysis, the asymptotic solution of our system of Nodal and Lefty~\eqref{N_L_Nondim} is given as follows
\begin{equation}
\begin{aligned}
\begin{pmatrix}
n(x, y) \\
l(x, y)
\end{pmatrix}
&=
\begin{pmatrix}
n_* \\
l_*
\end{pmatrix}
+ \varepsilon A_{1\infty} \phi\, \cos\left(\Phi_1 x\right) \cos\left(\Psi_1 y\right)+ \varepsilon^2 A_{1\infty}^2 \sum_{i,j=0,2} \mathbf{w}_{2ij}\, \cos\left( i\Phi_1 x\right) \cos\left( j\Psi_1 y\right)+ O\left(\varepsilon^3\right),
\end{aligned}
\end{equation}
where $\phi$ is defined in equation~\eqref{psi}, and $\mathbf{w}_{2ij}$ are solutions to the linear systems in~\eqref{w_2_ij}.
\label{prop3}
\end{prop}

\subsection{The subcritical bifurcation case:}
In the case where the coefficient $L$ is negative, equation~\eqref{landeau_2_positive} fails to accurately describe the behavior of the amplitude near the bifurcation. Consequently, following the recommendation in \cite{becherer2009probing}, we extend our weakly nonlinear analysis to $O\left(\varepsilon^5\right)$ to derive the quintic Stuart–Landau equation. This extension necessitates introducing new expansions of the partial operator in time and bifurcation parameters as follows:

\begin{equation}
 \frac{\partial}{\partial t}= \varepsilon \frac{\partial}{\partial T_1}+ \varepsilon^2 \frac{\partial}{\partial T_2}+ \varepsilon^3\frac{\partial}{\partial T_3}+ \varepsilon^4 \frac{\partial}{\partial T_4},
 \label{time2}
 \end{equation}
\begin{equation}
d=d_c+\varepsilon^2 d_2+\varepsilon^3 d_3+\varepsilon^4 d_4+O\left(\varepsilon^5\right).
\label{bif2}
\end{equation}

By adding higher order partial derivatives of the nonlinearities to \eqref{recast}  and substituting equations \eqref{time2}, \eqref{bif2} and \eqref{Wn1} into it, we find that equations \eqref{1}-\eqref{3} remain valid. Assuming the amplitude equation holds, the solvability condition $\langle\mathbf{G}, \mathbf{\psi}\rangle=0$ is satisfied. Then, the solution of the third order is given by 
\begin{equation}
\begin{split}
\mathbf{w}_3 &= A_1 \mathbf{w}_{311}^{(1)} \cos \left(\Phi_1 x\right) \cos \left(\Psi_1 y\right)+ A_1^3 \sum_{i, j=1,3} \mathbf{w}_{3 i j} \cos \left(i \Phi_1 x\right) \cos \left(j \Psi_1 y\right),
\end{split}
\label{sol3}
\end{equation}
where the vectors $\mathbf{w}_{3ij}$, with $i$ and $j$ ranging from 1 to 3, are defined by
\begin{equation}
\begin{aligned} & \mathfrak{L}_{11}^1 \mathbf{w}_{311}^{(1)}=\mathbf{G}_{11}^{(1)}+\sigma \phi, \quad \mathfrak{L}_{11}^1 \mathbf{w}_{311}=\mathbf{G}_{11}^{(3)}-L \phi, \\ & \mathfrak{L}_{31}^1 \mathbf{w}_{331}=\mathbf{G}_1^*, \quad \mathfrak{L}_{13}^1 \mathbf{w}_{313}=\mathbf{G}_2^*, \quad \mathfrak{L}_{33}^1 \mathbf{w}_{333}=\mathbf{G}_3^* .\end{aligned}
\label{w_3_ij}
\end{equation}

For the fourth order, we obtain the following linear equation 
\begin{equation}
O\left(\varepsilon^4\right): \quad \mathcal{L}^{d_c} \mathbf{w}_4=\mathbf{H},
\end{equation}
where
$$
\begin{aligned}
\mathbf{H} &= \frac{\partial \mathbf{w}_1}{\partial T_3}+\frac{\partial \mathbf{w}_2}{\partial T_2}-\frac{1}{2} \mathcal{Q}\left(\mathbf{w}_2, \mathbf{w}_2\right) -\mathcal{Q}\left(\mathbf{w}_1, \mathbf{w}_3\right) 
-\left(\begin{array}{cc}
0 & 0 \\
0 & d_2
\end{array}\right) \nabla^2 \mathbf{w}_2-\left(\begin{array}{cc}
0 & 0 \\
0 & d_3
\end{array}\right) \nabla^2 \mathbf{w}_1 \\
&-\dfrac{1}{24}\begin{pmatrix}
h_{nnnn}n_1^4+4h_{nnnl}n_1^3l_1 
+ 4h_{llln}n_1l_1^3+6h_{nnll}n_1^2l_1^2+h_{llll}l_1^4 \\
\tau(h_{nnnn}n_1^4+4h_{nnnl}n_1^3l_1
+ 4h_{llln}n_1l_1^3+6h_{nnll}n_1^2l_1^2+h_{llll}l_1^4)
\end{pmatrix},
\end{aligned}
$$
and $h_{nnnn}$, $h_{nnnl}$, $h_{nnll}$, $h_{llln}$ and $h_{llll}$ are the fourth-order partial derivatives of
$h$, evaluated at $(n^*,l^*)$.

\vspace{0.2cm}

Integrating the solutions \eqref{sol1}, \eqref{sol2} and \eqref{sol3} into the vector $H$, and completing a rigorous calculation, we obtain  
\begin{equation}
\begin{aligned}
\mathbf{H} &= \mathbf{H}_{11} \cos \left(\Phi_1 x\right) \cos \left(\Psi_1 y\right) + \sum_{i, j=0,2} A_1^2 \mathbf{H}_{i j}^{(1)} \cos \left(i \Phi_1 x\right) \cos \left(j \Psi_1 y\right)+ \sum_{i, j=0,2,4} A_1^4 \mathbf{H}_{i j} \cos \left(i \Phi_1 x\right) \cos \left(j \Psi_1 y\right),
\end{aligned}
\end{equation}
where 
$$\begin{aligned}
\mathbf{H}_{00}^{(1)}&= 2 \sigma \mathbf{w}_{200}-\frac{1}{4} \mathcal{Q}\left(\phi, \mathbf{w}_{311}^{(1)}\right), \\
\mathbf{H}_{00}&= -2 L \mathbf{w}_{200}-\frac{1}{4} \mathcal{Q}\left(\phi, \mathbf{w}_{311}\right)
-\frac{1}{2} \mathcal{Q}\left(\mathbf{w}_{200}, \mathbf{w}_{200}\right)-\frac{1}{4} \mathcal{Q}\left(\mathbf{w}_{220}, \mathbf{w}_{220}\right)
-\frac{1}{4} \mathcal{Q}\left(\mathbf{w}_{202}, \mathbf{w}_{202}\right)-\frac{1}{8} \mathcal{Q}\left(\mathbf{w}_{222}, \mathbf{w}_{222}\right)-\frac{9}{64} \boldsymbol{\Theta},\\
\mathbf{H}_{11}&=\frac{\partial A_1}{\partial T_3} \phi+d_3 \bar{k_c}^2\left(\begin{array}{c}
0 \\
M
\end{array}\right),\\
\mathbf{H}_{20}^{(1)}&=2 \sigma \mathbf{w}_{220}-\frac{1}{4} \mathcal{Q}\left(\phi, \mathbf{w}_{311}^{(1)}\right)
+4 d_2 \Phi_1^2\left(\begin{array}{cc}
0 & 0 \\
0 & d_2
\end{array}\right) \mathbf{w}_{220}, \\
\mathbf{H}_{20}&=-2 L \mathbf{w}_{220}-\frac{1}{4} \mathcal{Q}\left(\boldsymbol{\phi}, \mathbf{w}_{311}+\mathbf{w}_{331}\right)-\mathcal{Q}\left(\mathbf{w}_{200}, \mathbf{w}_{220}\right)-\frac{1}{2} \mathcal{Q}\left(\mathbf{w}_{202}, \mathbf{w}_{222}\right)
-\frac{3}{16} \boldsymbol{\Theta},\\
\mathbf{H}_{02}^{(1)}&=2 \sigma \mathbf{w}_{202}-\frac{1}{4} \mathcal{Q}\left(\phi, \mathbf{w}_{311}^{(1)}\right)
+4 d_2 \Psi_1^2\left(\begin{array}{cc}0 & 0 \\ 0 & d_2\end{array}\right) \mathbf{w}_{202},\\
\mathbf{H}_{02}&=-2 L \mathbf{w}_{202}-\frac{1}{4} \mathcal{Q}\left(\phi, \mathbf{w}_{311}+\mathbf{w}_{313}\right)
-\mathcal{Q}\left(\mathbf{w}_{200}, \mathbf{w}_{202}\right)-\frac{1}{2} \mathcal{Q}\left(\mathbf{w}_{220}, \mathbf{w}_{222}\right)
-\frac{3}{16} \boldsymbol{\Theta},\\
\mathbf{H}_{22}^{(1)}&=2\sigma \mathbf{w}_{222}-\frac{1}{4} \mathcal{Q}\left(\phi, \mathbf{w}_{311}^{(1)}\right)
+4d_2\bar{k}_c^2\left(\begin{array}{cc}0 & 0 \\ 0 & d_2\end{array}\right) \mathbf{w}_{222},\\
\mathbf{H}_{22}&=-2 L \mathbf{w}_{222}-\frac{1}{4} \mathcal{Q}\left(\phi, \sum_{i, j=1,3} \mathbf{w}_{3 i j}\right) -\mathcal{Q}\left(\mathbf{w}_{200}, \mathbf{w}_{222}\right)-\mathcal{Q}\left(\mathbf{w}_{202}, \mathbf{w}_{220}\right)
-\frac{1}{4} \boldsymbol{\Theta},\\
\mathbf{H}_{40}&=-\frac{1}{4}\left(\mathcal{Q}\left(\phi, \mathbf{w}_{331}\right)+\mathcal{Q}\left(\mathbf{w}_{220}, \mathbf{w}_{220}\right)
+\frac{1}{2} \mathcal{Q}\left(\mathbf{w}_{222}, \mathbf{w}_{222}\right)\right)-\frac{3}{64} \boldsymbol{\Theta},\\
\mathbf{H}_{04}&=-\frac{1}{4}\left(\mathcal{Q}\left(\phi,\mathbf{w}_{313}\right)+\mathcal{Q}\left(\mathbf{w}_{202}, \mathbf{w}_{202}\right)
+\frac{1}{2} \mathcal{Q}\left(\mathbf{w}_{222}, \mathbf{w}_{222}\right)\right)-\frac{3}{64}\boldsymbol{\Theta},\\
\mathbf{H}_{42}&=-\frac{1}{4}\left(\mathcal{Q}\left(\phi,\mathbf{w}_{331}+\mathbf{w}_{333}\right)
+2 \mathcal{Q}\left(\mathbf{w}_{220}, \mathbf{w}_{222}\right)\right)-\frac{1}{16}\boldsymbol{\Theta},\\ 
\mathbf{H}_{24}&=-\frac{1}{4}\left(\mathcal{Q}\left(\phi, \mathbf{w}_{313}+\mathbf{w}_{333}\right)
+2 \mathcal{Q}\left(\mathbf{w}_{202}, \mathbf{w}_{222}\right)\right)-\frac{1}{16}\boldsymbol{\Theta}, \\  
\mathbf{H}_{44}&=-\frac{1}{4}\left(\mathcal{Q}\left(\phi, \mathbf{w}_{333}\right)
+\frac{1}{2} \mathcal{Q}\left(\mathbf{w}_{222}, \mathbf{w}_{222}\right)\right)-\frac{1}{64}\boldsymbol{\Theta},
\end{aligned}$$

and

$
\Theta=\dfrac{1}{24}\left(\begin{array}{@{}l@{}}
h_{nnnn}+4h_{nnnl}M+ 4h_{nlll}M^3+6h_{nnll}M^2+h_{llll}M^4 \\
\tau\left(h_{nnnn}+4h_{nnnl}M+ 4h_{nlll}M^3+6h_{nnll}M^2+h_{llll}M^4\right)
\end{array}\right).
$

\vspace{0.2cm}
By setting $T_3 = 0$ and $d_3 = 0$, we satisfy the solvability condition. Consequently, the fourth-order solution that fulfills the Neumann boundary conditions is
\begin{equation}
\begin{split}
\mathbf{w}_4 &= A_1^2 \sum_{i, j=0,2} \mathbf{w}_{4 i j}^{(1)} \cos \left(i \Phi_1 x\right) \cos \left(j \Psi_1 y\right)+ A_1^4 \sum_{i, j=0,2,4} \mathbf{w}_{4 i j} \cos \left(i \Phi_1 x\right) \cos \left(j \Psi_1 y\right),
\end{split}
\label{sol4}
\end{equation}
where the vectors $\mathbf{w}_{4ij}^{(1)}$ and $\mathbf{w}_{4ij}$, corresponding to indices $i,j=0,2$ and $i,j=0,2,4$ respectively, solve the following linear equations:
\begin{equation}
\mathfrak{L}_{i j}^1 \mathbf{w}_{4 i j}^{(1)}=\mathbf{H}_{i j}^{(1)}, \, \mathfrak{L}_{i j}^1 \mathbf{w}_{4 i j}=\mathbf{H}_{i j}.
\label{w_4_ij}
\end{equation}

For the final order 

$$
O\left(\varepsilon^5\right): \quad \mathcal{L}^{d_c} \mathbf{w}_5=\mathbf{P},
$$

where

$$
\begin{aligned}
\mathbf{P} &= \partial_{T_4} \mathbf{w}_1+\partial_{T_2} \mathbf{w}_3-\left(\begin{array}{cc}
0 & 0\\
0 & d_2
\end{array}\right) \nabla^2 \mathbf{w}_3 -\left(\begin{array}{cc}
0 & 0\\
0 & d_4
\end{array}\right) \nabla^2 \mathbf{w}_1 -\dfrac{1}{120}\left(\begin{array}{l}
h_{nnnnn}n_1^5+5h_{nnnnl}n_1^4l_1\\
\tau\left(h_{nnnnn}n_1^5+5h_{nnnnl}n_1^4l_1\right)\end{array}\right) \\
&-\dfrac{1}{120}\left(\begin{array}{@{}l@{}}
5h_{nllll}n_1l_1^4+10h_{nnnll}n_1^3l_1^2+10h_{nnlll}n_1^2l_1^3+h_{lllll}l_1^5\\
\tau \left(5h_{nllll}n_1l_1^4+10h_{nnnll}n_1^3l_1^2+10h_{nnlll}n_1^2l_1^3+h_{lllll}l_1^5\right)
\end{array}\right),
\end{aligned}
$$
and $h_{nnnnn}$, $h_{nnnnl}$, $h_{nnnll}$, $h_{nllll}$,$h_{nnlll}$ and $h_{lllll}$ are the fifth-order partial derivatives of
$h$, evaluated at $(n^*,l^*)$.

\vspace{0.2cm}

Similar to our previous approach, by substituting  \eqref{sol1}, \eqref{sol2}, \eqref{sol3} and \eqref{sol4} in the vector $P$, we achieve
\begin{equation}
\mathbf{P}=\left(\frac{\partial A_1}{\partial T_4} \phi+\mathbf{P}_{11}^{(1)} A_1+\mathbf{P}_{11}^{(3)} A_1^3+\mathbf{P}_{11}^{(5)} A_1^5\right) \cos \left(\Phi_1 x\right) \cos \left(\Psi_1 y\right)+\mathbf{P}^*,
\end{equation}

where
$$\begin{aligned}
\mathbf{P}_{11}^{(1)}= & \sigma \mathbf{w}_{311}^{(1)}+d_4 \bar{k_c}^2\left(\begin{array}{c}
0 \\
M
\end{array}\right)+\left(\begin{array}{cc}
0 & 0 \\
0 & d_2
\end{array}\right) \bar{k_c}^2 \mathbf{w}_{311}^{(1)}, \\
\mathbf{P}_{11}^{(3)}= & -L \mathbf{w}_{311}+3 \sigma \mathbf{w}_{311}-\mathcal{Q}\left(\phi, \mathbf{w}_{400}^{(1)}+\frac{1}{2} \mathbf{w}_{420}^{(1)}+\frac{1}{2} \mathbf{w}_{402}^{(1)}+\frac{1}{4} \mathbf{w}_{422}^{(1)}\right) -\mathcal{Q}\left(\mathbf{w}_{200}+\frac{1}{2} \mathbf{w}_{220}+\frac{1}{2} \mathbf{w}_{202}+\frac{1}{4} \mathbf{w}_{222}, \mathbf{w}_{311}^{(1)}\right) \\
&+\bar{k}_c^2\left(\begin{array}{cc}
0 & 0 \\
0 & d_2
\end{array}\right) \mathbf{w}_{311},\\
\mathbf{P}_{11}^{(5)}= & -3 L \mathbf{w}_{311} -\mathcal{Q}\left(\phi, \mathbf{w}_{400}+\frac{1}{2} \mathbf{w}_{420}+\frac{1}{2} \mathbf{w}_{402}+\frac{1}{4} \mathbf{w}_{422}\right)-\mathcal{Q}\left(\mathbf{w}_{200}+\frac{1}{2} \mathbf{w}_{220}+\frac{1}{2} \mathbf{w}_{202}+\frac{1}{4} \mathbf{w}_{222}, \mathbf{w}_{311}\right)-\frac{1}{2} \mathcal{Q}\left(\mathbf{w}_{200}, \mathbf{w}_{331}\right)\\
&-\frac{1}{2} \mathcal{Q}\left(\mathbf{w}_{202}, \mathbf{w}_{313}\right) -\frac{1}{4} \mathcal{Q}\left(\mathbf{w}_{222}, \mathbf{w}_{313}+\mathbf{w}_{331}+\mathbf{w}_{333}\right)-\frac{25}{64} \Upsilon,
\end{aligned}$$
and 
$$
\begin{aligned}
\Upsilon &= \dfrac{1}{120}\left(\begin{array}{@{}l@{}}
h_{nnnnn}+5h_{nnnnl}M+ 10h_{nnnll}M^2 +10h_{nnlll}M^3\\
+5h_{nllll}M^4 +h_{lllll}M^5 \\
\\
\tau\left(h_{nnnnn}+5h_{nnnnl}M+ 10h_{nnnll}M^2+10h_{nnlll}M^3\right. \\
\left.+5h_{nllll}M^4 +h_{lllll}M^5\right)
\end{array}\right).
\end{aligned}
$$
The solvability condition of  the fifth order is 
\begin{equation}
\frac{\partial A_1}{\partial T_4}=\widehat{\sigma} A_1-\widehat{L} A_1^3+\widehat{R} A_1^5,
\label{S_L_Q_2D}
\end{equation}
where the coefficients are

$$
\begin{aligned}
& \widehat{\sigma}=-\frac{\left\langle\mathbf{P}_{11}^{(1)} \cos \left(\Phi_1 x\right) \cos \left(\Psi_1 y\right), \boldsymbol{\psi}\right\rangle}{\left\langle\phi \cos \left(\Phi_1 x\right) \cos \left(\Psi_1 y\right), \boldsymbol{\psi}\right\rangle},\quad
\widehat{L}=\frac{\left\langle\mathbf{P}_{11}^{(3)} \cos \left(\Phi_1 x\right) \cos \left(\Psi_1 y\right), \boldsymbol{\psi}\right\rangle}{\left\langle\phi \cos \left(\Phi_1 x\right) \cos \left(\Psi_1 y\right), \boldsymbol{\psi}\right\rangle}, \;
\text{and}
 \; \widehat{R}=-\frac{\left\langle\mathbf{P}_{11}^{(5)} \cos \left(\Phi_1 x\right) \cos \left(\Psi_1 y\right), \boldsymbol{\psi}\right\rangle}{\left\langle\phi \cos \left(\Phi_1 x\right) \cos \left(\Psi_1 y\right), \boldsymbol{\psi}\right\rangle} .
\end{aligned}
$$
\vspace{0.2cm}

Adding up \eqref{S_L_Q_2D} to \eqref{landeau_2_positive} one gets the quintic Stuart–Landau equation
\begin{equation}
\frac{\partial A_1}{\partial T_2}=\bar{\sigma} A_1-\bar{L} A_1^3+\bar{R} A_1^5,
\label{S_L_Q_2D_1}
\end{equation}
with 
$
\bar{\sigma}=\sigma+\varepsilon^2 \widehat{\sigma}, \quad \bar{L}=L+\varepsilon^2 \widehat{L}, \quad \bar{R}=\varepsilon^2 \widehat{R}.
$

Assuming that $L < 0$, $\bar{\sigma} > 0$, $\bar{L} < 0$, and $\bar{R} < 0$ with $\varepsilon \ll 1$, the system possesses a stable stationary state characterized by the following amplitude
$
A_{1 \infty}=\sqrt{\frac{\bar{L}-\sqrt{\bar{L}^2-4 \bar{\sigma} \bar{R}}}{2 \bar{R}}} .
$
from this, we obtain:
\begin{prop}
Assume that condition 1 of Proposition~\ref{prop3} holds, and additionally:
\begin{enumerate}
    \item The parameter $\varepsilon^2 = \dfrac{d - d_c}{d_c}$ is sufficiently small such that the Landau coefficient $L$ in equation~\eqref{landeau_2_positive} is negative.

    \item The coefficient $\bar{R}$  is negative.
\end{enumerate}

Then, according to equation~\eqref{Wn1}, the Nodal and Lefty system~\eqref{N_L_Nondim} admits an asymptotic solution given as follows

\begin{equation}
\begin{aligned}
&\left(\begin{array}{l}
n(x, y) \\
l(x, y)
\end{array}\right)=\left(\begin{array}{l}
n_* \\
l_*
\end{array}\right)+\varepsilon A_{1 \infty} \phi \cos \left(\Phi_1 x\right) \cos \left(\Psi_1 y\right)+\varepsilon^2 A_{1 \infty}^2 \sum_{i, j=0,2} \mathbf{w}_{2 i j} \cos \left(i \Phi_1 x\right) \cos \left(j \Psi_1 y\right)\\
&+\varepsilon^3\left(A_{1 \infty} \mathbf{w}_{311}^{(1)} \cos \left(\Phi_1 x\right) \cos \left(\Psi_1 y\right)\right.\left.\quad+A_{1 \infty}^3 \sum_{i, j=1,3} \mathbf{w}_{3 i j} \cos \left(i \Phi_1 x\right) \cos \left(j \Psi_1 y\right)\right) \\
&+\varepsilon^4\left(\sum_{i, j=0,2} A_{1 \infty}^2 \mathbf{w}_{4 i j}^{(1)} \cos \left(i \Phi_1 x\right) \cos \left(j \Psi_1 y\right)\right. \left.\quad+\sum_{i, j=0,2,4} A_{1 \infty}^4 \mathbf{w}_{4 i j} \cos \left(i \Phi_1 x\right) \cos \left(j \Psi_1 y\right)\right)+O\left(\varepsilon^5\right).
\end{aligned}
\end{equation}
where $\phi$ comes from equation~\eqref{psi}, and the vectors $\mathbf{w}_{kij}$ and $\mathbf{w}_{kij}^{(1)}$ solve the linear systems \eqref{w_2_ij}, \eqref{w_3_ij} and \eqref{w_4_ij} for each order $k=\{2,3,4\}$.
\label{prop4}
\end{prop}

\begin{figure*}[!htbp]
  \centering
  \subfloat[]{\includegraphics[width=0.45\textwidth]{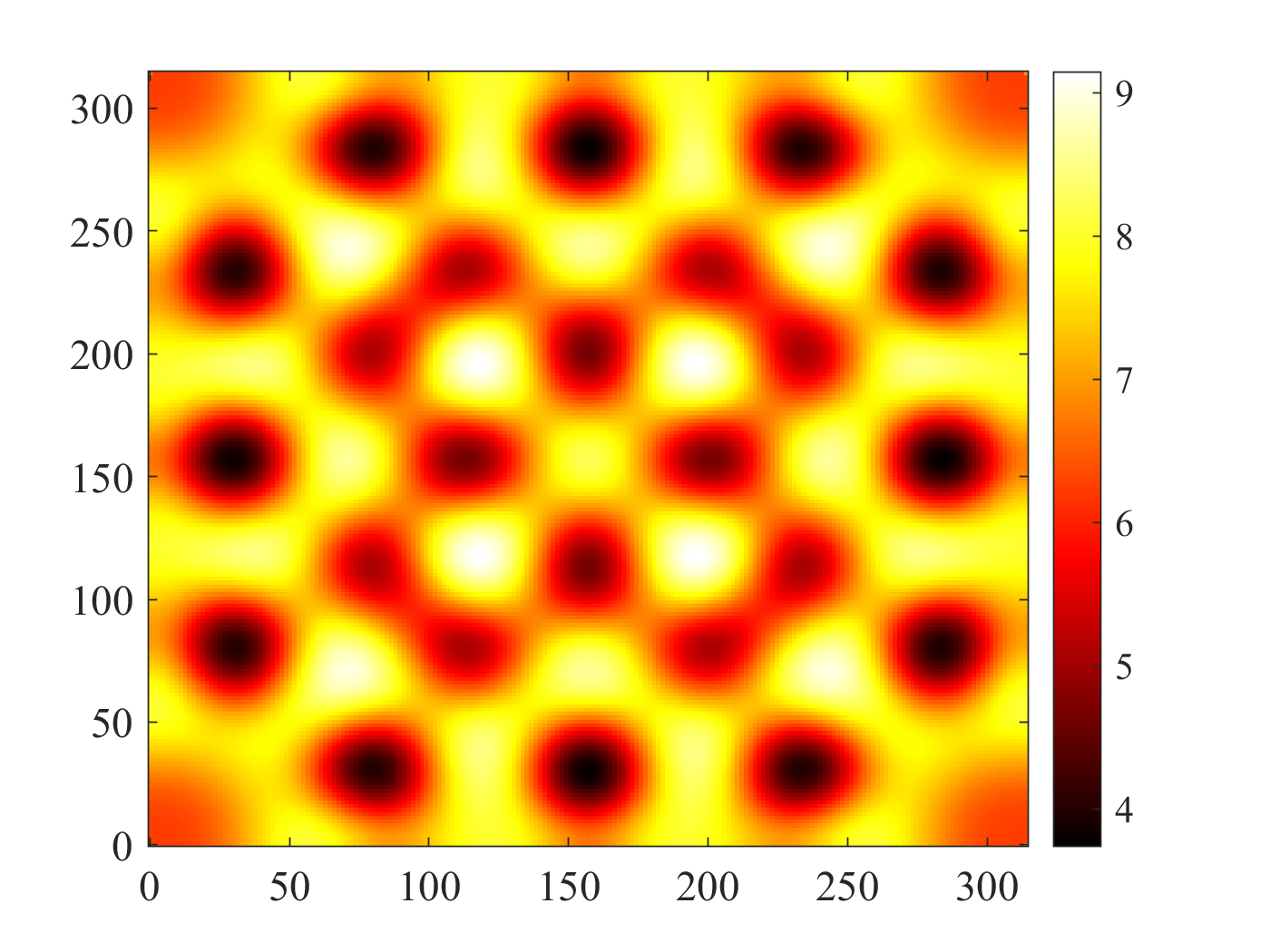}\label{squares}}
  \hspace{15mm}
  \subfloat[]{\includegraphics[width=0.45\textwidth]{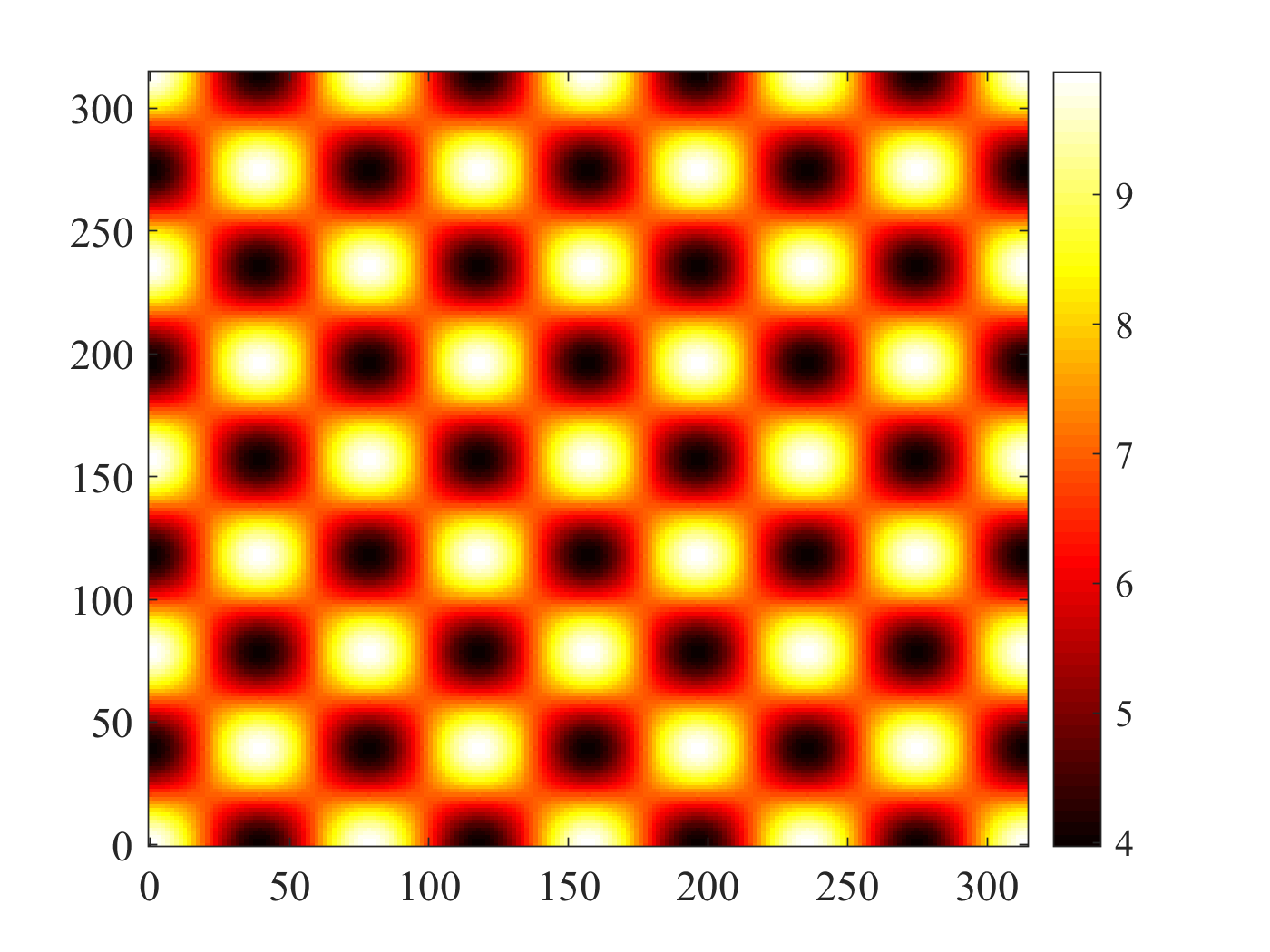}\label{wnl_squares}}
  \caption{Single-mode pattern formation in the supercritical regime. (a) numerical simulations of system \eqref{N_L_Nondim} and (b) the first-order weakly nonlinear approximation. The parameters used for these simulations are $\tau = 3.10$, $k_{+} = 0$, $n_n = 2.63$, $n_l = 1.09$, $\gamma_n = 0.02$, $\gamma_l = 0.06$, and $\nu = 0.62$. The numerical simulations were performed until $t_f = 5000$.
}\label{fig_squares}
\end{figure*}

\begin{figure*}[!htbp]
  \centering
  \subfloat[]{\includegraphics[width=0.45\textwidth]{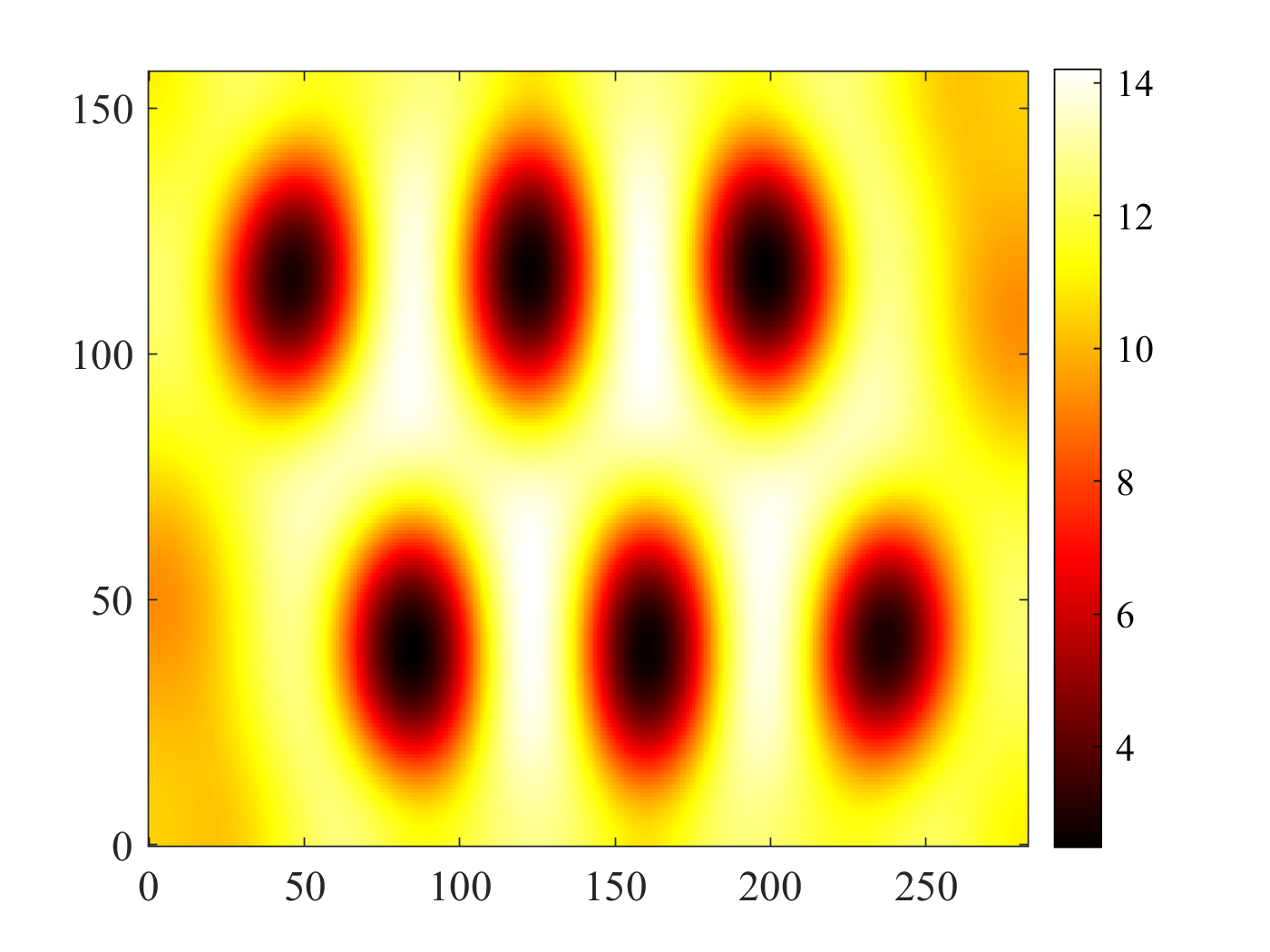}\label{fig_stripe_comp:subfig1}}
  \hspace{15mm}
  \subfloat[]{\includegraphics[width=0.45\textwidth]{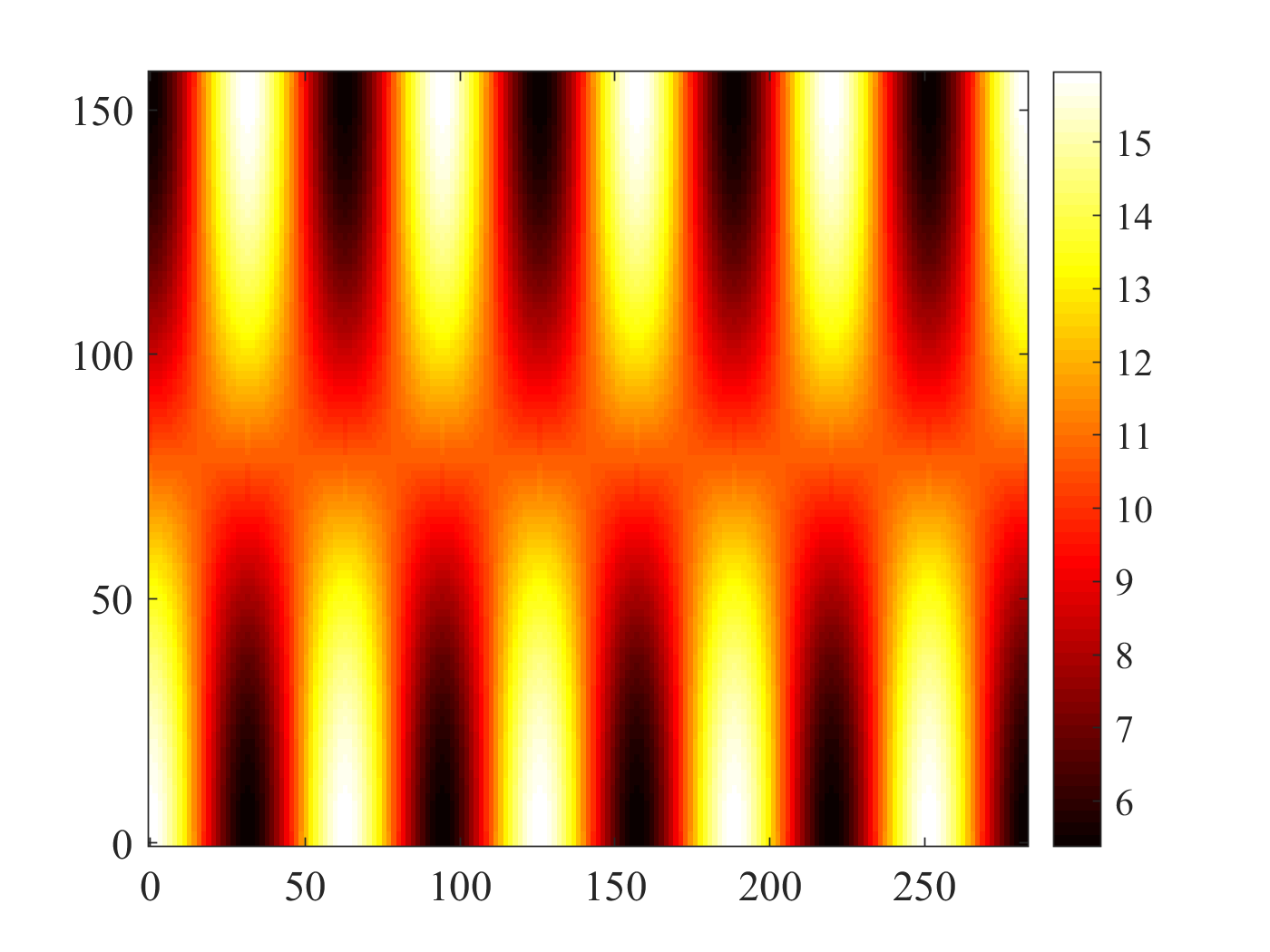}\label{fig_stripe_comp:subfig2}}
  \caption{Single-mode pattern formation in the subcritical regime. (a) Numerical simulation, (b) First-order weakly nonlinear approximation. The parameters used for these simulations are $\tau = 2.5$, $k_{+} = 0$, $n_n = 2.63$, $n_l = 1.09$, $\gamma_n = 0.02$, $\gamma_l = 0.06$, and $\nu = 0.62$. The numerical simulations were performed until $t_f = 5000$,}
  \label{fig:single_mode_subcritical}
\end{figure*}

\begin{figure*}[!htbp]
  \centering
  \subfloat[]{\includegraphics[width=0.45\textwidth]{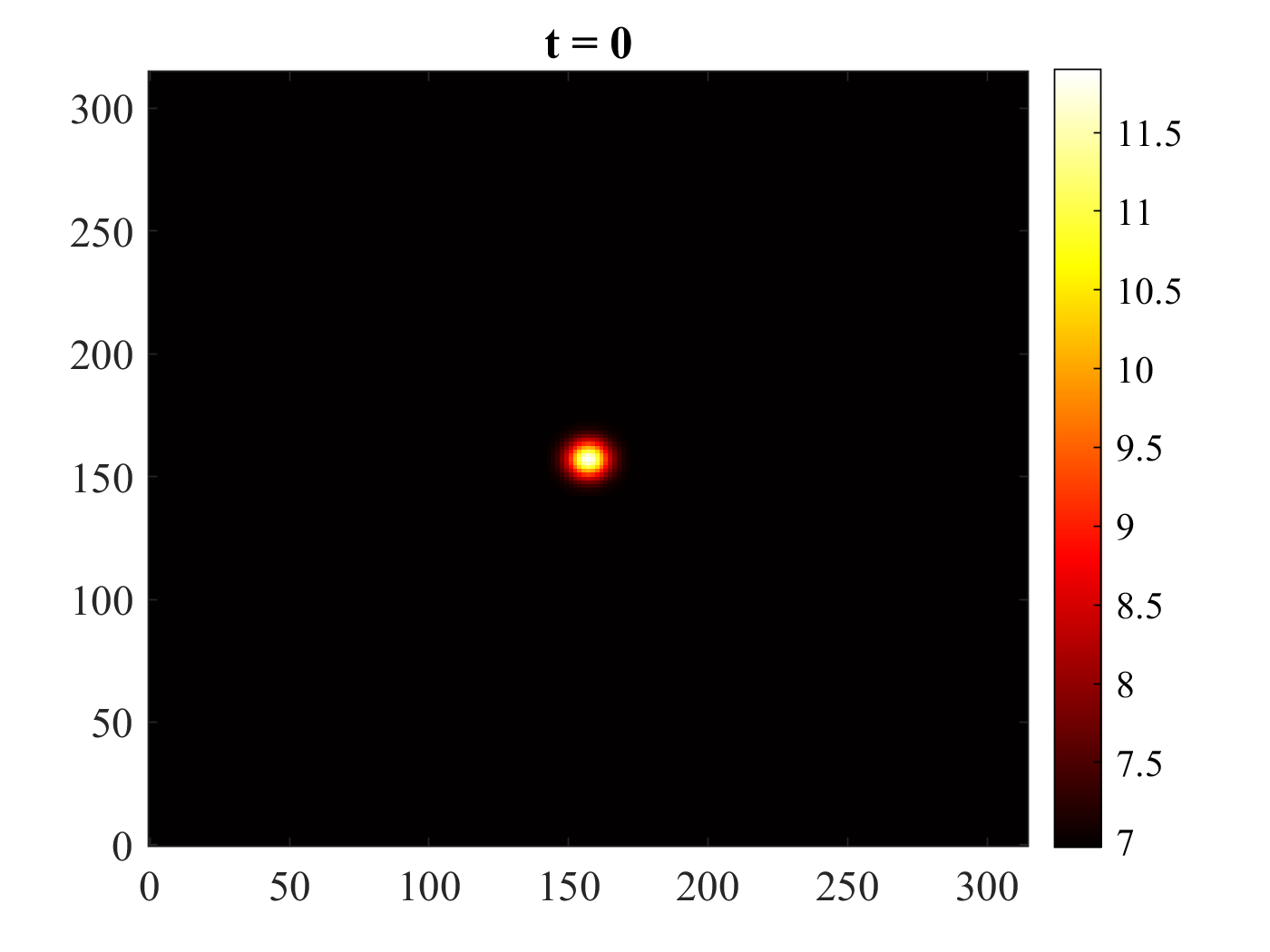}\label{fig_super2:subfig1}}
  \hfill
  \subfloat[]{\includegraphics[width=0.45\textwidth]{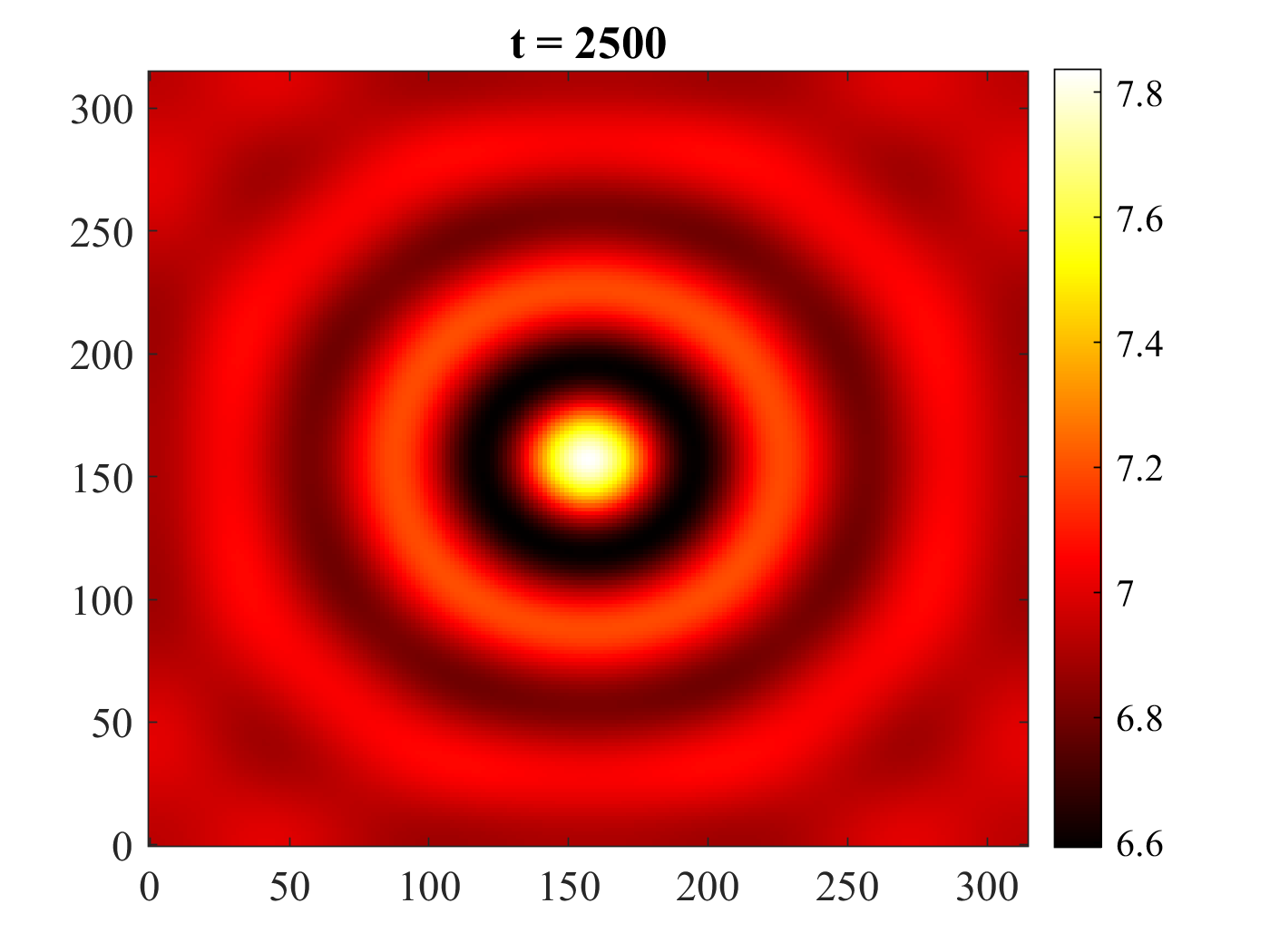}\label{fig_super2:subfig2}}
  \\
  \subfloat[]{\includegraphics[width=0.45\textwidth]{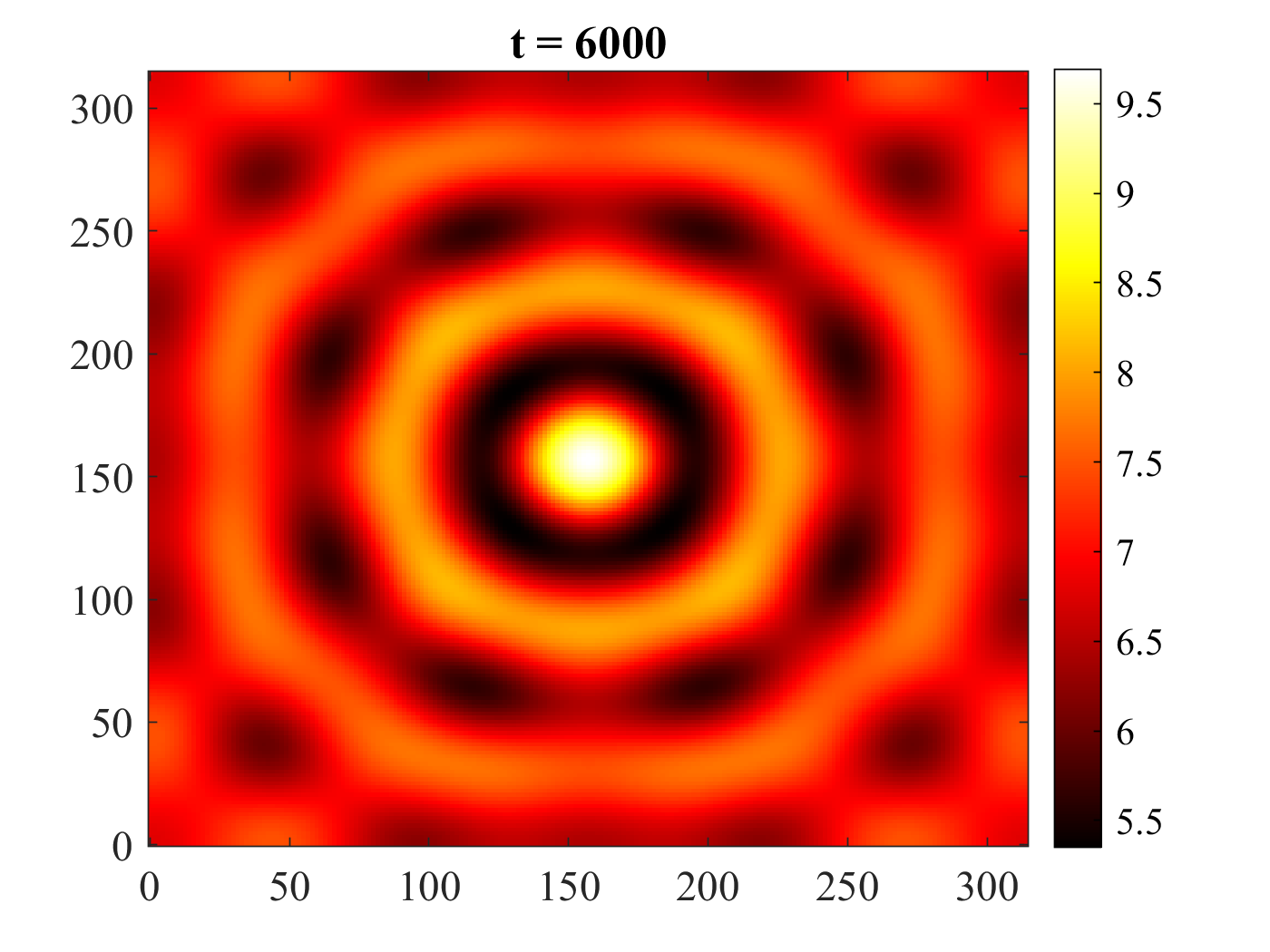}\label{fig_super2:subfig3}}
  \hfill
  \subfloat[]{\includegraphics[width=0.45\textwidth]{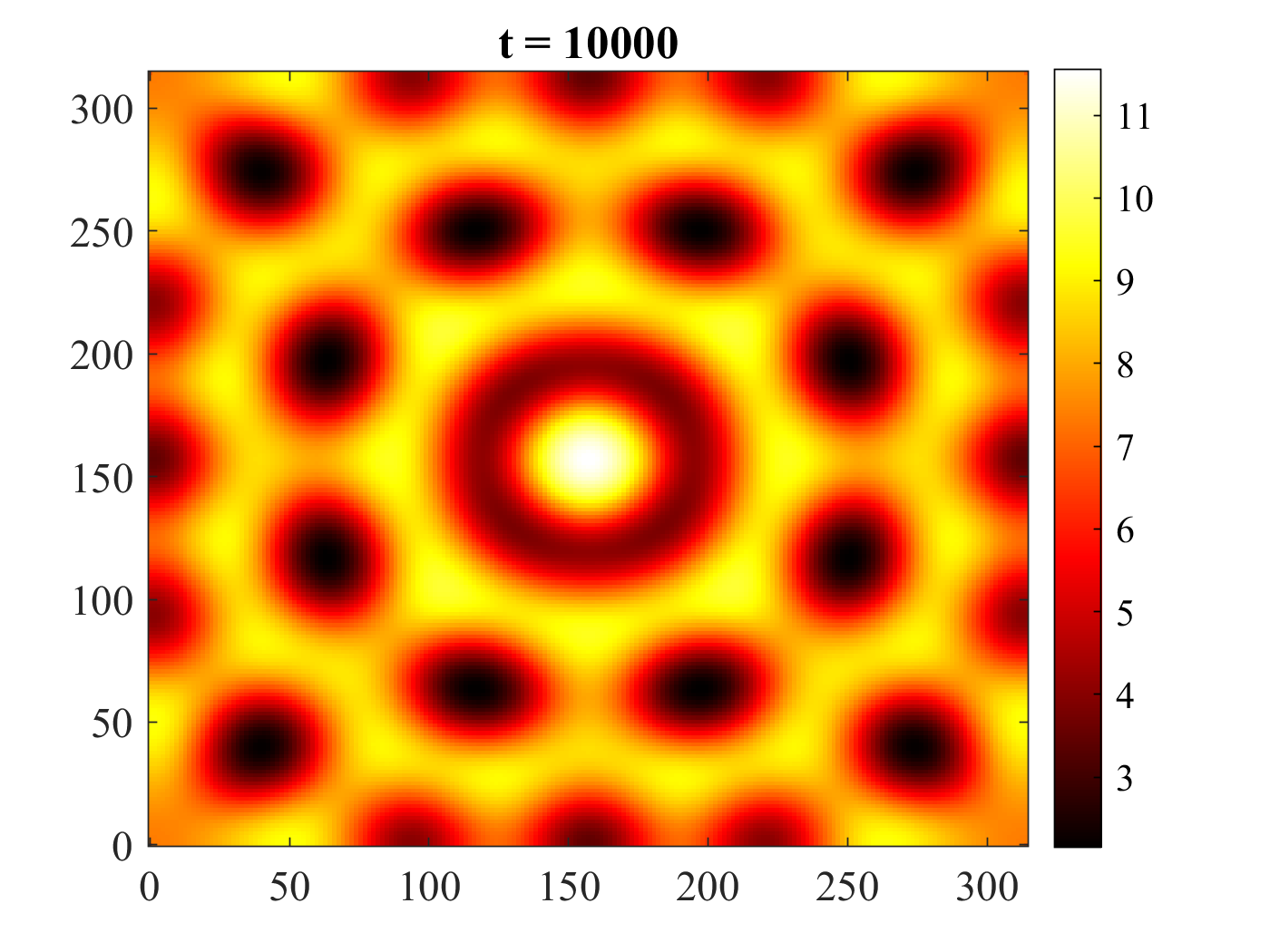}\label{fig_super2:subfig4}}
  \caption{Temporal evolution of supercritical pattern formation for $d=19.5$. Parameters: $\tau=3.10$, $k_+=0$, $n_n = 2.63$, $n_l = 1.09$, $\gamma_n = 0.02$, $\gamma_l = 0.06$, and $\nu = 0.62$.}
  \label{fig:supercritical_d19}
\end{figure*}

\begin{figure*}[!htbp]
  \centering
  \subfloat[]{\includegraphics[width=0.45\textwidth]{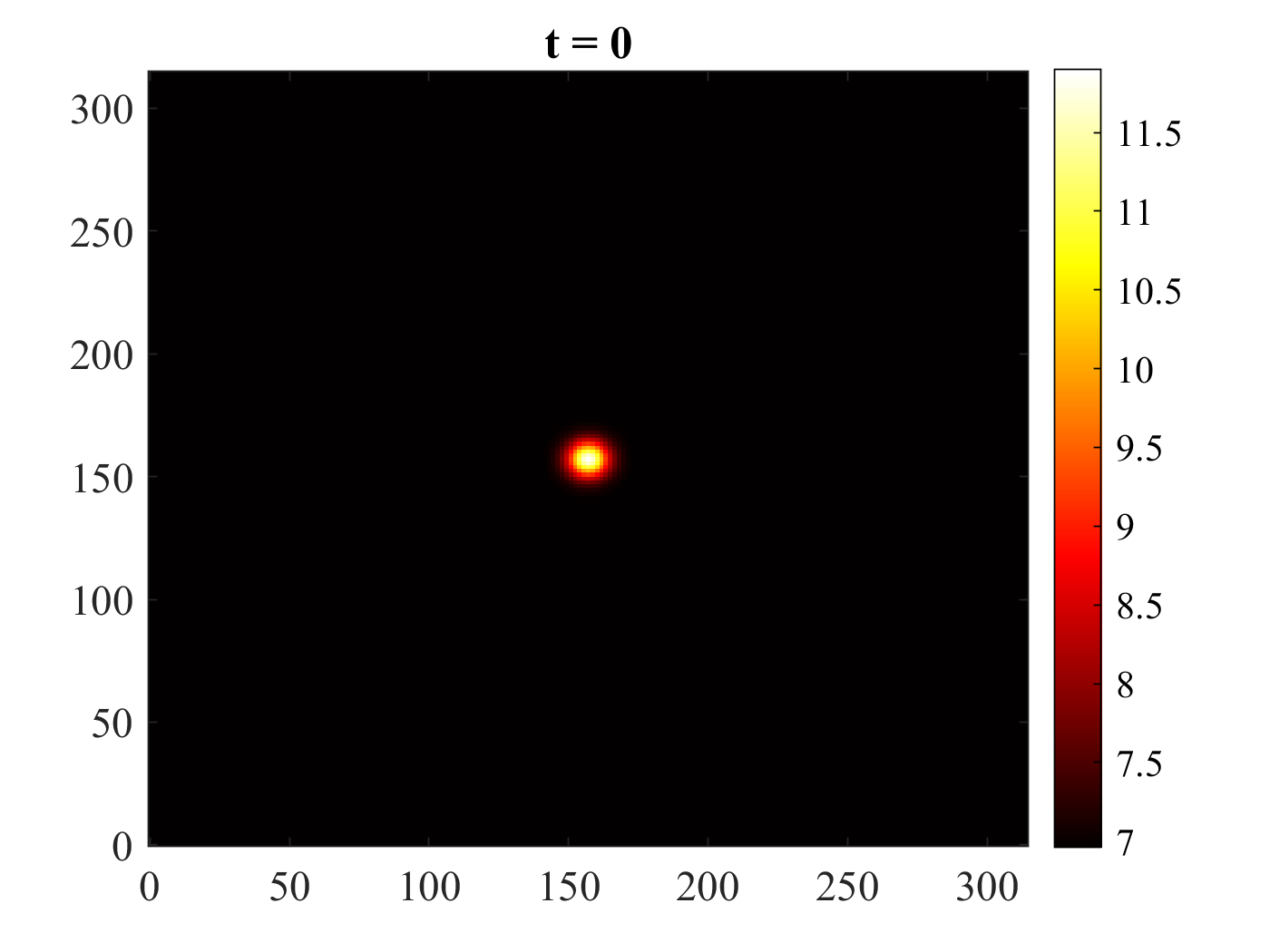}\label{fig_super1:subfig1}}
  \hfill
  \subfloat[]{\includegraphics[width=0.45\textwidth]{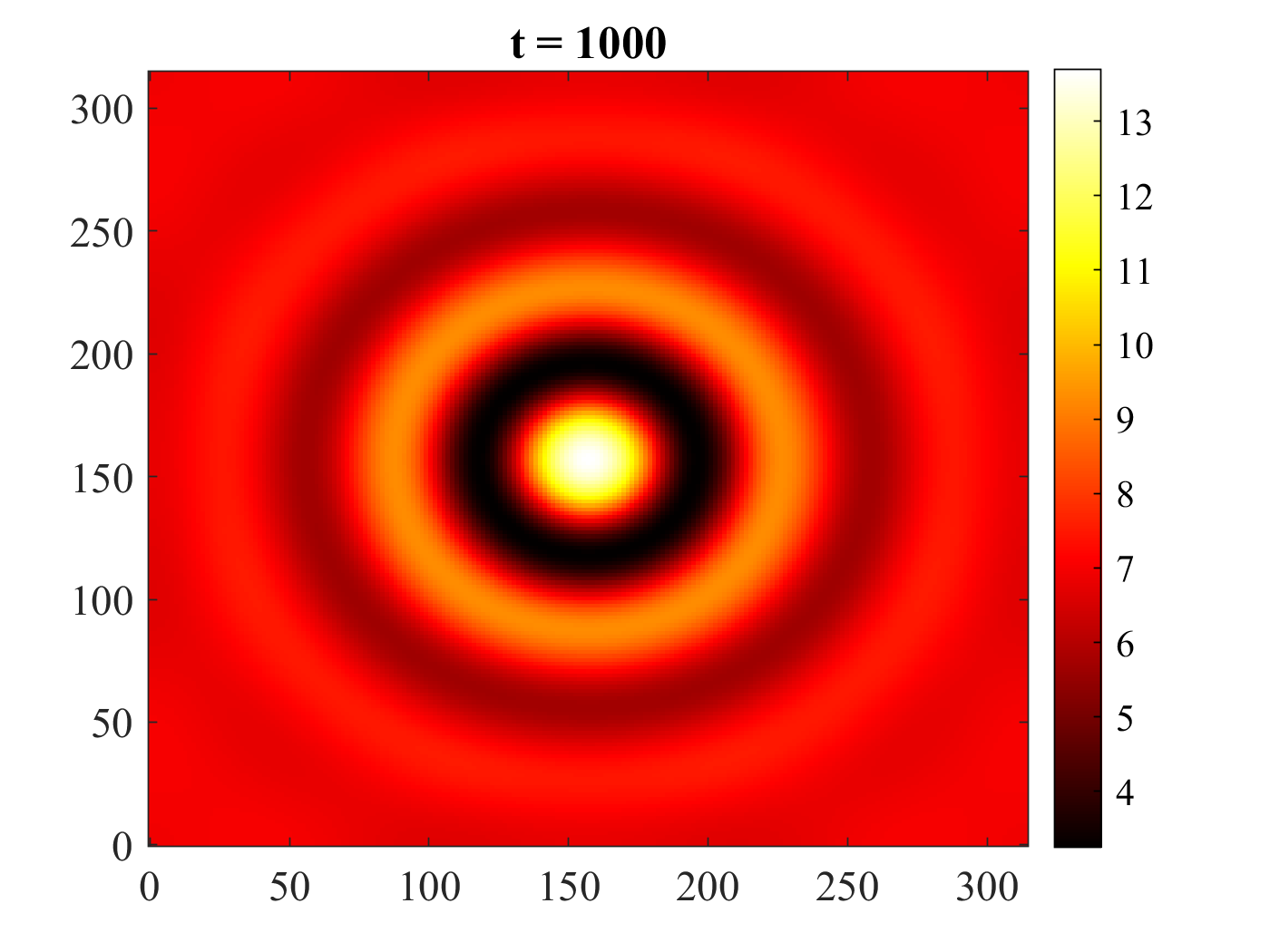}\label{fig_super1:subfig2}}
  \\
  \subfloat[]{\includegraphics[width=0.45\textwidth]{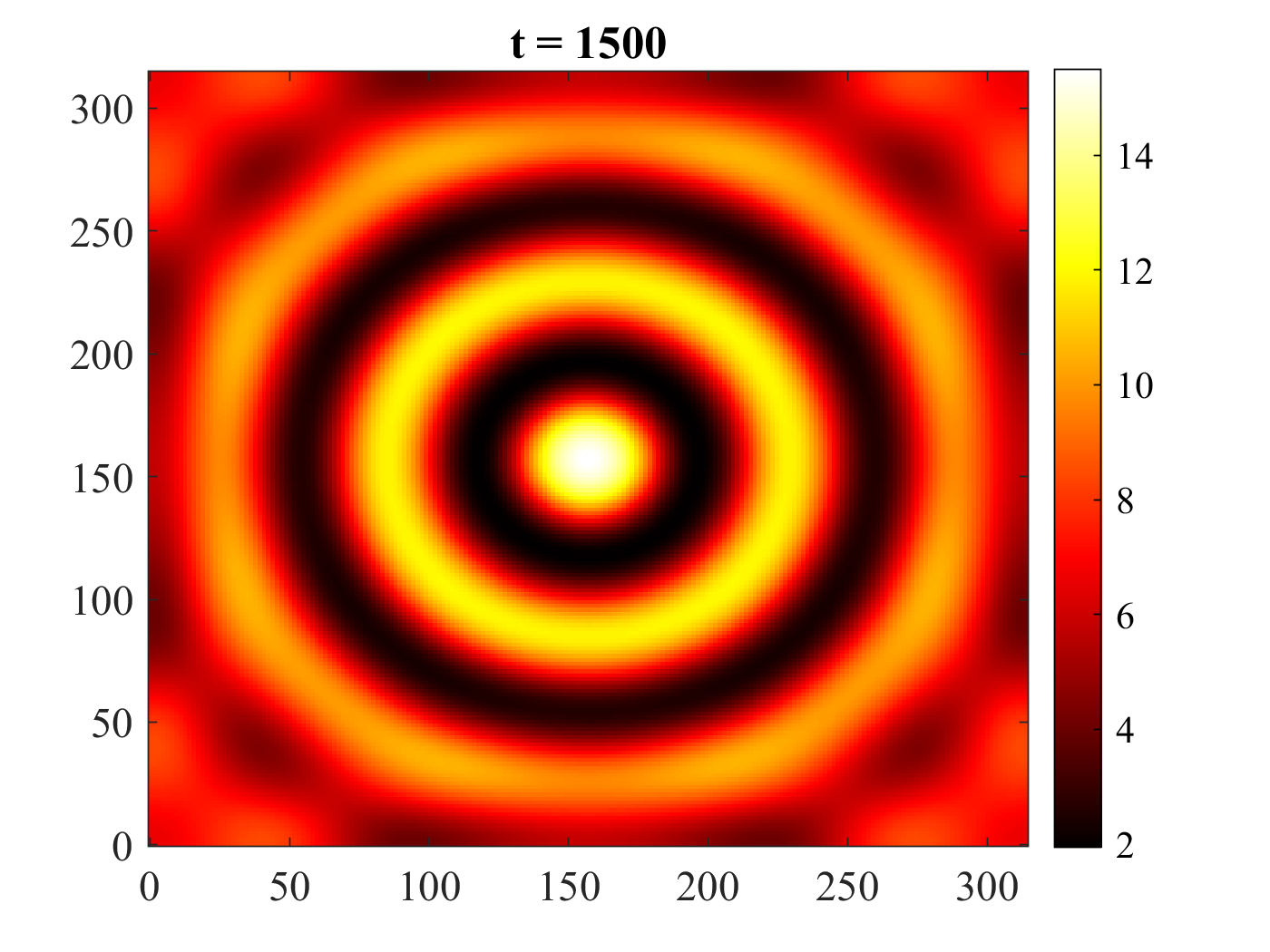}\label{fig_super1:subfig3}}
  \hfill
  \subfloat[]{\includegraphics[width=0.45\textwidth]{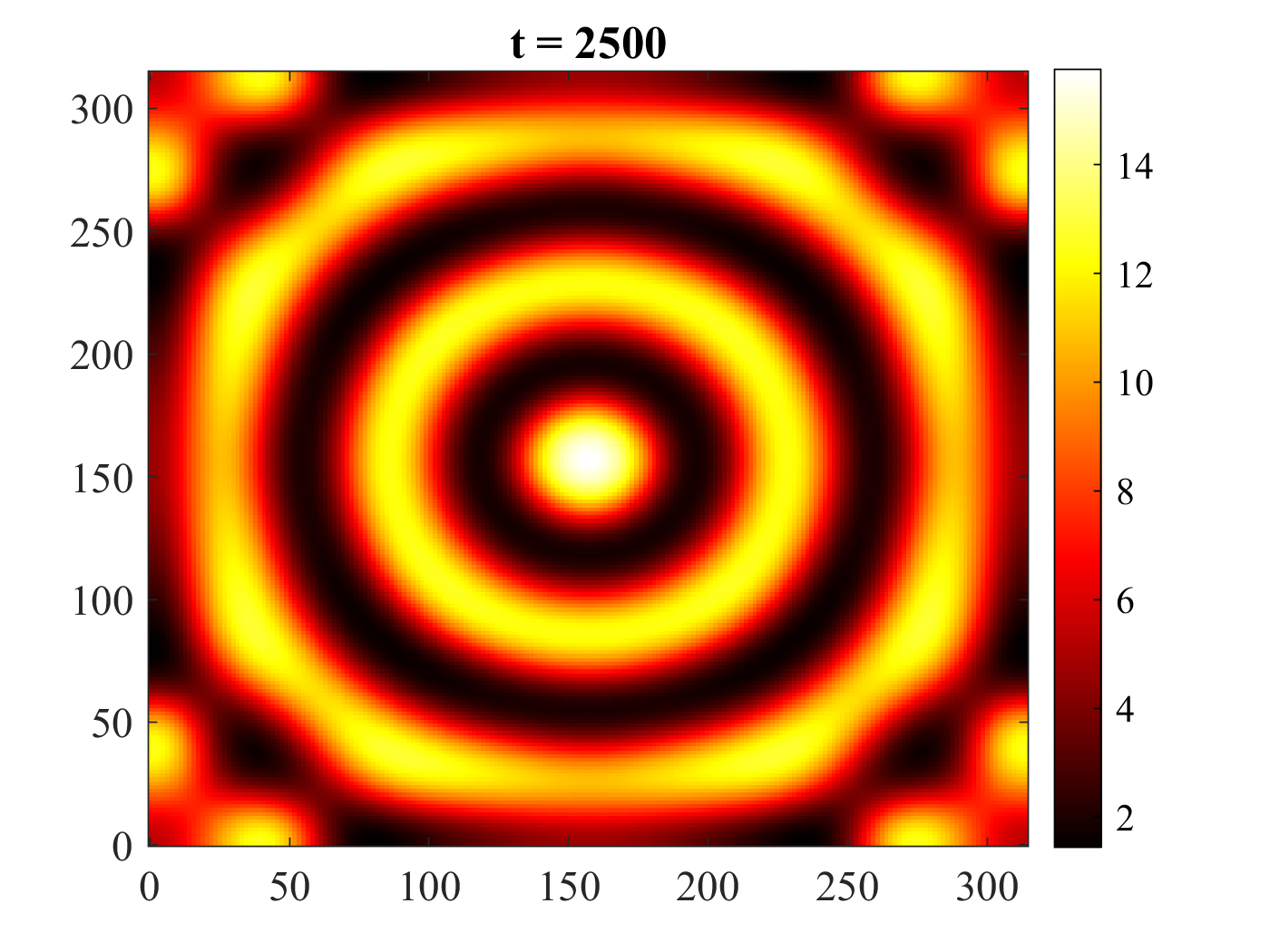}\label{fig_super1:subfig4}}
  \caption{Temporal evolution of supercritical pattern formation for $d=27$. Parameters: $\tau=3.10$, $k_+=0$, $n_n = 2.63$, $n_l = 1.09$, $\gamma_n = 0.02$, $\gamma_l = 0.06$, and $\nu = 0.62$.}
  \label{fig:supercritical_d27}
\end{figure*}

\section{Numerical simulations}
\label{sec6}
For the numerical simulations, we perform the numerical integration of the system using the Crank-Nicolson method, varying the parameters $\tau$ and $k_+$ and applying zero-flux boundary conditions. We focus on the numerical integration around the positive steady state $E_2$.

\subsection{Supercritical and subcritical pattern formation: comparison with WNL analysis}
\label{sec6_1}
We validate our weakly nonlinear analysis by comparing theoretical predictions with full numerical simulations in both supercritical and subcritical regimes. We present two representative examples that show the comparison between the numerical simulation of the system and the analytical approach.

\subsubsection{Supercritical bifurcation example}

For the supercritical regime, we choose parameters $(\tau,k_{+})=(3.10,0)$, which gives the homogeneous steady state $E_2=(6.95,9.05)$. We consider a square spatial domain $\Omega=[0,100\pi] \times [0,100\pi]$ with deviation $\varepsilon=0.2$. This yields a critical wavenumber $\bar{k}_c^2=0.0128$, where condition \eqref{wave2} is satisfied by the unique pair $(z_1,s_1)=(8,8)$.

\subsubsection{Subcritical bifurcation example}

We examine single-mode pattern formation in the subcritical regime to validate our weakly nonlinear analysis against numerical simulations. We choose parameters $(\tau,k_{+})=(2.5,0)$ in the subcritical region, which results in the equilibrium steady state $E_2=(10.66,11.18)$.

For this analysis, we consider a spatial domain $\Omega=[0, 90\pi] \times [0, 50\pi]$ with deviation $\varepsilon=0.3$. This yields to the critical wavenumber $\bar{k}_c^2=0.0104$, with condition \eqref{wave2} satisfied by the single pair $(z_1,s_1)=(9,1)$.

According to our weakly nonlinear analysis, we obtain the Landau coefficient $L=-1.46 \times 10^{-5} < 0$ from equation \eqref{landeau_2_positive}. Since $L < 0$, the cubic Stuart-Landau equation fails to describe the amplitude behavior near the bifurcation, requiring extension to fifth-order analysis as described in Section \ref{sec4}.

Following the quintic analysis, we calculate the coefficients $\bar{\sigma}=0.008$, $\bar{L}=-3.72 \times 10^{-5}$, and $\bar{R}=-6.3 \times 10^{-8}$ from equation \eqref{S_L_Q_2D_1}. This provides a concrete example of subcritical bifurcation where the equilibrium amplitude is determined by equation \eqref{S_L_Q_2D_1}.

To enable comparison between the weakly nonlinear approximation and the full numerical simulations for both examples, we initialize the system near the homogeneous steady state $E_2 = (n^*, l^*)$ with a spatially periodic perturbation. The initial conditions are given by
\begin{align}
n(x, y, 0) &= n^* + \delta \cos(0.1x) \cos(0.1y), \\
l(x, y, 0) &= l^* + \delta \cos(0.1x) \cos(0.1y),
\end{align}
where $\delta$ is a random perturbation with maximum amplitude of 20\% of $E_2$. This choice of initial conditions allows for a direct comparison between the amplitude of the emerging patterns in the numerical simulation and that predicted by the weakly nonlinear analysis.

The comparison between numerical simulation results and the first-order weakly nonlinear approximation demonstrates that the theoretical amplitude prediction shows relatively a good agreement with the numerical simulation in both regimes (Figures \ref{fig_squares} and \ref{fig:single_mode_subcritical}), validating our analytical approach.

\begin{figure*}[!htbp]
  \centering
  \subfloat[]{\includegraphics[width=0.45\textwidth]{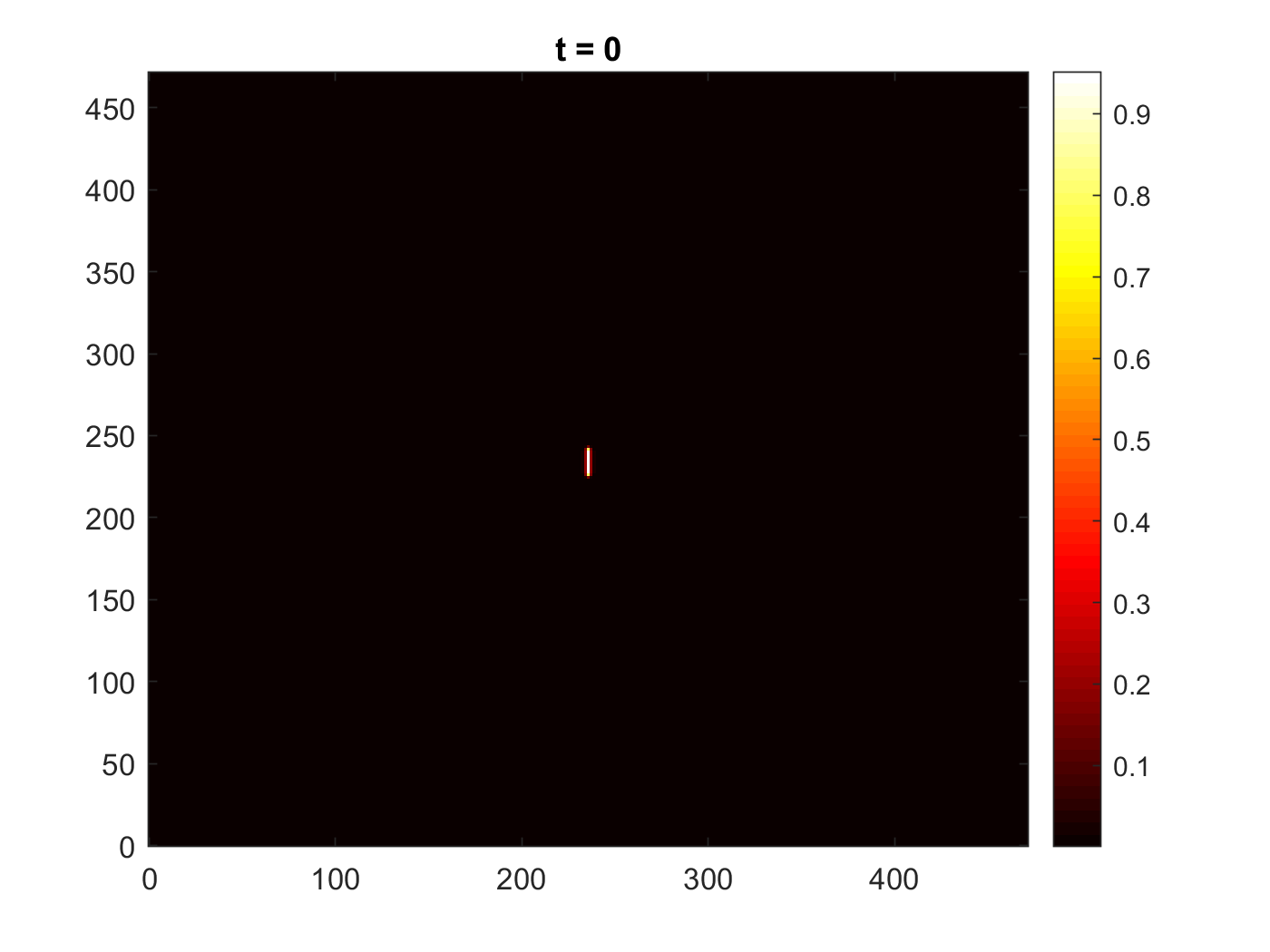}\label{fig2:subfig1}}
  \hfill
  \subfloat[]{\includegraphics[width=0.45\textwidth]{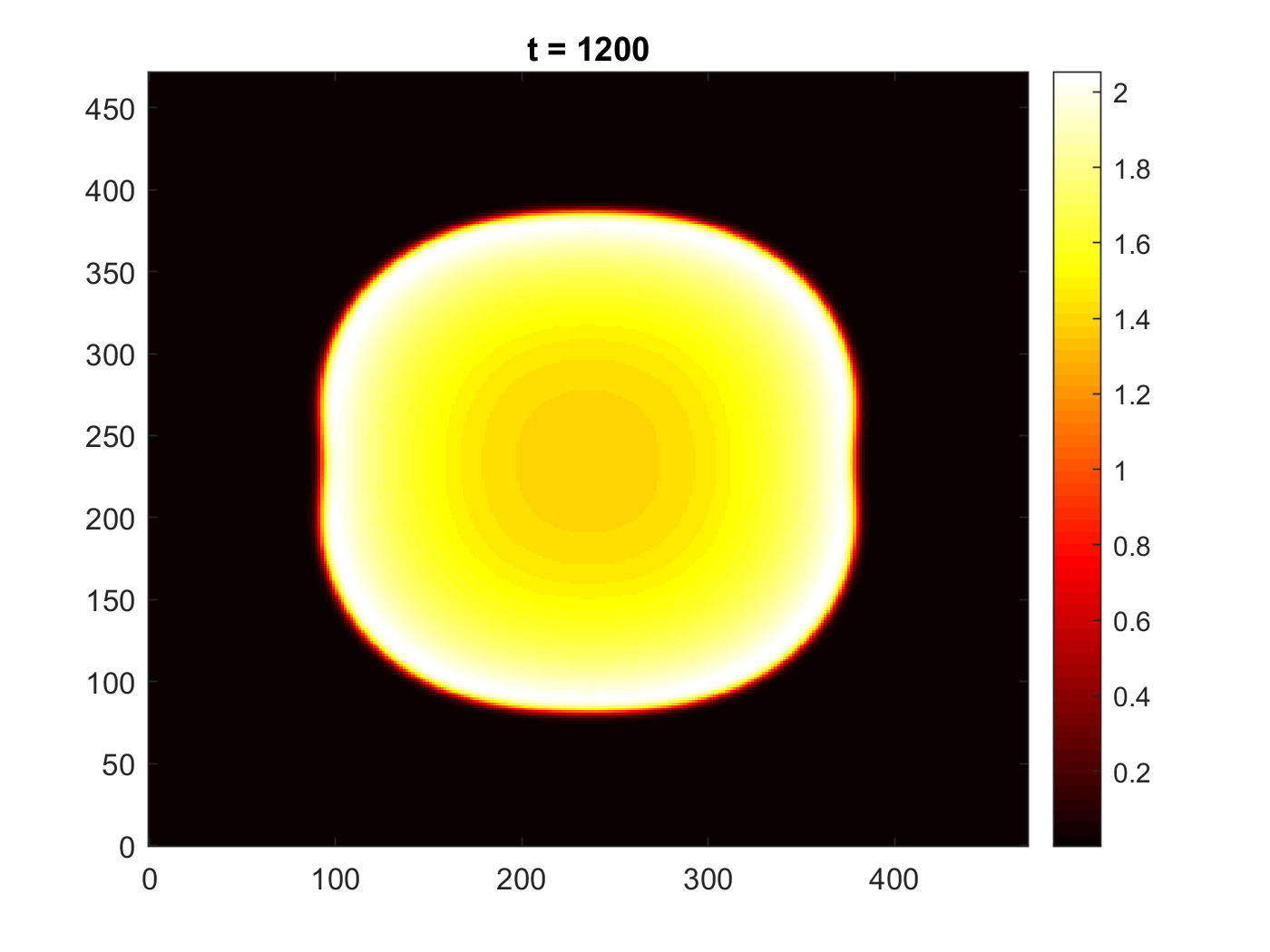}\label{fig2:subfig2}}
  \\
  \subfloat[]{\includegraphics[width=0.45\textwidth]{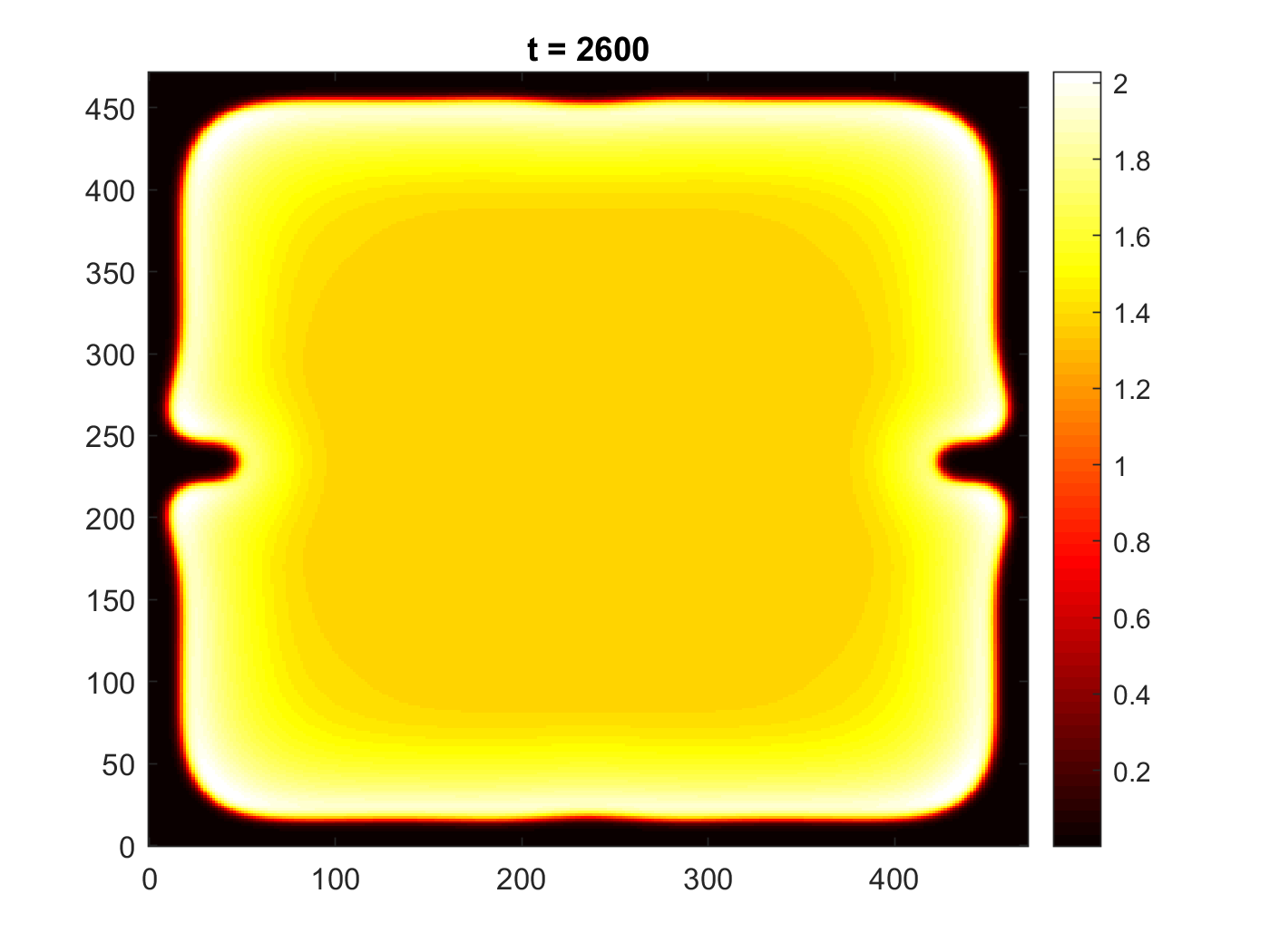}\label{fig2:subfig3}}
  \hfill
  \subfloat[]{\includegraphics[width=0.45\textwidth]{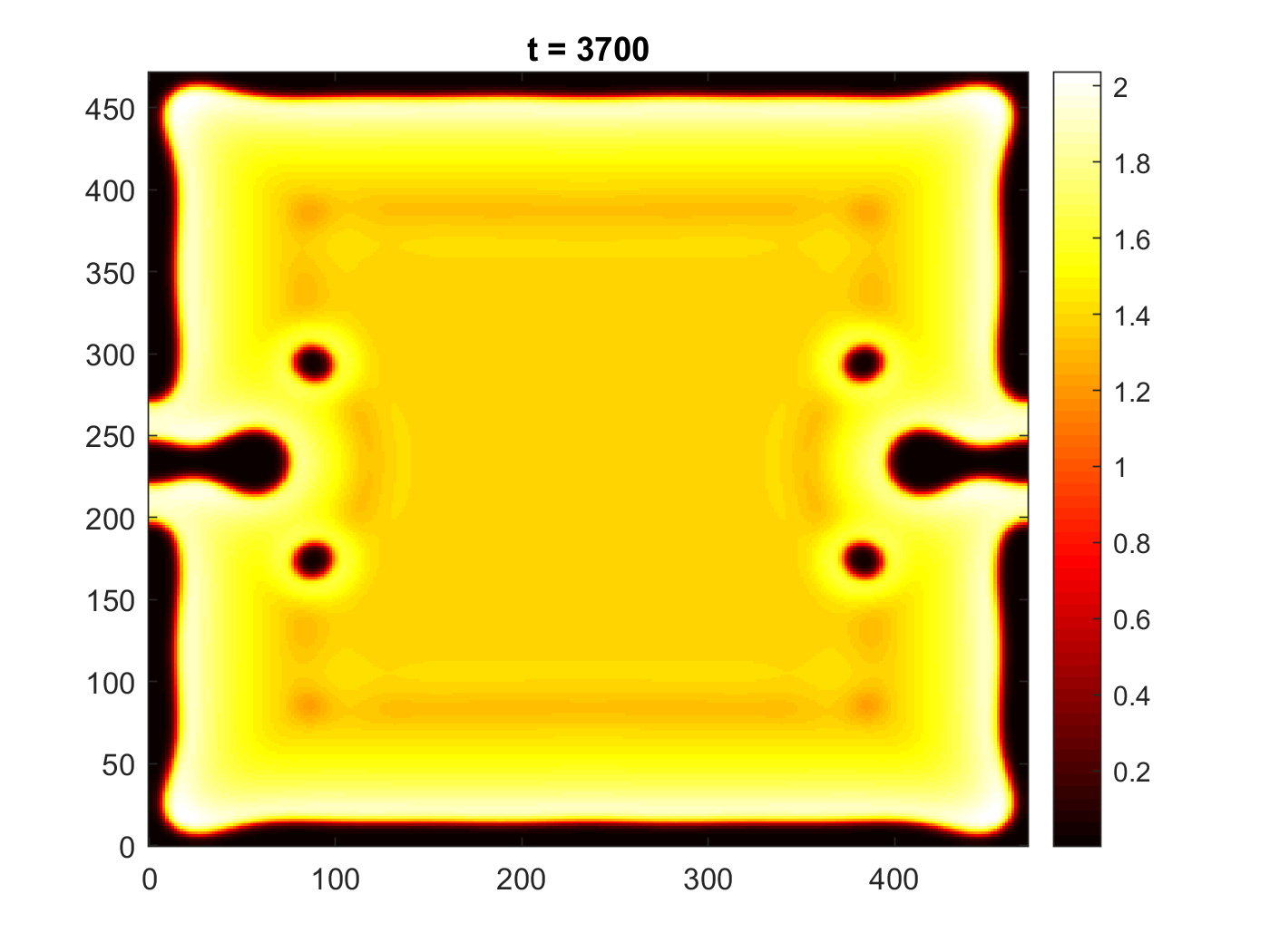}\label{fig2:subfig4}}
  \\
  \subfloat[]{\includegraphics[width=0.45\textwidth]{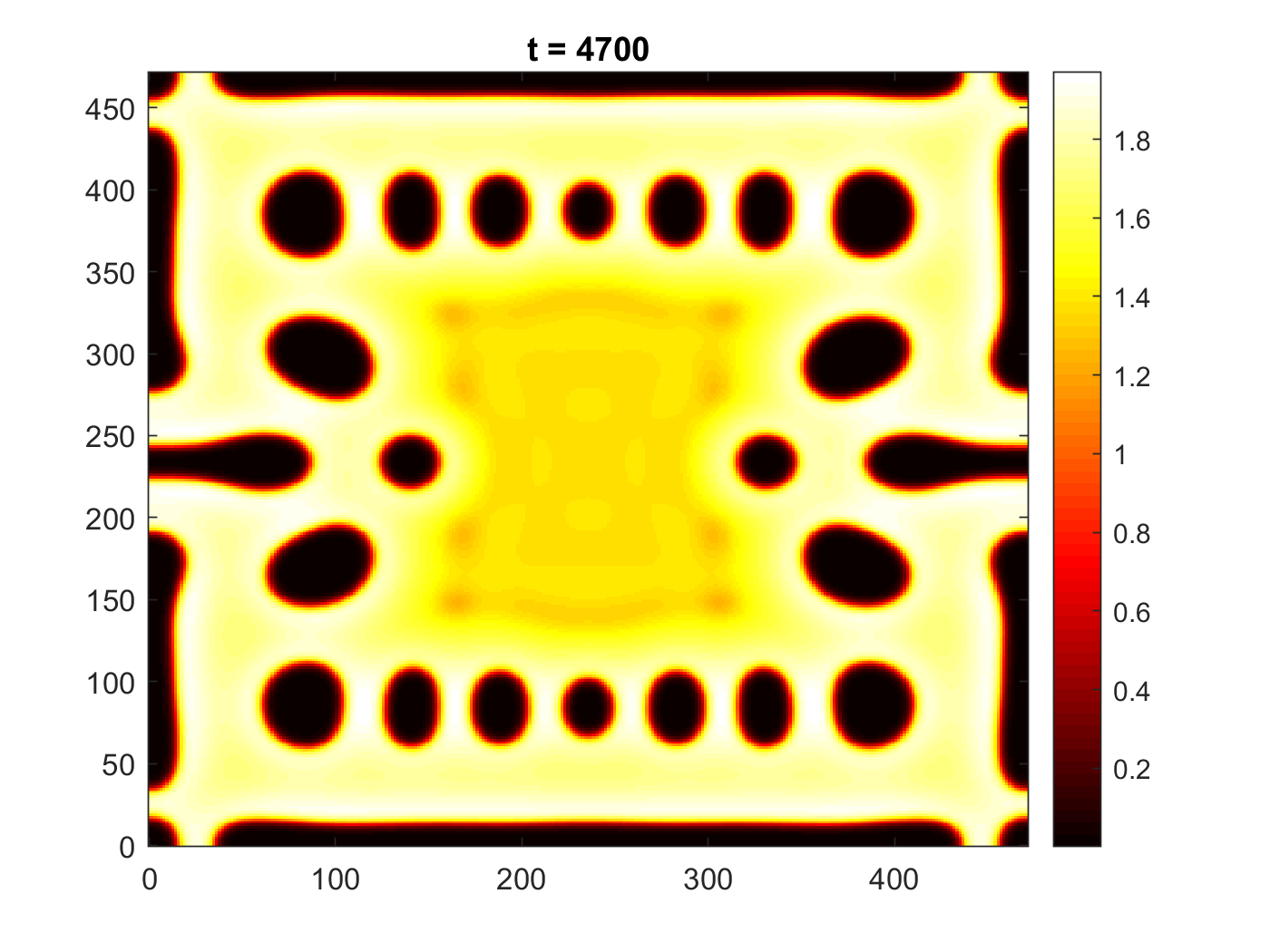}\label{fig2:subfig5}}
  \hfill
  \subfloat[]{\includegraphics[width=0.45\textwidth]{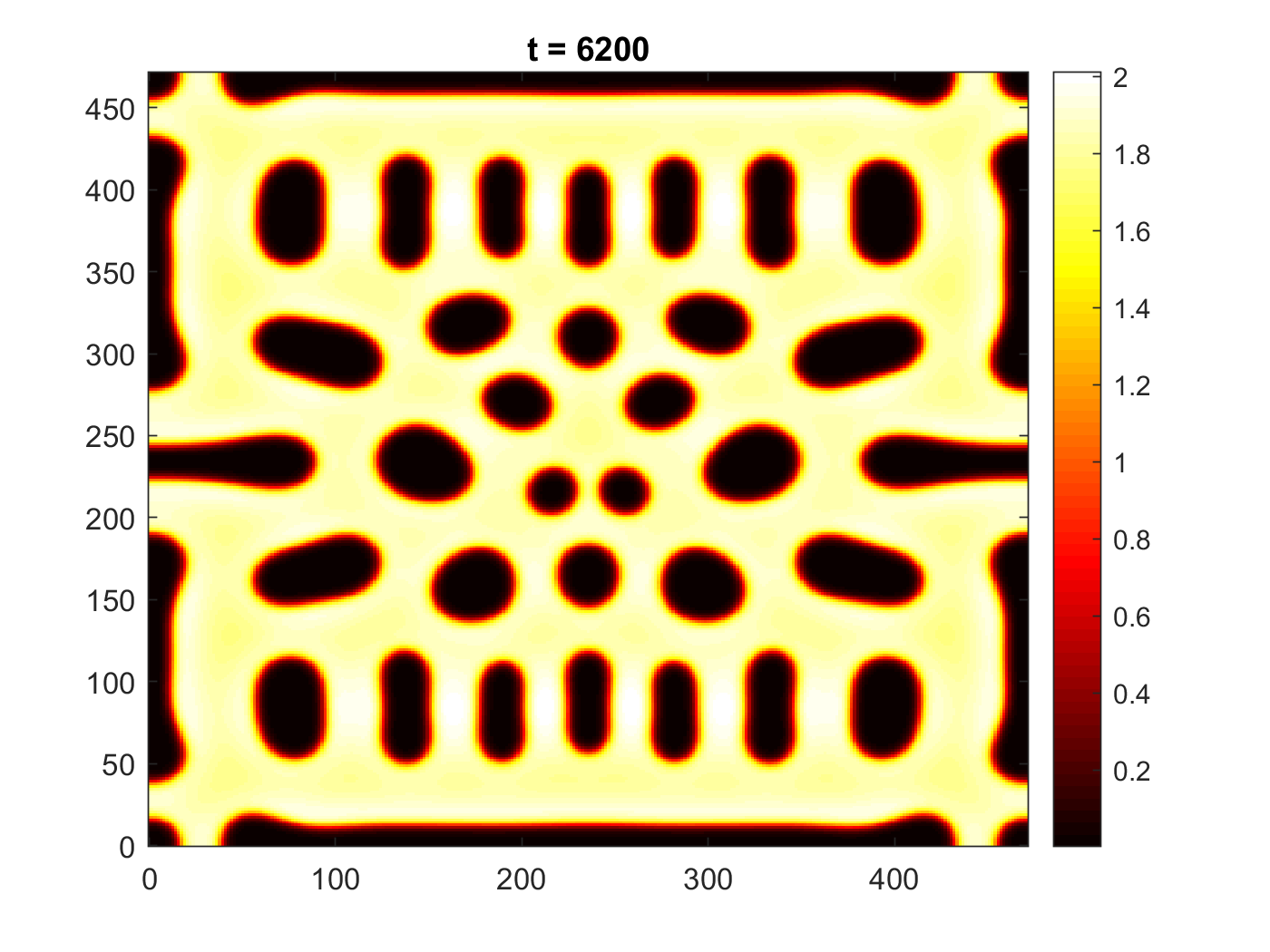}\label{fig2:subfig6}}
  \caption{ Snapshots from a 2D spatiotemporal simulation depict the concentration of Nodal with reversed spots. The initial structure is a small perturbation around the equilibrium state $E_2$ at the center of the spatial domain. The parameters are as follows: $\tau=0.63$, $k_{+}=1.04$, $n_n=2.63$, $n_l=1.09$, $\gamma_n=0.02$, $\gamma_l=0.06$ and $\nu=0.62$.}
  \label{fig:subfigures5}
\end{figure*}

\begin{figure*}[!htbp]
  \centering
  \subfloat[]{\includegraphics[width=0.45\textwidth]{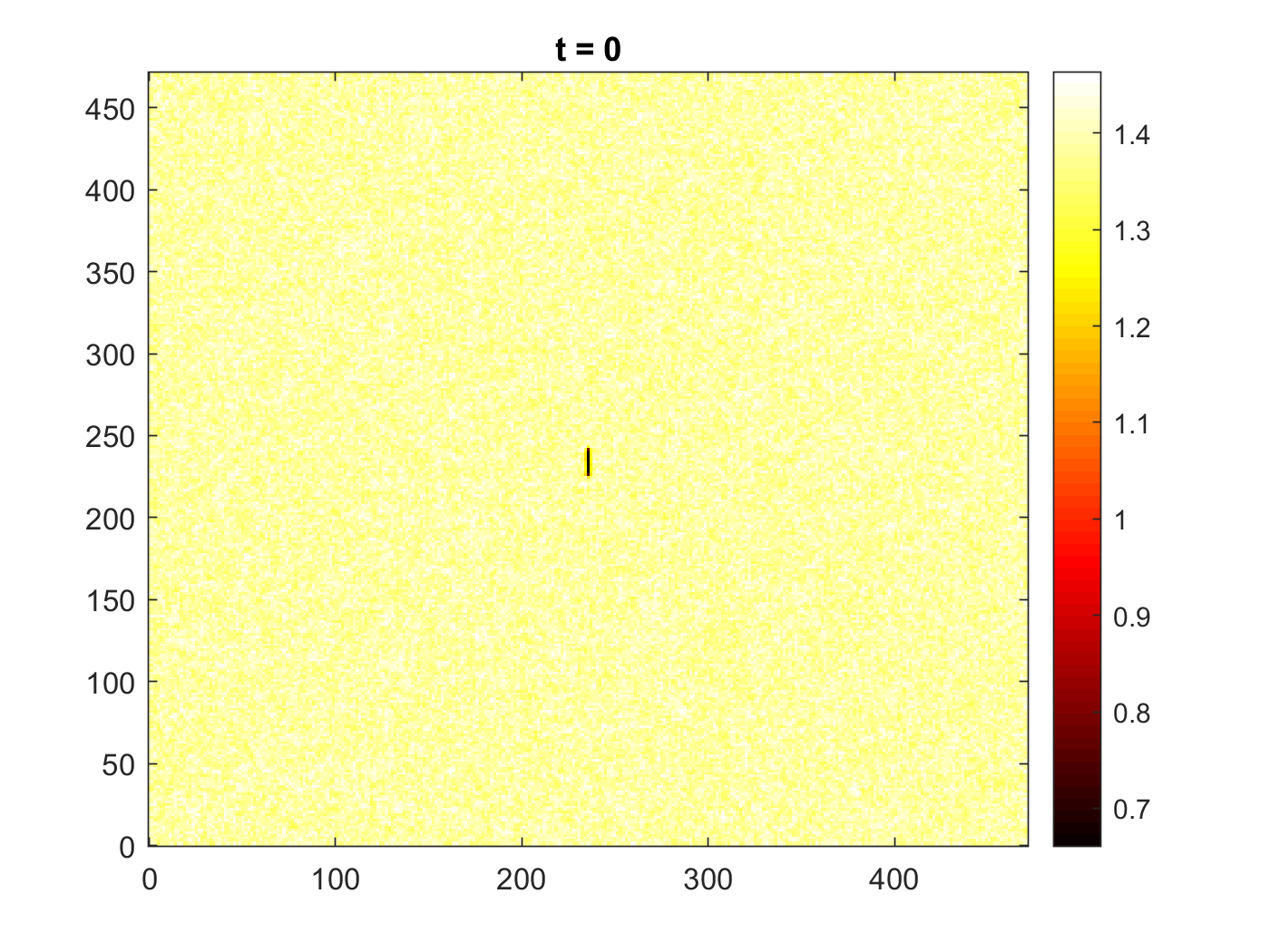}\label{fig3:subfig1}}
  \hfill
  \subfloat[]{\includegraphics[width=0.45\textwidth]{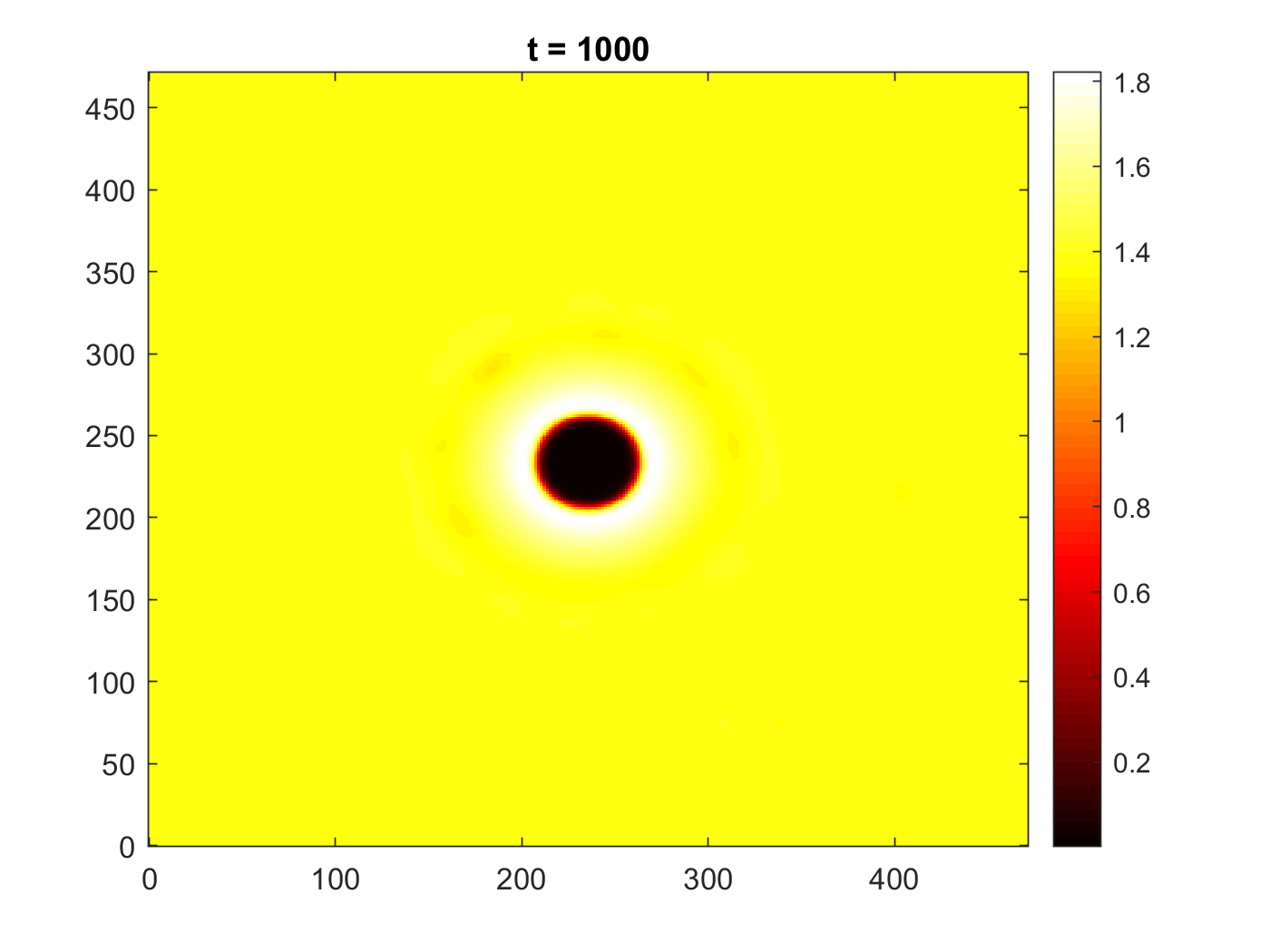}\label{fig3:subfig2}}
  \\
  \subfloat[]{\includegraphics[width=0.45\textwidth]{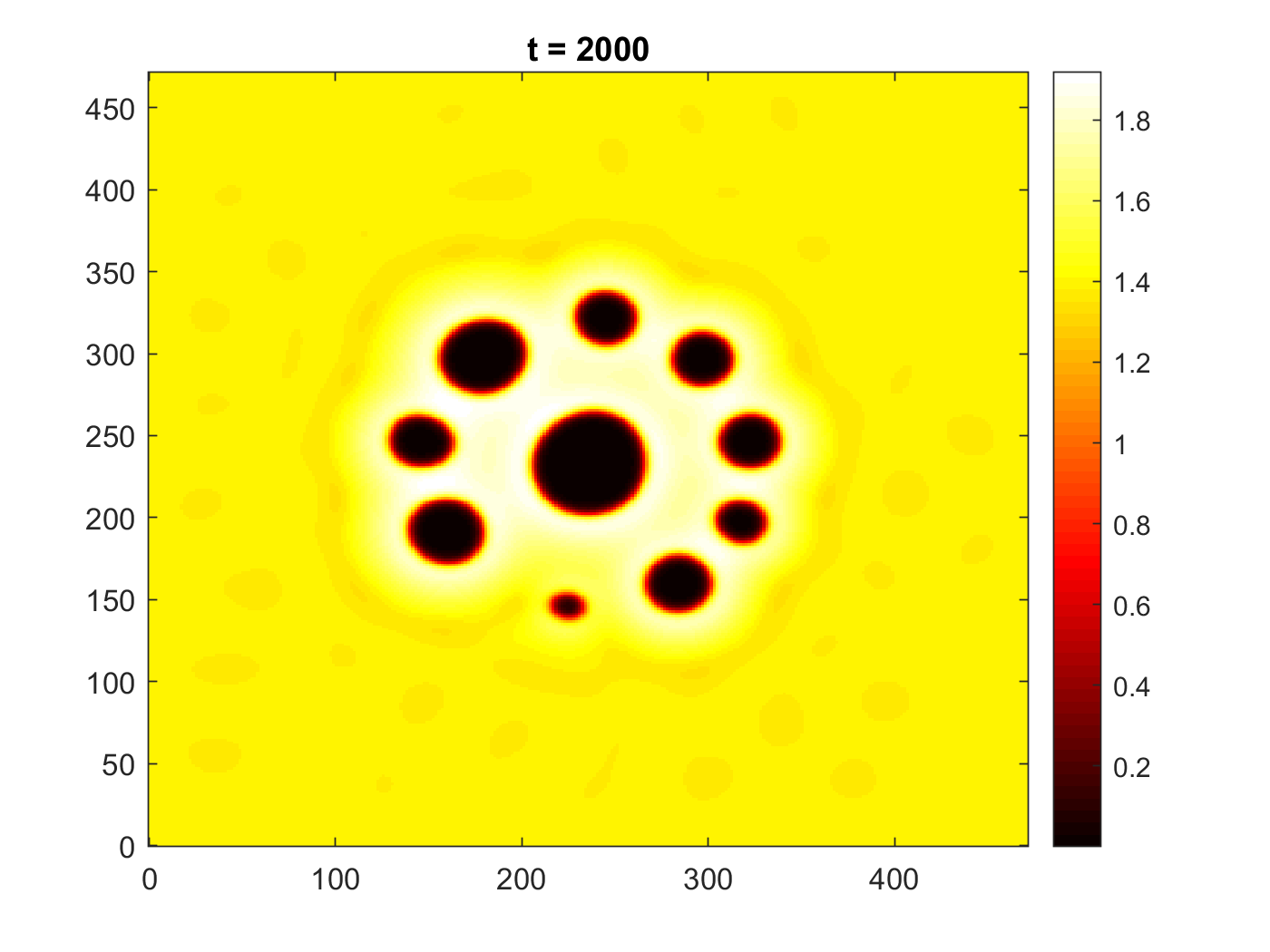}\label{fig3:subfig3}}
  \hfill
  \subfloat[]{\includegraphics[width=0.45\textwidth]{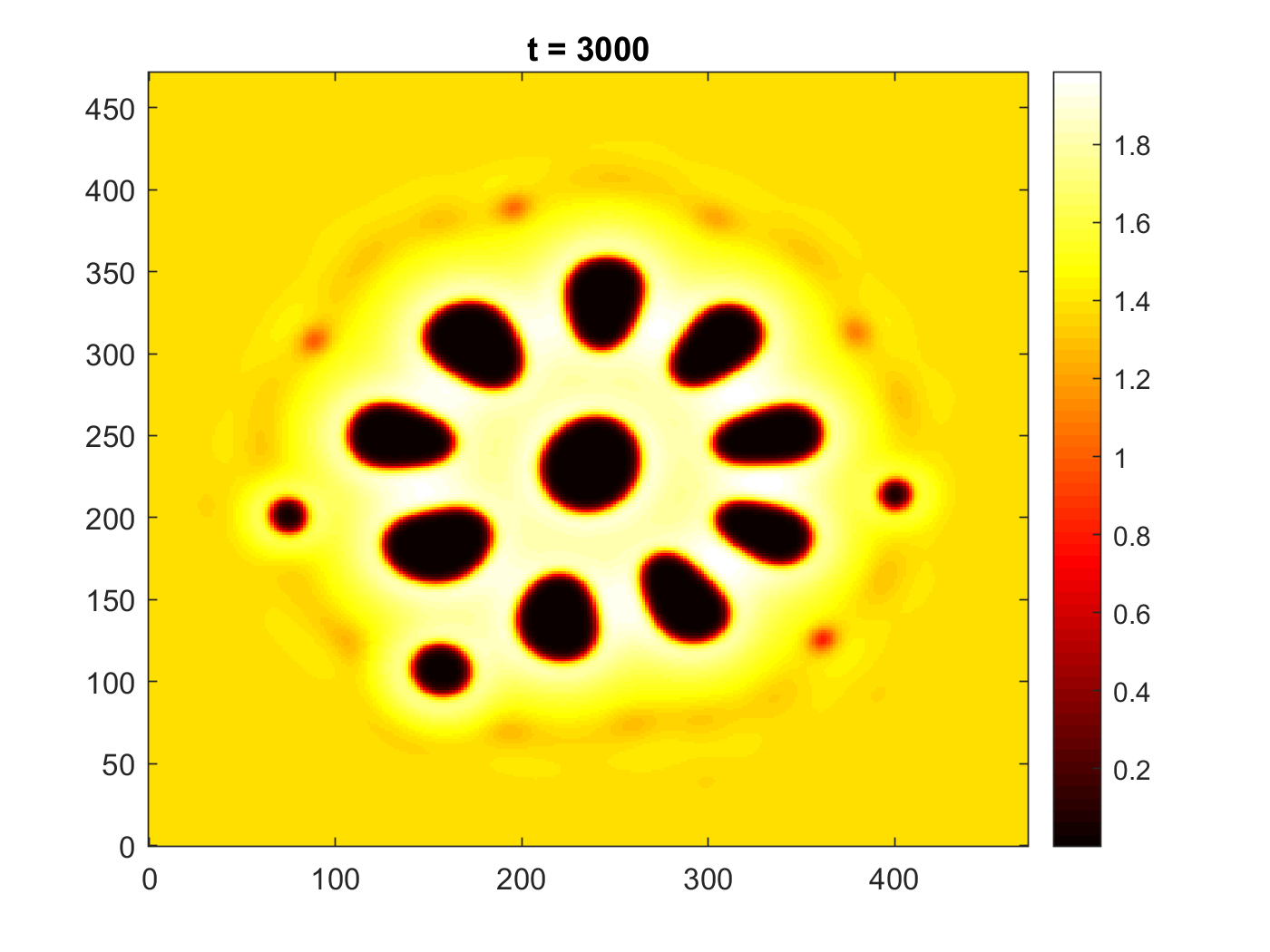}\label{fig3:subfig4}}
  \\
  \subfloat[]{\includegraphics[width=0.45\textwidth]{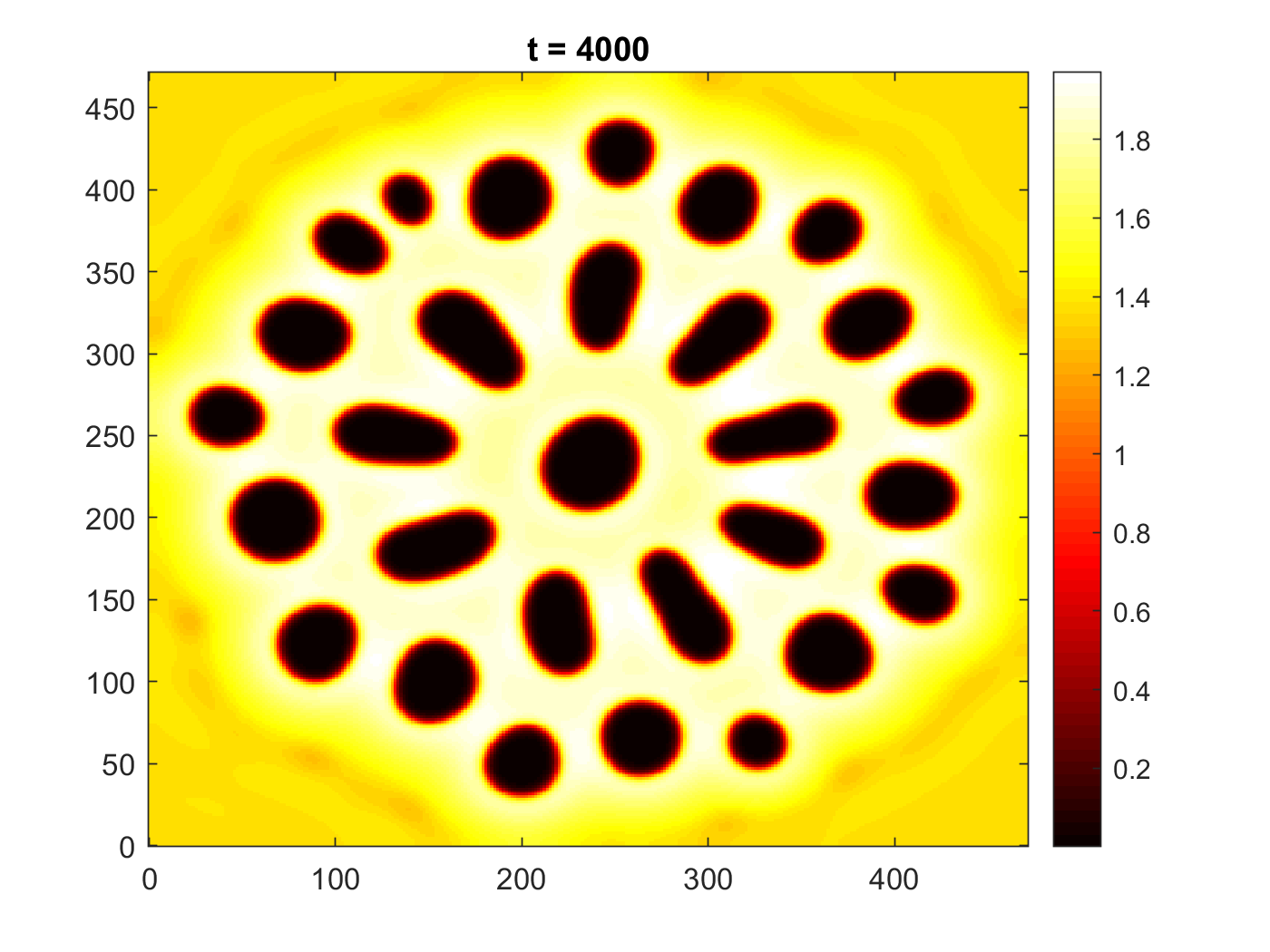}\label{fig3:subfig5}}
  \hfill
  \subfloat[]{\includegraphics[width=0.45\textwidth]{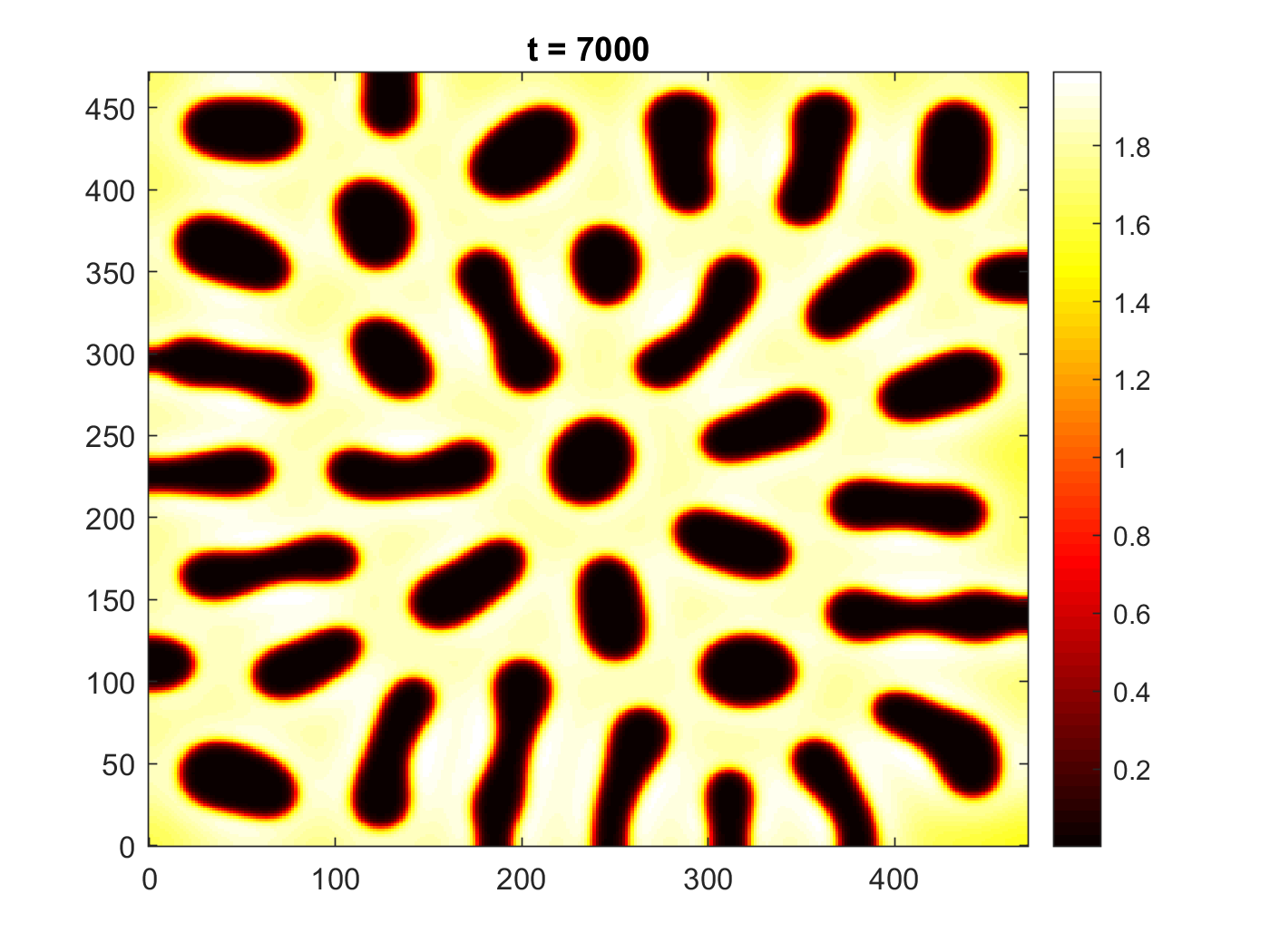}\label{fig3:subfig6}}
  \caption{ Snapshots from space temporal 2D simulation showing the concentration of Nodal and describing the self-replication process. The initial structure is small perturbations around $E_2$ with small concentration chosen in the center of the spatial domain. We used the same parameters as those presented in Figure \ref{fig:subfigures5}.}
  \label{fig:subfigures6}
\end{figure*}

\begin{figure*}[!htbp]
  \centering
  \subfloat[]{\includegraphics[width=0.45\textwidth]{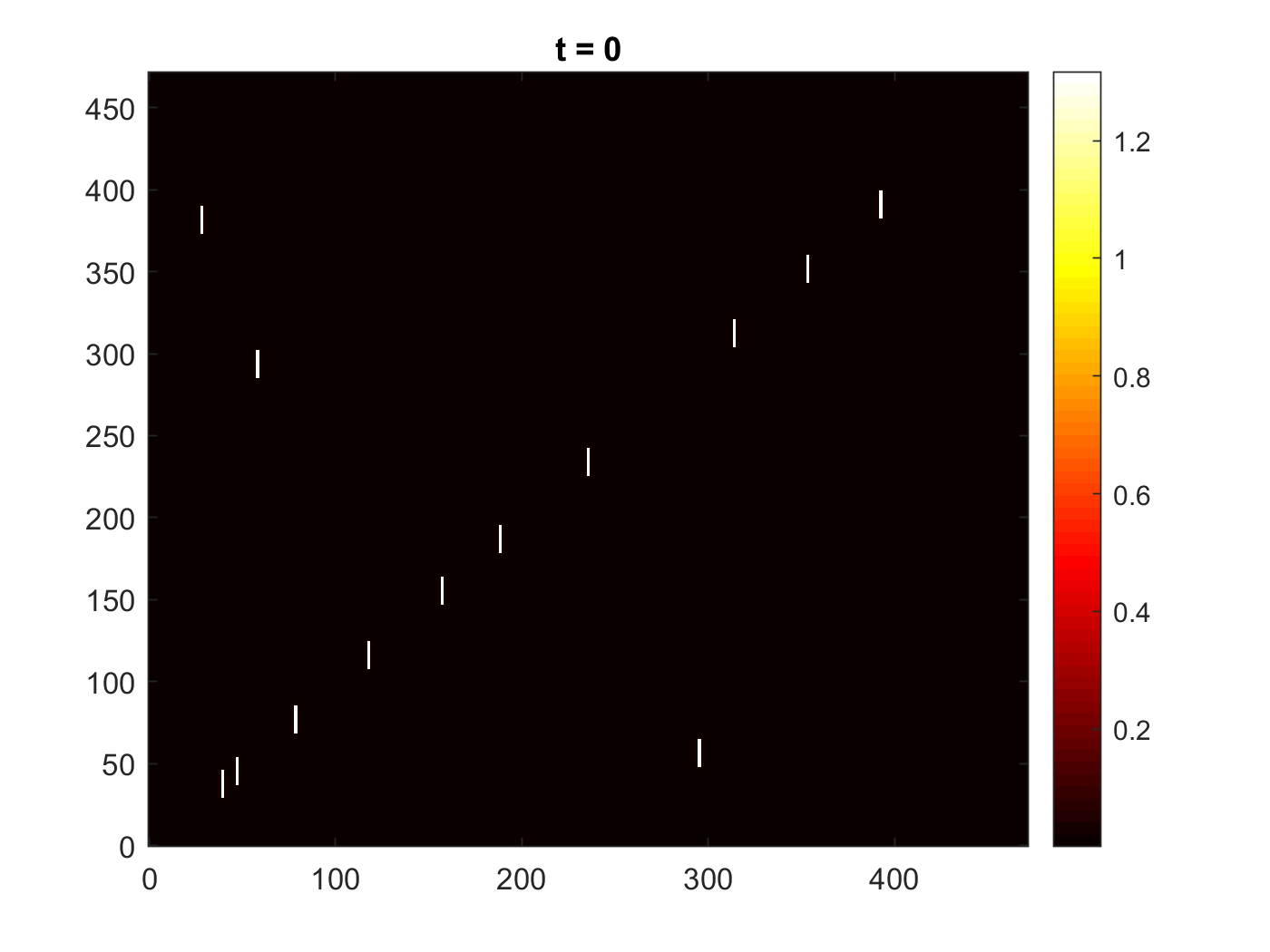}\label{fig4:subfig1}}
  \hfill
  \subfloat[]{\includegraphics[width=0.45\textwidth]{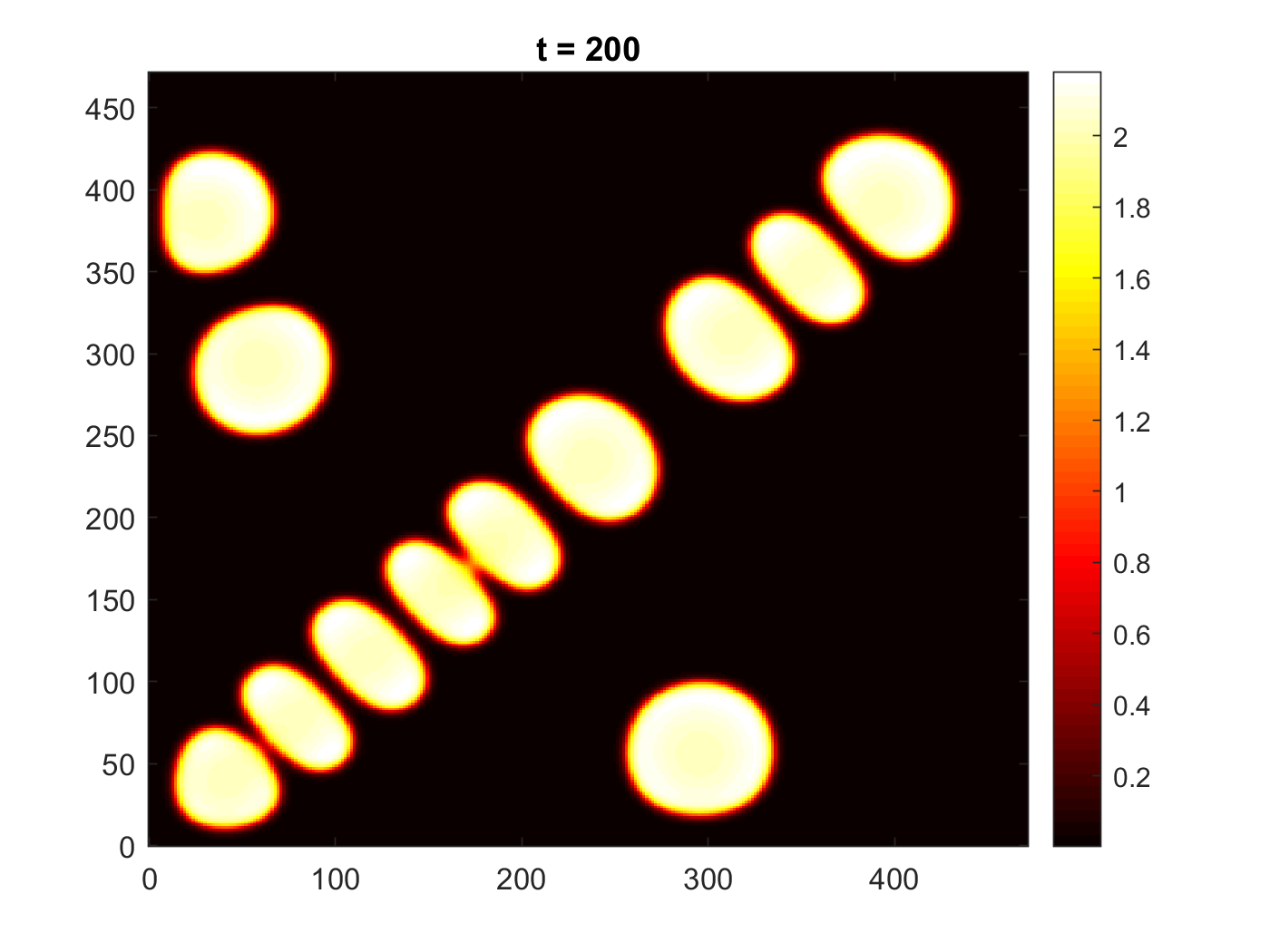}\label{fig4:subfig2}}
  \\
  \subfloat[]{\includegraphics[width=0.45\textwidth]{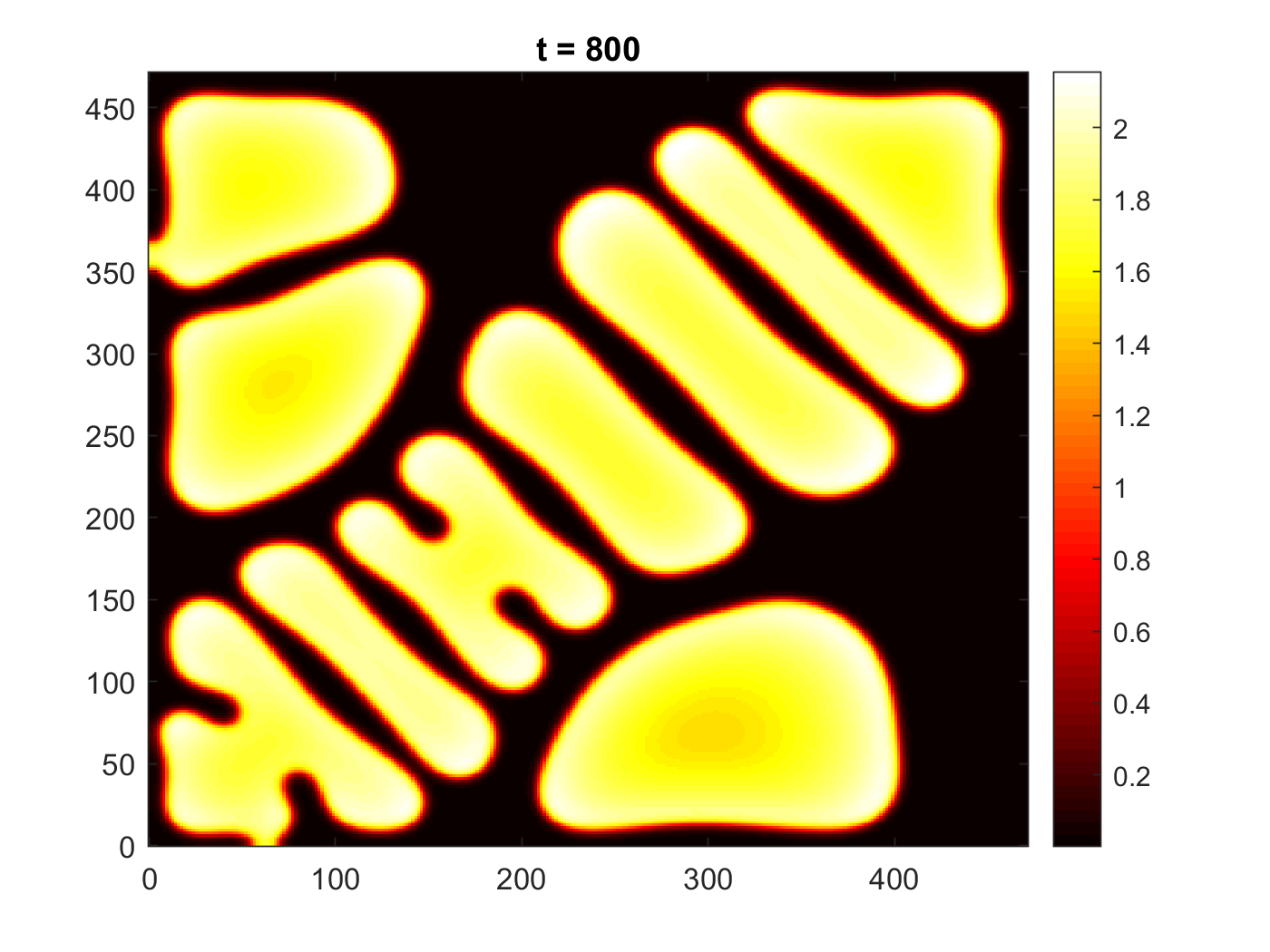}\label{fig4:subfig3}}
  \hfill
  \subfloat[]{\includegraphics[width=0.45\textwidth]{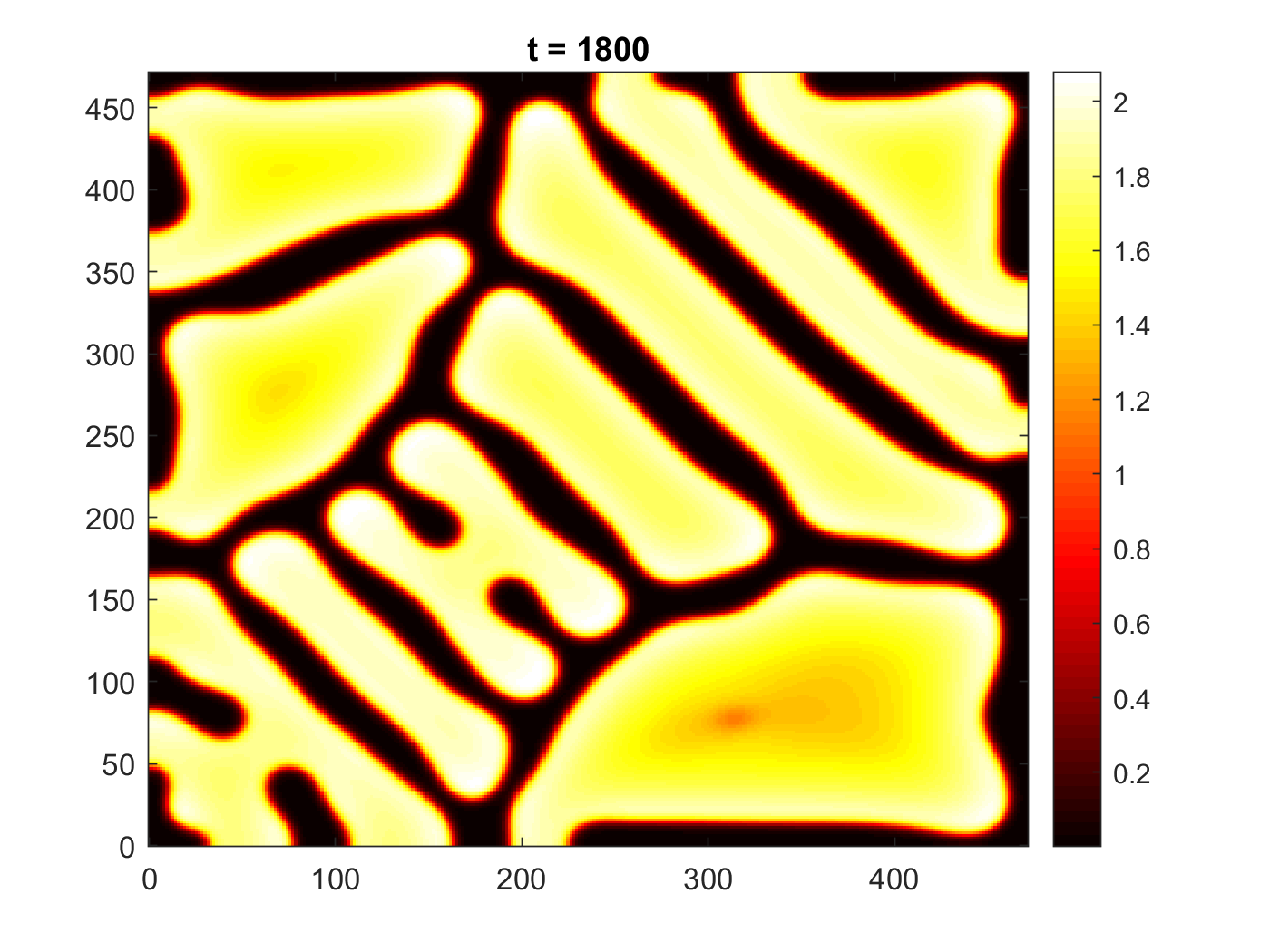}\label{fig4:subfig4}}
  \\
  \subfloat[]{\includegraphics[width=0.45\textwidth]{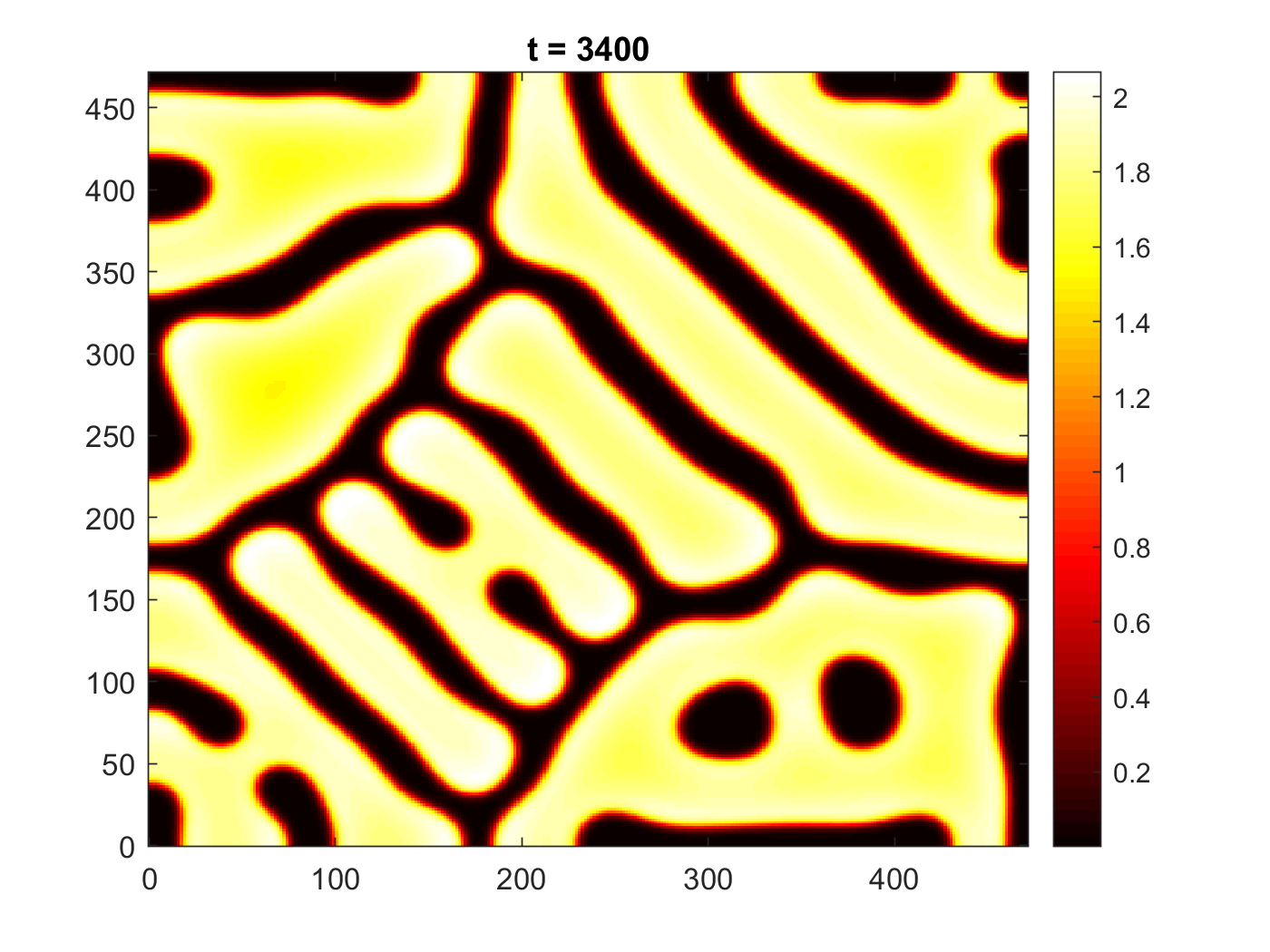}\label{fig4:subfig5}}
  \hfill
  \subfloat[]{\includegraphics[width=0.45\textwidth]{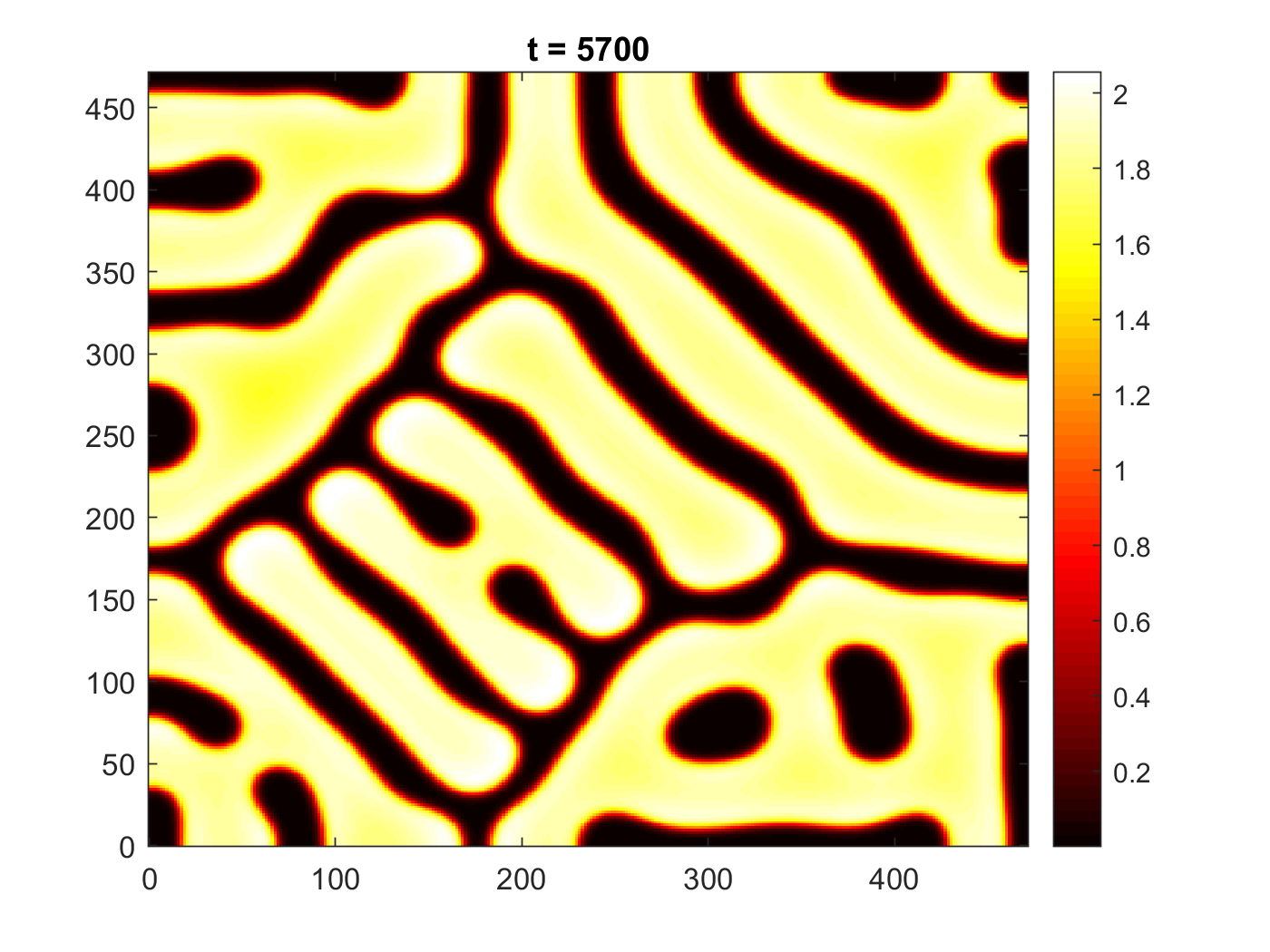}\label{fig4:subfig6}}
  \caption{ Snapshots from space temporal 2D simulation showing the concentration of Nodal and describing the initial perturbations cover the whole domain with time. The initial structure is a set of perturbations of $E_2$ chosen randomly in the spatial domain. We used the same parameters as those presented in Figure \ref{fig:subfigures5}.}
  \label{fig:subfigures7}
\end{figure*}

\subsection{Pattern expansion in the supercritical regime}

\label{sec6_2}

We examine pattern expansion dynamics in the supercritical regime using a localized perturbation at the center of the spatial domain. We consider two diffusion values: $d = 19.5$ and $d = 27$, with parameters $(\tau, k_+) = (3.10, 0)$ and critical diffusion $d_c = 18.34$.

The initial condition consists of a localized perturbation around the steady state $E_2 = (6.95, 9.05)$ placed at the domain center, with the remainder at the uniform steady state.

Figures \ref{fig:supercritical_d19} and \ref{fig:supercritical_d27} show the temporal evolution for both cases. For $d = 19.5$, close to the bifurcation threshold, the perturbation expands radially forming circular stripes that subsequently break into spots. This stripe-to-spot transition is characteristic of supercritical Turing bifurcations and has been observed in the Brusselator model \cite{jensen1994localized}.

For $d = 27$, further from the critical threshold, higher pattern amplitudes develop and the stripe structures remain stable without fragmenting into spots. Note the significant difference in time scales between the two cases, with $d = 27$ exhibiting much faster pattern formation due to the larger distance from the critical value. This amplitude-dependent behavior follows the theoretical scaling $A_{1\infty} = \sqrt{\sigma/L}$ where $\sigma \propto (d-d_c)$. The enhanced amplitudes at larger $d$ values stabilize the stripe morphology.

This qualitative behavior is different from the subcritical dynamics shown in section \ref{sec6.3}. Additionally, the distance-dependent scaling behavior is not present in subcritical bifurcation.

\subsection{Localised spots in the subcritical regime}
\label{sec6.3}
In this part, we focus on localized structures in the subcritical regime, examining their growth and expansion near the bifurcation threshold. The growth of these localized structures is investigated in detail in the Lengyel-Epstein and Fitzhugh-Nagumo models \cite{jensen1994localized}, \cite{purwins2010dissipative}. These structures can grow by either adding new localized structures to the tail of the initial structure, a process referred to as self-replication, or by elongating, deforming, and rupturing the initial structure, leading to space-filling and non-branching patterns, a process known as self-completion \cite{kuznetsov2017pattern}. We specifically concentrate on the numerical simulation of the system in the competitive and direct inhibition scenarios. Particularly in the subcritical region in figure \ref{fig0:subfig4}. We set the parameters $(\tau,k_+)=(0.63,1.04)$ and, starting from three different initial configurations \ref{fig2:subfig1}, \ref{fig3:subfig1} \ref{fig4:subfig1}. The first configuration involves a single spatial perturbation around the equilibrium state $E_2=(1.38,0.36)$ chosen in the center of a low homogeneous concentration domain of Nodal and Lefty. In contrast, the second one consists of random spatial perturbations of $E_2$ with a single low concentration of Nodal and Lefty chosen in the center of the spatial domain. The third configuration consists of a series of perturbations chosen arbitrarily in the spatial domain of homogeneous low concentration of Nodal and Lefty. 

The first and second configurations converge to low-concentration patches located in high-density concentration, referred to as reversed spots. In the first configuration, the perturbation expands in space. Eventually, it stops, leading to spots that self-replicate and fill the domain occupied by the expansion of the initial perturbation. In this process, the perturbation does not reach small regions on the sides of the boundaries, forming stripes parallel to the boundaries \ref{fig:subfigures5}.  The second configuration involves the self-replication of spots \ref{fig:subfigures6}, where the spots near the boundaries tend to align perpendicularly to them. In the third configuration, the perturbations extend and cover the whole domain \ref{fig:subfigures7}. These three examples are among other localized structures observed in the system \eqref{N_L_Nondim} in the subcritical regime by varying the initial conditions and parameters in this region.

\section{Discussion and Conclusion}
\label{sec7}

The field of synthetic developmental biology presents new opportunities for controlling morphogenesis and advancing our understanding of the mechanisms underlying pattern formation. There is a critical need to develop mathematical and computational models capable of abstracting the complexities of engineered reaction-diffusion circuits to fully exploit these advancements. Analysis of these models can provide valuable insights and quantitative predictions that can be experimentally tested and refined. In this paper, we focused on analyzing the morphogenetic model of Nodal and Lefty, representing the first engineered mammalian reaction-diffusion circuit. Going beyond linear analysis predictions, we employ a two-dimensional weakly nonlinear analysis to describe the behavior of the system near the bifurcation threshold. Close to the Turing bifurcation threshold, we distinguish between supercritical and subcritical bifurcations, each exhibiting distinct qualitative behavior in amplitude. We identified the regions of each bifurcation in the parameter space. Our computations reveal that increasing the association rate of Nodal and Lefty $k_+$ narrows the Turing space while promoting subcritical pattern formation, as illustrated in figure \ref{fig0:subfigures4}. Conversely, other results suggest that increasing cooperativity in this type of morphogenetic model expands the Turing space \cite{diambra2015cooperativity}.

The weakly nonlinear analysis method demonstrates good qualitative and quantitative agreement with numerical simulations in both regimes. For supercritical bifurcations with slight deviation from the threshold, our predictions match simulations well. We further validated our approach in the subcritical regime through single-mode pattern formation, showing a relatively good agreement between quintic Stuart-Landau predictions and full numerical simulations, confirming the necessity of fifth-order analysis when the cubic Landau coefficient is negative.

Our investigation of pattern expansion dynamics from localized perturbations in the supercritical regime reveals stripe-to-spot fragmentation near the bifurcation threshold, consistent with behavior observed in the Brusselator model \cite{jensen1994localized}. This expansion behavior is qualitatively different from the subcritical regime.
In the subcritical regime, pattern formation can occur beyond the bifurcation threshold, leading to the emergence of localized structures due to bistability between spatially uniform and spatially periodic patterns. An example of these localized structures in the subcritical regime includes the spotty patterns observed experimentally in the CDIMA reaction \cite{davies1998dividing}. Moreover, pattern formation in this regime requires a more significant perturbation to be initiated, suggesting that patterns in the subcritical regime are more robust and resistant to small perturbations in parameter space compared to the typical Turing bifurcation \cite{brena2014subcritical}.  Additionally, the subcritical region exhibits a wide variety of localized structures, such as the spotted patterns in our numerical simulations \ref{fig:subfigures6}, which expand and grow, eventually occupying the entire spatial domain. Notably, the jump behavior of pattern amplitude in the subcritical regime aligns more closely with the pattern mechanism observed in the engineered reaction-diffusion circuit \cite{sekine2018synthetic}, where transitions between low and high concentrations of Nodal served as initial configurations.

Furthermore, transient patterns can arise from the interplay of the saddle-node bifurcation and diffusion-driven instability. These transient patterns can either evolve toward stable Turing patterns or dissipate back to the uniform state. Although these transient patterns have been studied in a similar morphogenetic model \cite{guisoni2022transient}, their precise role in morphogenesis still needs to be determined.

Finally, this study provides the first exploration of nonlinear dynamics in a morphogenetic reaction-diffusion model controlled by a single regulatory function for both morphogens. The findings offer valuable insights that can inform the design of control strategies for pattern formation within the Nodal and Lefty signaling pathways. Future research should investigate the robustness of Turing patterns in these types of morphogenetic models. Methods that maintain the subcriticality of the model while expanding the Turing space would be highly valuable.

\bibliographystyle{elsarticle-harv} 
\bibliography{refs}

\begin{thebibliography}{37}
\expandafter\ifx\csname natexlab\endcsname\relax\def\natexlab#1{#1}\fi
\providecommand{\url}[1]{\texttt{#1}}
\providecommand{\href}[2]{#2}
\providecommand{\path}[1]{#1}
\providecommand{\DOIprefix}{doi:}
\providecommand{\ArXivprefix}{arXiv:}
\providecommand{\URLprefix}{URL: }
\providecommand{\Pubmedprefix}{pmid:}
\providecommand{\doi}[1]{\href{http://dx.doi.org/#1}{\path{#1}}}
\providecommand{\Pubmed}[1]{\href{pmid:#1}{\path{#1}}}
\providecommand{\bibinfo}[2]{#2}
\ifx\xfnm\relax \def\xfnm[#1]{\unskip,\space#1}\fi
%Type = Article
\bibitem[{AMANN(1985)}]{amann1985global}
\bibinfo{author}{AMANN, H.}, \bibinfo{year}{1985}.
\newblock \bibinfo{title}{Global existence for semilinear parabolic systems}.
\newblock \bibinfo{journal}{Journal f{\"u}r die reine und angewandte Mathematik} \bibinfo{volume}{360}, \bibinfo{pages}{47--83}.
%Type = Article
\bibitem[{Becherer et~al.(2009)Becherer, Morozov and van Saarloos}]{becherer2009probing}
\bibinfo{author}{Becherer, P.}, \bibinfo{author}{Morozov, A.N.}, \bibinfo{author}{van Saarloos, W.}, \bibinfo{year}{2009}.
\newblock \bibinfo{title}{Probing a subcritical instability with an amplitude expansion: An exploration of how far one can get}.
\newblock \bibinfo{journal}{Physica D: Nonlinear Phenomena} \bibinfo{volume}{238}, \bibinfo{pages}{1827--1840}.
%Type = Article
\bibitem[{Bre{\~n}a-Medina and Champneys(2014)}]{brena2014subcritical}
\bibinfo{author}{Bre{\~n}a-Medina, V.}, \bibinfo{author}{Champneys, A.}, \bibinfo{year}{2014}.
\newblock \bibinfo{title}{Subcritical turing bifurcation and the morphogenesis of localized patterns}.
\newblock \bibinfo{journal}{Physical Review E} \bibinfo{volume}{90}, \bibinfo{pages}{032923}.
%Type = Article
\bibitem[{Chen and Shen(2004)}]{chen2004two}
\bibinfo{author}{Chen, C.}, \bibinfo{author}{Shen, M.M.}, \bibinfo{year}{2004}.
\newblock \bibinfo{title}{Two modes by which lefty proteins inhibit nodal signaling}.
\newblock \bibinfo{journal}{Current Biology} \bibinfo{volume}{14}, \bibinfo{pages}{618--624}.
%Type = Article
\bibitem[{Cross and Hohenberg(1993)}]{cross1993pattern}
\bibinfo{author}{Cross, M.C.}, \bibinfo{author}{Hohenberg, P.C.}, \bibinfo{year}{1993}.
\newblock \bibinfo{title}{Pattern formation outside of equilibrium}.
\newblock \bibinfo{journal}{Reviews of Modern Physics} \bibinfo{volume}{65}, \bibinfo{pages}{851--1112}.
%Type = Article
\bibitem[{Cruywagen et~al.(1997)Cruywagen, Murray and Maini}]{cruywagen1997biological}
\bibinfo{author}{Cruywagen, G.C.}, \bibinfo{author}{Murray, J.D.}, \bibinfo{author}{Maini, P.K.}, \bibinfo{year}{1997}.
\newblock \bibinfo{title}{Biological pattern formation on two-dimensional spatial domains: a nonlinear bifurcation analysis}.
\newblock \bibinfo{journal}{SIAM Journal on Applied Mathematics} \bibinfo{volume}{57}, \bibinfo{pages}{1485--1509}.
%Type = Article
\bibitem[{Davies et~al.(1998)Davies, Blanchedeau, Dulos and De~Kepper}]{davies1998dividing}
\bibinfo{author}{Davies, P.}, \bibinfo{author}{Blanchedeau, P.}, \bibinfo{author}{Dulos, E.}, \bibinfo{author}{De~Kepper, P.}, \bibinfo{year}{1998}.
\newblock \bibinfo{title}{Dividing blobs, chemical flowers, and patterned islands in a reaction- diffusion system}.
\newblock \bibinfo{journal}{The Journal of Physical Chemistry A} \bibinfo{volume}{102}, \bibinfo{pages}{8236--8244}.
%Type = Article
\bibitem[{Diambra et~al.(2015)Diambra, Senthivel, Menendez and Isalan}]{diambra2015cooperativity}
\bibinfo{author}{Diambra, L.}, \bibinfo{author}{Senthivel, V.R.}, \bibinfo{author}{Menendez, D.B.}, \bibinfo{author}{Isalan, M.}, \bibinfo{year}{2015}.
\newblock \bibinfo{title}{Cooperativity to increase turing pattern space for synthetic biology}.
\newblock \bibinfo{journal}{ACS synthetic biology} \bibinfo{volume}{4}, \bibinfo{pages}{177--186}.
%Type = Article
\bibitem[{Gambino et~al.(2019)Gambino, Lombardo, Rubino and Sammartino}]{gambino2019pattern}
\bibinfo{author}{Gambino, G.}, \bibinfo{author}{Lombardo, M.}, \bibinfo{author}{Rubino, G.}, \bibinfo{author}{Sammartino, M.}, \bibinfo{year}{2019}.
\newblock \bibinfo{title}{Pattern selection in the 2d fitzhugh--nagumo model}.
\newblock \bibinfo{journal}{Ricerche di Matematica} \bibinfo{volume}{68}, \bibinfo{pages}{535--549}.
%Type = Article
\bibitem[{Gierer and Meinhardt(1972)}]{gierer1972theory}
\bibinfo{author}{Gierer, A.}, \bibinfo{author}{Meinhardt, H.}, \bibinfo{year}{1972}.
\newblock \bibinfo{title}{A theory of biological pattern formation}.
\newblock \bibinfo{journal}{Kybernetik} \bibinfo{volume}{12}, \bibinfo{pages}{30--39}.
%Type = Article
\bibitem[{Guisoni and Diambra(2022)}]{guisoni2022transient}
\bibinfo{author}{Guisoni, N.}, \bibinfo{author}{Diambra, L.}, \bibinfo{year}{2022}.
\newblock \bibinfo{title}{Transient turing patterns in a morphogenetic model}.
\newblock \bibinfo{journal}{Frontiers in Physics} \bibinfo{volume}{10}, \bibinfo{pages}{927152}.
%Type = Article
\bibitem[{Gutierrez et~al.(2012)Gutierrez, Monteoliva and Diambra}]{gutierrez2012cooperative}
\bibinfo{author}{Gutierrez, P.S.}, \bibinfo{author}{Monteoliva, D.}, \bibinfo{author}{Diambra, L.}, \bibinfo{year}{2012}.
\newblock \bibinfo{title}{Cooperative binding of transcription factors promotes bimodal gene expression response} .
%Type = Article
\bibitem[{Han and Dai(2017)}]{han2017cross}
\bibinfo{author}{Han, R.}, \bibinfo{author}{Dai, B.}, \bibinfo{year}{2017}.
\newblock \bibinfo{title}{Cross-diffusion-driven turing instability and weakly nonlinear analysis of turing patterns in a uni-directional consumer-resource system}.
\newblock \bibinfo{journal}{Boundary Value Problems} \bibinfo{volume}{2017}, \bibinfo{pages}{125}.
%Type = Book
\bibitem[{Hoyle(2006)}]{hoyle2006pattern}
\bibinfo{author}{Hoyle, R.B.}, \bibinfo{year}{2006}.
\newblock \bibinfo{title}{Pattern formation: an introduction to methods}.
\newblock \bibinfo{publisher}{Cambridge University Press}.
%Type = Article
\bibitem[{Jensen et~al.(1994)Jensen, Pannbacker, Mosekilde, Dewel and Borckmans}]{jensen1994localized}
\bibinfo{author}{Jensen, O.}, \bibinfo{author}{Pannbacker, V.O.}, \bibinfo{author}{Mosekilde, E.}, \bibinfo{author}{Dewel, G.}, \bibinfo{author}{Borckmans, P.}, \bibinfo{year}{1994}.
\newblock \bibinfo{title}{Localized structures and front propagation in the lengyel-epstein model}.
\newblock \bibinfo{journal}{Physical Review E} \bibinfo{volume}{50}, \bibinfo{pages}{736}.
%Type = Article
\bibitem[{Krause et~al.(2024)Krause, Gaffney, Jewell, Klika and Walker}]{Krause2024}
\bibinfo{author}{Krause, A.L.}, \bibinfo{author}{Gaffney, E.A.}, \bibinfo{author}{Jewell, T.J.}, \bibinfo{author}{Klika, V.}, \bibinfo{author}{Walker, B.J.}, \bibinfo{year}{2024}.
\newblock \bibinfo{title}{Turing instabilities are not enough to ensure pattern formation}.
\newblock \bibinfo{journal}{Bull Math Biol} \bibinfo{volume}{86}, \bibinfo{pages}{21}.
%Type = Article
\bibitem[{Kuznetsov et~al.(2017)Kuznetsov, Kolobov and Polezhaev}]{kuznetsov2017pattern}
\bibinfo{author}{Kuznetsov, M.}, \bibinfo{author}{Kolobov, A.}, \bibinfo{author}{Polezhaev, A.}, \bibinfo{year}{2017}.
\newblock \bibinfo{title}{Pattern formation in a reaction-diffusion system of fitzhugh-nagumo type before the onset of subcritical turing bifurcation}.
\newblock \bibinfo{journal}{Physical Review E} \bibinfo{volume}{95}, \bibinfo{pages}{052208}.
%Type = Article
\bibitem[{Liu et~al.(2011)Liu, Fu, Liu, Ren, Chau, Li, Xiang, Zeng, Chen, Tang et~al.}]{liu2011sequential}
\bibinfo{author}{Liu, C.}, \bibinfo{author}{Fu, X.}, \bibinfo{author}{Liu, L.}, \bibinfo{author}{Ren, X.}, \bibinfo{author}{Chau, C.K.}, \bibinfo{author}{Li, S.}, \bibinfo{author}{Xiang, L.}, \bibinfo{author}{Zeng, H.}, \bibinfo{author}{Chen, G.}, \bibinfo{author}{Tang, L.H.}, et~al., \bibinfo{year}{2011}.
\newblock \bibinfo{title}{Sequential establishment of stripe patterns in an expanding cell population}.
\newblock \bibinfo{journal}{Science} \bibinfo{volume}{334}, \bibinfo{pages}{238--241}.
%Type = Article
\bibitem[{Matkowsky(1970)}]{matkowsky1970nonlinear}
\bibinfo{author}{Matkowsky, B.J.}, \bibinfo{year}{1970}.
\newblock \bibinfo{title}{Nonlinear dynamic stability: a formal theory}.
\newblock \bibinfo{journal}{SIAM Journal on Applied Mathematics} \bibinfo{volume}{18}, \bibinfo{pages}{872--883}.
%Type = Article
\bibitem[{Middleton et~al.(2013)Middleton, King and Loose}]{middleton2013wave}
\bibinfo{author}{Middleton, A.}, \bibinfo{author}{King, J.}, \bibinfo{author}{Loose, M.}, \bibinfo{year}{2013}.
\newblock \bibinfo{title}{Wave pinning and spatial patterning in a mathematical model of antivin/lefty--nodal signalling}.
\newblock \bibinfo{journal}{Journal of Mathematical Biology} \bibinfo{volume}{67}, \bibinfo{pages}{1393--1424}.
%Type = Article
\bibitem[{M{\"u}ller et~al.(2012)M{\"u}ller, Rogers, Jordan, Lee, Robson, Ramanathan and Schier}]{muller2012differential}
\bibinfo{author}{M{\"u}ller, P.}, \bibinfo{author}{Rogers, K.W.}, \bibinfo{author}{Jordan, B.M.}, \bibinfo{author}{Lee, J.S.}, \bibinfo{author}{Robson, D.}, \bibinfo{author}{Ramanathan, S.}, \bibinfo{author}{Schier, A.F.}, \bibinfo{year}{2012}.
\newblock \bibinfo{title}{Differential diffusivity of nodal and lefty underlies a reaction-diffusion patterning system}.
\newblock \bibinfo{journal}{Science} \bibinfo{volume}{336}, \bibinfo{pages}{721--724}.
%Type = Article
\bibitem[{Nakamura et~al.(2006)Nakamura, Mine, Nakaguchi, Mochizuki, Yamamoto, Yashiro, Meno and Hamada}]{nakamura2006generation}
\bibinfo{author}{Nakamura, T.}, \bibinfo{author}{Mine, N.}, \bibinfo{author}{Nakaguchi, E.}, \bibinfo{author}{Mochizuki, A.}, \bibinfo{author}{Yamamoto, M.}, \bibinfo{author}{Yashiro, K.}, \bibinfo{author}{Meno, C.}, \bibinfo{author}{Hamada, H.}, \bibinfo{year}{2006}.
\newblock \bibinfo{title}{Generation of robust left-right asymmetry in the mouse embryo requires a self-enhancement and lateral-inhibition system}.
\newblock \bibinfo{journal}{Developmental Cell} \bibinfo{volume}{11}, \bibinfo{pages}{495--504}.
%Type = Book
\bibitem[{Nayfeh(1973)}]{nayfeh1973perturbation}
\bibinfo{author}{Nayfeh, A.H.}, \bibinfo{year}{1973}.
\newblock \bibinfo{title}{Perturbation methods}.
\newblock \bibinfo{publisher}{John Wiley \& Sons}.
%Type = Article
\bibitem[{Newell and Whitehead(1969)}]{newell1969finite}
\bibinfo{author}{Newell, A.C.}, \bibinfo{author}{Whitehead, J.A.}, \bibinfo{year}{1969}.
\newblock \bibinfo{title}{Finite bandwidth, finite amplitude convection}.
\newblock \bibinfo{journal}{Journal of Fluid Mechanics} \bibinfo{volume}{38}, \bibinfo{pages}{279--303}.
%Type = Article
\bibitem[{Paul et~al.(2024)Paul, Adetunji and Hong}]{Paul2024}
\bibinfo{author}{Paul, S.}, \bibinfo{author}{Adetunji, J.}, \bibinfo{author}{Hong, T.}, \bibinfo{year}{2024}.
\newblock \bibinfo{title}{Widespread biochemical reaction networks enable turing patterns without imposed feedback}.
\newblock \bibinfo{journal}{Nature Communications} \bibinfo{volume}{15}, \bibinfo{pages}{8380}.
%Type = Article
\bibitem[{Peng and Zhang(2016)}]{peng2016turing}
\bibinfo{author}{Peng, Y.}, \bibinfo{author}{Zhang, T.}, \bibinfo{year}{2016}.
\newblock \bibinfo{title}{Turing instability and pattern induced by cross-diffusion in a predator-prey system with allee effect}.
\newblock \bibinfo{journal}{Applied Mathematics and Computation} \bibinfo{volume}{275}, \bibinfo{pages}{1--12}.
%Type = Article
\bibitem[{Purwins et~al.(2010)Purwins, B{\"o}deker and Amiranashvili}]{purwins2010dissipative}
\bibinfo{author}{Purwins, H.G.}, \bibinfo{author}{B{\"o}deker, H.}, \bibinfo{author}{Amiranashvili, S.}, \bibinfo{year}{2010}.
\newblock \bibinfo{title}{Dissipative solitons}.
\newblock \bibinfo{journal}{Advances in Physics} \bibinfo{volume}{59}, \bibinfo{pages}{485--701}.
%Type = Article
\bibitem[{Robertson(2004)}]{robertson2004embryonic}
\bibinfo{author}{Robertson, E.J.}, \bibinfo{year}{2004}.
\newblock \bibinfo{title}{Embryonic patterning and the evolution of vertebrate body plans}.
\newblock \bibinfo{journal}{Nature Reviews Genetics} \bibinfo{volume}{5}, \bibinfo{pages}{823--833}.
%Type = Article
\bibitem[{Schier(2009)}]{schier2003nodal}
\bibinfo{author}{Schier, A.F.}, \bibinfo{year}{2009}.
\newblock \bibinfo{title}{Nodal morphogens}.
\newblock \bibinfo{journal}{Cold Spring Harbor Perspectives in Biology} \bibinfo{volume}{1}, \bibinfo{pages}{a003459}.
%Type = Article
\bibitem[{Sekine et~al.(2018)Sekine, Shibata and Ebisuya}]{sekine2018synthetic}
\bibinfo{author}{Sekine, R.}, \bibinfo{author}{Shibata, T.}, \bibinfo{author}{Ebisuya, M.}, \bibinfo{year}{2018}.
\newblock \bibinfo{title}{Synthetic mammalian pattern formation driven by differential diffusivity of nodal and lefty}.
\newblock \bibinfo{journal}{Nature communications} \bibinfo{volume}{9}, \bibinfo{pages}{5456}.
%Type = Book
\bibitem[{Smoller(2012)}]{smoller2012shock}
\bibinfo{author}{Smoller, J.}, \bibinfo{year}{2012}.
\newblock \bibinfo{title}{Shock waves and reaction—diffusion equations}. volume \bibinfo{volume}{258}.
\newblock \bibinfo{publisher}{Springer Science \& Business Media}.
%Type = Article
\bibitem[{Stephenson and Wollkind(1995)}]{stephenson1995weakly}
\bibinfo{author}{Stephenson, L.E.}, \bibinfo{author}{Wollkind, D.J.}, \bibinfo{year}{1995}.
\newblock \bibinfo{title}{Weakly nonlinear stability analyses of one-dimensional turing pattern formation in activator-inhibitor/immobilizer model systems}.
\newblock \bibinfo{journal}{Journal of Mathematical Biology} \bibinfo{volume}{33}, \bibinfo{pages}{771--815}.
%Type = Article
\bibitem[{Stuart(1960)}]{stuart1960nonlinear}
\bibinfo{author}{Stuart, J.}, \bibinfo{year}{1960}.
\newblock \bibinfo{title}{On the non-linear mechanics of wave disturbances in unstable laminar flow}.
\newblock \bibinfo{journal}{Journal of Fluid Mechanics} \bibinfo{volume}{9}, \bibinfo{pages}{353--370}.
%Type = Article
\bibitem[{Turing(1952)}]{turing1990chemical}
\bibinfo{author}{Turing, A.M.}, \bibinfo{year}{1952}.
\newblock \bibinfo{title}{The chemical basis of morphogenesis}.
\newblock \bibinfo{journal}{Philosophical Transactions of the Royal Academy of London} \bibinfo{volume}{237}, \bibinfo{pages}{37--72}.
%Type = Article
\bibitem[{Ulloa and Briscoe(2007)}]{Ulloa2007}
\bibinfo{author}{Ulloa, F.}, \bibinfo{author}{Briscoe, J.}, \bibinfo{year}{2007}.
\newblock \bibinfo{title}{Morphogens and the control of cell proliferation and patterning in the spinal cord}.
\newblock \bibinfo{journal}{Cell Cycle} \bibinfo{volume}{6}, \bibinfo{pages}{2640--2649}.
%Type = Article
\bibitem[{Wollkind et~al.(1994)Wollkind, Manoranjan and Zhang}]{wollkind1994weakly}
\bibinfo{author}{Wollkind, D.J.}, \bibinfo{author}{Manoranjan, V.S.}, \bibinfo{author}{Zhang, L.}, \bibinfo{year}{1994}.
\newblock \bibinfo{title}{Weakly nonlinear stability analyses of prototype reaction-diffusion model equations}.
\newblock \bibinfo{journal}{Siam Review} \bibinfo{volume}{36}, \bibinfo{pages}{176--214}.
%Type = Article
\bibitem[{Wolpert(1969)}]{wolpert1969positional}
\bibinfo{author}{Wolpert, L.}, \bibinfo{year}{1969}.
\newblock \bibinfo{title}{Positional information and the spatial pattern of cellular differentiation}.
\newblock \bibinfo{journal}{Journal of theoretical biology} \bibinfo{volume}{25}, \bibinfo{pages}{1--47}.

\end{thebibliography}


%\begin{thebibliography}{00}


%\bibitem[Lamport(1994)]{lamport94}
  %Leslie Lamport,
  %\textit{\LaTeX: a document preparation system},
  %Addison Wesley, Massachusetts,
  %2nd edition,
  %1994.

%\end{thebibliography}

\end{document}